%% file: Max-SRTI-arxiv.tex
\documentclass[a4paper,UKenglish,cleveref,autoref]{lipics-v2019}

\hideLIPIcs

\usepackage{MnSymbol}
\crefname{paragraph}{Paragraph}{Paragraphs}

\usepackage[utf8]{inputenc}
\usepackage{enumitem}
\usepackage{tikz}
\usepackage{bm}
\usetikzlibrary{calc,backgrounds,fit,shapes.geometric}

\definecolor{mylila}{RGB}{190,186,218}
\definecolor{myyellow}{RGB}{255,255,179}
\definecolor{mygreen}{RGB}{141,211,199}

\tikzstyle{vertex}=[draw, circle, fill, inner sep = 2pt]

\tikzstyle{squaredvertex}=[draw, fill, red]

\input{fancyProblem}

\DeclareMathOperator{\rk}{rk}
\DeclareMathOperator{\cut}{cut}

\DeclareMathOperator{\tw}{tw}
\DeclareMathOperator{\adh}{adh}
\DeclareMathOperator{\tcw}{tcw}
\DeclareMathOperator{\td}{td}
\newcommand{\N}{N}

\DeclareMathOperator{\fes}{fes}
\DeclareMathOperator{\tor}{tor}
\DeclareMathOperator{\clos}{clos}
\DeclareMathOperator{\sig}{sig}

\DeclareMathOperator{\best}{best}
\DeclareMathOperator{\vc}{vc}
\DeclareMathOperator{\OPT}{OPT}

\newcommand{\oldminusone}{{\uparrow}}
\newcommand{\oldzero}{{\parallel}}
\newcommand{\oldone}{{\downarrow}}
\newcommand{\false}{\text{false}}
\newcommand{\true}{\text{true}}

\usepackage{xspace}
\usepackage[textwidth=20mm,obeyFinal,colorinlistoftodos]{todonotes}
\newcommand{\mytodo}[2]{\todo[size=\tiny, color=#1!50!white]{#2}\xspace}
\newcommand{\myinlinetodo}[2]{\todo[size=\small, color=#1!50!white, inline, caption={}]{#2}\xspace}
\newcommand{\registerAuthor}[3]{%
  \expandafter\newcommand\csname #2com\endcsname[1]{\mytodo{#3}{\textsc{#2}: ##1}}%
  \expandafter\newcommand\csname #2inline\endcsname[1]{\myinlinetodo{#3}{\textsc{#2}: ##1}}%
  \expandafter\newcommand\csname #2inlineLater\endcsname[1]{\lv{\myinlinetodo{#3}{\textsc{#2}: ##1}}}%
}

\registerAuthor{Dušan Knop}{dk}{purple}
\registerAuthor{Klaus Heeger}{kh}{green}

\usepackage{xspace}

\theoremstyle{plain}
\newtheorem{observation}[theorem]{Observation}

\usepackage{comment}

\title{
Parameterized Complexity of Stable Roommates with Ties and Incomplete Lists
Through the Lens of Graph Parameters
\thanks{An extended abstract of this work appears in the proceedings of the 30th International Symposium on Algorithms and Complexity (ISAAC 2019)~\cite{BredereckHKN19}.
In this full version, we provide many missing proofs and new results on the W[1]-hardness of \textsc{Stable Roommates with Ties and Incomplete Preferences} parameterized by feedback vertex number, strengthen the W[1]-hardness results for feedback vertex number to disjoint paths modulator number, and present an FPT-algorithm for \textsc{Maximum Stable Marriage with Ties and Incomplete Preferences} parameterized by tree-cut width.}}

\title{
Parameterized Complexity of Stable Roommates with Ties and Incomplete Lists
Through the Lens of Graph Parameters}

\titlerunning{Parameterized Complexity of Stable Roommates with Ties and Incomplete Lists}

\author{Robert Bredereck}{Humboldt Universit\"at zu Berlin, Institut f\"ur Informatik, Algorithm Engineering}{robert.bredereck@hu-berlin.de}{https://orcid.org/0000-0002-6303-6276}{}

\author{Klaus Heeger}{Technische Universit\"at Berlin, Algorithmics and Computational Complexity}{heeger@tu-berlin.de}{https://orcid.org/0000-0001-8779-0890}{Supported by DFG Research Training Group 2434 ``Facets of Complexity''.}

\author{Dušan Knop}{Department of Theoretical Computer Science, Faculty of Information Technology,\\ Czech Technical University in Prague, Prague, Czech Republic}{dusan.knop@fit.cvut.cz}{https://orcid.org/0000-0003-2588-5709}{Supported by the DFG, project MaMu (NI 369/19).}

\author{Rolf Niedermeier}{Technische Universit\"at Berlin, Algorithmics and Computational Complexity}{rolf.niedermeier@tu-berlin.de}{https://orcid.org/0000-0003-1703-1236}{}

\authorrunning{R. Bredereck, K. Heeger, D. Knop, and R. Niedermeier}

\Copyright{Robert Bredereck, Klaus Heeger, Dušan Knop, and Rolf Niedermeier}

\ccsdesc[500]{Theory of computation~Parameterized complexity and exact algorithms}
\ccsdesc[300]{Theory of computation~Graph algorithms analysis}
\ccsdesc[100]{Mathematics of computing~Matchings and factors}

\keywords{Stable matching, acceptability graph, fixed-parameter tractability, W[1]-hardness, graph parameters, treewidth, treedepth, tree-cut width, disjoint path modulator, feedback edge number}

\category{}

\relatedversion{An extended abstract of this work appears in the proceedings of the 30th International Symposium on Algorithms and Complexity (ISAAC 2019)~\cite{BredereckHKN19}.
In this full version, we provide many missing proofs and new results on the W[1]-hardness of \textsc{Stable Roommates with Ties and Incomplete Preferences} parameterized by feedback vertex number, strengthen the W[1]-hardness results for feedback vertex number to disjoint paths modulator number, and present an FPT-algorithm for \textsc{Maximum Stable Marriage with Ties and Incomplete Preferences} parameterized by tree-cut width.}

\supplement{}

\funding{Main work was done while all authors were with Technische Universit\"at Berlin.}

\acknowledgements{}

\nolinenumbers 

\EventEditors{John Q. Open and Joan R. Access}
\EventNoEds{2}
\EventLongTitle{42nd Conference on Very Important Topics (CVIT 2019)}
\EventShortTitle{CVIT 2019}
\EventAcronym{CVIT}
\EventYear{2019}
\EventDate{December 24--27, 2019}
\EventLocation{Little Whinging, United Kingdom}
\EventLogo{}
\SeriesVolume{42}
\ArticleNo{23}

\begin{document}

\maketitle
\begin{abstract}
We continue and extend previous work on the parameterized complexity analysis
of the NP-hard \textsc{Stable Roommates with Ties and Incomplete Lists}
problem, thereby strengthening earlier results
both on the side of parameterized hardness as well as on the side
of fixed-parameter
tractability.
Other than for its famous sister problem \textsc{Stable Marriage}
which focuses on a bipartite scenario,
\textsc{Stable Roommates with Incomplete Lists} allows for arbitrary
acceptability graphs whose edges specify the possible matchings
of each two agents (agents are represented by graph vertices).
Herein, incomplete lists and ties reflect the fact that
in realistic application
scenarios the agents cannot bring \emph{all} other agents into
a \emph{linear} order.
Among our main contributions is to show that it is W[1]-hard to compute a maximum-cardinality
stable matching for acceptability graphs of bounded treedepth, bounded tree-cut width, and bounded disjoint paths modulator number (these are each time the respective parameters).
However, if we `only' ask for perfect stable
matchings or the mere existence of a stable matching, then we obtain
fixed-parameter tractability with respect to tree-cut width but not with
respect to treedepth.
On the positive side, we also provide fixed-parameter tractability results for the parameter feedback edge set number.
\end{abstract}


\section{Introduction}
Computing stable matchings is a core topic in the
intersection of algorithm design and theory,
algorithmic game theory, and computational
social choice. It has numerous applications---the research goes back to the
1960s.
The classic (and most prominent from introductory textbooks) problem, \textsc{Stable Marriage}, is known to be solvable
in linear time.
It relies on complete bipartite graphs for the modeling
with the
two sides representing the same number of ``men'' and ``women''. Herein,
each side expresses preferences (linear orderings aka rankings)
over the opposite sex.
Informally, stability then means that no two matched
agents have reason to break up.
\textsc{Stable Roommates}, however, is not restricted
to a bipartite setting.
The input consists of a set~$V$ of agents together with a
preference list~$\mathcal{P}_v$ for every agent~$v \in V$,
where a preference list $\mathcal{P}_v$ is a strict (linear)
order on~$V \setminus \{v\}$.
The task is to find a \emph{stable matching}, that is, a set of pairs of agents
such that each agent is contained in at most one pair and there is no blocking edge
(i.e., a pair of agents who strictly prefer their mates in this pair to their partners in the matching; naturally, we assume that the agents prefer
to be matched to being unmatched).
Such a matching can be computed in polynomial time~\cite{Irving85}.
We refer to the monographs~\cite{GusfieldIrving89,Manlove13}
for a general discussion on \textsc{Stable Roommates}.
Recent practical applications of \textsc{Stable Roommates} and
its variations also to be studied here range from kidney exchange to connections in peer-to-peer networks~\cite{GaiLMMRV07,RothSU05,RothSU07}.

If the preference lists $\mathcal{P}_v$ for all agents~$v$ are complete, then
the graph-theoretic model behind is trivial---a complete graph reflects
that every agent ranks all other agents. In the more realistic scenario
that an agent may only rank part of all other agents,
the corresponding graph, referred to as \emph{acceptability graph},
is no longer a complete graph but can have an arbitrary structure.
We assume that the acceptability relation is symmetric, that is, if an agent~$v$ finds an agent~$u$ acceptable, then also agent~$u$ finds $v$~acceptable.
Moreover, to make the modeling of real-world scenarios more flexible and
realistic, one also allows ties in the preference lists (rankings) of the
agents, meaning that tied agents are considered equally good.
Unfortunately, once allowing ties in the preferences,
\textsc{Stable Roommates} already becomes NP-hard~\cite{ManloveIIMM02,Ronn90};
indeed, this is true even if each agent finds at most three other agents
acceptable~\cite{CsehIM19}.
Hence, in recent works specific (parameterized) complexity aspects of
\textsc{Stable Roommates with Ties and Incomplete Lists} (\textsc{SRTI})
have been investigated~\cite{Adil2018,DBLP:journals/aamas/BredereckCFN20,Chen-ICALP2018}.
In particular, while Bredereck et al.~\cite{DBLP:journals/aamas/BredereckCFN20} studied
restrictions on the structure of the preference lists,
Adil et al.~\cite{Adil2018} initiated the study of
structural restrictions of the underlying acceptability graph, including the
parameter treewidth of the acceptability graph.
We continue Adil et al.'s line of research
by systematically studying three variants (`maximum', `perfect', `existence') and by
extending significantly the range of graph parameters under study,
thus gaining a fairly comprehensive picture of the
parameterized complexity landscape of \textsc{SRTI}.

Notably, while previous work~\cite{Adil2018,GuptaSZ17} argued for the (also practical) relevance for
studying the structural parameters treewidth and vertex cover number,
our work extends this to further structural graph parameters that are either
stronger than the vertex cover number or yield more positive algorithmic results than possible for treewidth.
We study the arguably most natural optimization
version of \textsc{Stable Roommates} with ties and incomplete lists,
referred to as \textsc{Max-SRTI}:
\defProblemTask{\textsc{Max-SRTI}}
{A set~$V$ of agents and a preference profile~$\mathcal{P} = (\mathcal{P}_v)_{v \in V}$.}
{Find a maximum-cardinality stable matching or decide that none exists.}

In addition to \textsc{Max-SRTI}, we also study two in generally NP-hard variants.
The input is the same, but the task either changes to
finding a \emph{perfect} stable matching---this variant is \textsc{Perfect-SRTI}---or
to finding just \emph{any} stable matching---this variant is \textsc{SRTI-Existence}.\footnote{In the following, we consider a slightly different formulation of these problems: We assume that the input consists of the acceptability graph and rank functions. This is no restriction, as one can transform a set of agents and a profile to an acceptability graph and rank functions and vice versa in linear time.}

\defProblemTask{\textsc{Perfect-SRTI}}
{A set of agents $V$ and a preference profile $\mathcal{P} = (\mathcal{P}_v)_{v \in V}$.}
{Find a perfect stable matching or decide that none exists.}

\defProblemTask{\textsc{SRTI-Existence}}
{A set of agents $V$ and a preference profile $\mathcal{P} = (\mathcal{P}_v)_{v \in V}$.}
{Find a stable matching or decide that none exists.}

Clearly, any algorithm for \textsc{Max-SRTI} also solves \textsc{Perfect-SRTI} and \textsc{SRTI-Existence}.
An algorithm solving \textsc{Perfect-SRTI} can be used to solve \textsc{SRTI-Existence} by extending the preferences of every agent to include all other agents in a suitable way, ensuring that every stable matching needs to be perfect; however, this reduction turns sparse graphs into dense ones and therefore increases the parameters considered in this paper.
Note that on bipartite graphs, the Gale-Shapley algorithm~\cite{GaleShapley1962} always finds a stable matching, while it is NP-complete to find a perfect stable matching~\cite{Iwama1999,ManloveIIMM02}.

We denote the restriction of \textsc{Max-SRTI} and \textsc{Perfect-SRTI} to bipartite acceptability graphs by \textsc{Max-SMTI} and \textsc{Perfect-SMTI}.
Note that we do not consider the restriction of \textsc{SRTI-Existence} to bipartite acceptability graphs as the Gale-Shapley algorithm~\cite{GaleShapley1962} is guaranteed to find a solution in this case.

\subparagraph{Related Work.}
Restricted to bipartite acceptability graphs, where \textsc{Stable Roomates} is called \textsc{Stable Marriage},
\textsc{Max-SRTI} admits a polynomial-time factor-$\frac{2}{3}$-approximation \cite{McDermid2009}. However, even on bipartite graphs it is NP-hard to approximate \textsc{Max-SRTI} by a factor of $\frac{29}{33}$, and \textsc{Max-SRTI} cannot be approximated by a factor of $\frac{3}{4}+\epsilon$ for any $\epsilon >0$ unless \textsc{Vertex Cover} can be approximated by a factor strictly smaller than two~\cite{Yanagisawa2007}.
Note that, as we will show in our work, \textsc{SRTI-Existence}
is computationally hard in many cases, so good polynomial-time
or even fixed-parameter approximation algorithms
for \textsc{Max-SRTI} seem out of reach.

\textsc{Perfect-SRTI} was shown to be NP-hard even on bipartite graphs~\cite{Iwama1999,ManloveIIMM02}. This holds also for the more restrictive case when ties occur only on one side of the bipartition, and any preference list is either strictly ordered or a tie of length two~\cite{ManloveIIMM02}.
Furthermore, \textsc{Perfect-SMTI} is NP-hard when the acceptability graph is a complete bipartite graph minus one edge~\cite{CsehH20}.
As \textsc{SRTI-Existence} is NP-hard for complete graphs, all three problems considered in this paper are NP-hard on complete graphs (as every stable matching is a maximal matching). This implies paraNP-hardness for all parameters which are constant on cliques, including distance to clique, cliquewidth, neighborhood diversity, twin cover, the number of unmatched vertices, and modular width.

Following up on work by Bartholdi~III and Trick~\cite{BT86},
Bredereck et al.~\cite{DBLP:journals/aamas/BredereckCFN20} showed NP-hardness and  polynomial-time
solvability results for \textsc{SRTI-Existence} under several restrictions
constraining the agents' preference lists.

On a fairly general level, there is quite some work on employing methods
of parameterized algorithmics in the search for potential
islands of tractability for in general NP-hard stable
matching problems~\cite{Adil2018,DBLP:conf/sagt/BoehmerBHN20,DBLP:conf/wine/BoehmerH20,DBLP:conf/wine/BredereckHKN20,Chen-ICALP2018,Chen-EC2018,DBLP:conf/fsttcs/Gupta0R0Z20,DBLP:conf/wads/GuptaR0Z19,GuptaSZ17,Marx2010,Marx2011,MeeksR18,DBLP:journals/algorithmica/MnichS20}.
More specifically, Marx and Schlotter~\cite{Marx2010} showed that \textsc{Max-SMTI} is W[1]-hard when parameterized by the number of ties.
While it is known the problem is NP-hard even if the maximum length of a tie
is two~\cite{ManloveIIMM02}, they showed that \textsc{Max-SMTI} is fixed-parameter tractable when parameterized by the combined parameter `number of ties and maximum length of a tie'.
Meeks and Rastegari~\cite{MeeksR18} considered a setting where the agents are partitioned into different types having the same preferences.
They show that the problem is FPT in the number of types.
Mnich and Schlotter~\cite{DBLP:journals/algorithmica/MnichS20} defined \textsc{Stable Marriage with Covering Constraints}, where the task is to find a matching which matches a given set of agents, and minimizes the number of blocking pairs among all these matchings.
They showed the NP-hardness of this problem and investigated several parameters such as the number of blocking pairs or the maximum degree of the acceptability graph.

Most directly related to our work, however, Adil et al.~\cite{Adil2018}
started the research on structural restrictions of the acceptability graph,
which we continue and extend.  Their result is
an XP-algorithm for the parameter treewidth; indeed, they did not show
W[1]-hardness for this parameter, leaving this as an open question.
This open question was solved (also for the bipartite case)
by Gupta et al.~\cite{GuptaSZ17}, who further considered various variants
(such as sex-equal or balanced) of stable marriage with respect to treewidth variants.\footnote{Indeed, without knowing the work of Gupta et al.~\cite{GuptaSZ17}, our work initially was strongly motivated by Adil et al.'s~\cite{Adil2018}
open question for treewidth.
To our surprise, although the Adil et al.~\cite{Adil2018} paper has been revised six months after the publication of
Gupta et al.~\cite{GuptaSZ17}, it was not mentioned by Adil et al.~\cite{Adil2018} that this open question was answered
by a subset of the authors, namely Gupta et al.~\cite{GuptaSZ17}.}
Moreover, Adil et al.~\cite{Adil2018} showed that \textsc{Max-SRTI}
is fixed-parameter tractable when parameterized by the size of the solution
(that is, the cardinality of the set of edges in the stable matching)
and that \textsc{Max-SRTI} restricted to planar acceptability
graphs for the same parameter is fixed-parameter tractable, even obtaining subexponential running time.\footnote{More precisely, Adil et al.~\cite{Adil2018} state their result for the parameter `size of a maximum matching of the acceptability graph', which is only by a factor at most two greater than the size of a stable matching.}

\subparagraph{Our Contributions}
We continue the study of algorithms for \textsc{Max-SRTI} and its variants based on
structural limitations of the acceptability graph.
In particular, we extend the results of Adil et al.~\cite{Adil2018} in several ways.
For an overview on our results we refer to \Cref{fig:SRTIResultsOverview}. We highlight a few results in what follows.
We observe that Adil et al.'s dynamic programming-based XP-algorithm designed for the
parameter treewidth\footnote{It only gives containment in~XP for this parameter, and only this is stated by Adil et al.~\cite{Adil2018}.} indeed yields fixed-parameter tractability for the combined parameter treewidth and maximum degree.
We complement their XP result and the above mentioned results by showing that \textsc{Max-SRTI} is W[1]-hard for the graph parameters treedepth, disjoint paths modulator set (for these two parameters also on bipartite graphs), and tree-cut width.
Notably, all these graph parameters are `weaker'~\cite{KN12} than treewidth
and these mutually independent results imply W[1]-hardness with respect to
treewidth; the latter was also shown in the independent work of Gupta et al.~\cite{GuptaSZ17}.

For the two related problems \textsc{Perfect-SRTI} and \textsc{SRTI-Existence},
on the contrary, we show fixed-parameter tractability with respect to the parameter tree-cut width.
Furthermore,\textsc{Max-SMTI} is fixed-parameter tractable with respect to the parameter tree-cut width.
These results confirm the intuition that tree-cut width, a recently introduced~\cite{Wollan15} and since then already
well researched graph parameter~\cite{Ganian2015,DBLP:journals/algorithmica/GanianKO21,Ganian-ALENEX2019,KimOPST18,MarxW14}
`lying between' treewidth and the combined parameter `treewidth and maximum vertex degree', is a  better suited structural parameter
for edge-oriented graph problems than treewidth is.
Moreover, we extend our W[1]-hardness results to \textsc{Perfect-SRTI} and \textsc{SRTI-Existence} parameterized by treedepth and disjoint paths modulator number.
Finally, we give an FPT-algorithm for \textsc{Max-SRTI} parameterized by feedback edge number.
    \footnote{Note that the
preliminary version of this paper~\cite{BredereckHKN19} contained a
typo, claiming erroneously a running time of $2^{\fes (G)} n^{O(1)}$
instead of $3^{\fes (G)} n^{O(1)}$.}

In summary, we provide a complete picture of the (graph-)parameterized computational complexity landscape for the three studied problems---see \Cref{fig:SRTIResultsOverview} for an overview of our results.
Among other things, for the
parameter tree-cut width \Cref{fig:SRTIResultsOverview} depicts a
surprising complexity gap between \textsc{Max-SRTI} on the
one side (W[1]-hardness) and \textsc{Perfect-SRTI} and \textsc{SRTI-Existence}
(fixed-parameter tractability) on the other side.

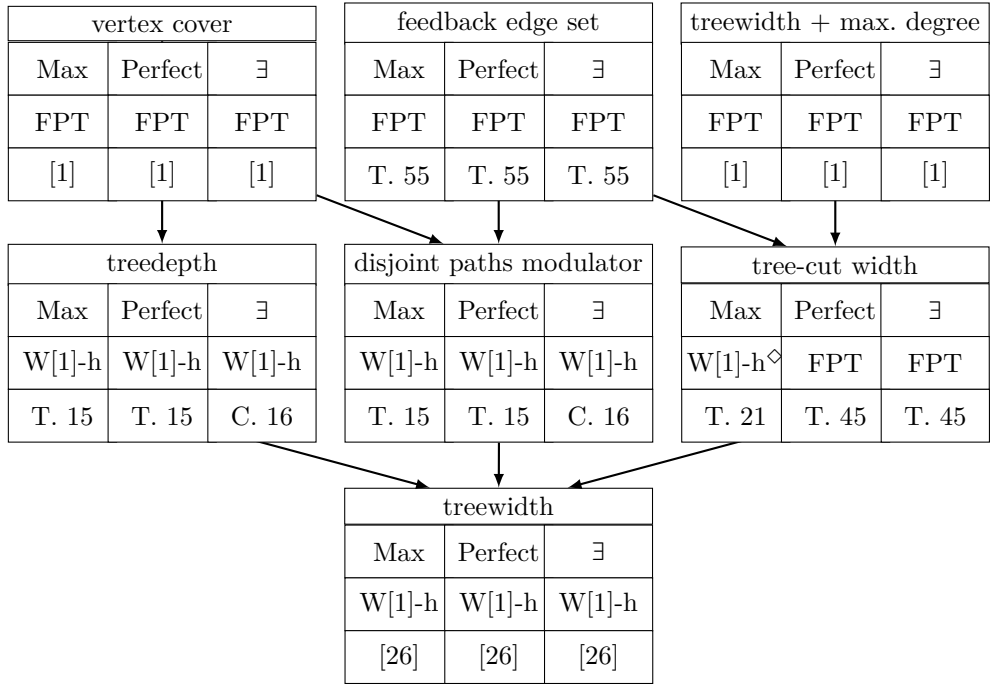
\begin{figure}[bt]
  \begin{center}
    \begin{tikzpicture}[xscale = 1.64]
      \tikzstyle{parameter}=[draw, minimum width = 4.05 cm, text centered, minimum height = 0.6]
      \tikzstyle{arrow} = [thick,->,>=latex]
      \node[parameter] (tw) at (5, -1.2) {treewidth};
      \node[parameter] (tcw) at (7.7, 2) {tree-cut width};
      \node[parameter] (twd) at (7.7, 5.2) {treewidth + max.\ degree};
      \node[parameter] (vc) at (2.3, 5.2) {vertex cover};
      \node[parameter] (td) at (2.3, 2) {treedepth};
      \node[parameter] (fvs) at (5, 2) {disjoint paths modulator};
      \node[parameter] (fes) at (5, 5.2) {feedback edge set};

      \tikzstyle{problemvariant}=[draw, minimum width = 1.424cm, minimum height = 0.7 cm, text centered, fill = white, inner sep = 0]
      \tikzstyle{problemvariantshort}=[draw, minimum width = 1.424cm, minimum height = 0.7 cm, text centered, fill = white, inner sep = 0]

      \node[problemvariant] (twm) at ($(tw) + (-0.8, -0.6)$) {Max};
      \node[problemvariant] (twp) at ($(tw) + (-0., -0.6)$) {Perfect};
      \node[problemvariant] (twe) at ($(tw) + (0.8, -0.6)$) {$\exists$};

      \node[problemvariant] (tdm) at ($(td) + (-0.8, -0.6)$) {Max};
      \node[problemvariant] (tdp) at ($(td) + (-0., -0.6)$) {Perfect};
      \node[problemvariant] (tde) at ($(td) + (0.8, -0.6)$) {$\exists$};

      \node[problemvariant] (vcm) at ($(vc) + (-0.8, -0.6)$) {Max};
      \node[problemvariant] (vcp) at ($(vc) + (-0., -0.6)$) {Perfect};
      \node[problemvariant] (vce) at ($(vc) + (0.8, -0.6)$) {$\exists$};

      \node[problemvariant] (twdm) at ($(twd) + (-0.8, -0.6)$) {Max};
      \node[problemvariant] (twdp) at ($(twd) + (0., -0.6)$) {Perfect};
      \node[problemvariant] (twde) at ($(twd) + (0.8, -0.6)$) {$\exists$};

      \node[problemvariant] (tcwm) at ($(tcw) + (-0.8, -0.6)$) {Max};
      \node[problemvariant] (tcwp) at ($(tcw) + (-0., -0.6)$) {Perfect};
      \node[problemvariant] (tcwe) at ($(tcw) + (0.8, -0.6)$) {$\exists$};

      \node[problemvariant] (fesm) at ($(fes) + (-0.8, -0.6)$) {Max};
      \node[problemvariant] (fesp) at ($(fes) + (-0, -0.6)$) {Perfect};
      \node[problemvariant] (fese) at ($(fes) + (0.8, -0.6)$) {$\exists$};

      \node[problemvariant] (fvsm) at ($(fvs) + (-0.8, -0.6)$) {Max};
      \node[problemvariant] (fvsp) at ($(fvs) + (-0., -0.6)$) {Perfect};
      \node[problemvariant] (fvse) at ($(fvs) + (0.8, -0.6)$) {$\exists$};

      \node[problemvariant] (tw1) at ($(twm) + (0, -0.7)$) {\text{W[1]-h}};
      \node[problemvariant] (tw2) at ($(twp) + (0, -0.7)$) {\text{W[1]-h}};
      \node[problemvariant] (tw3) at ($(twe) + (0, -0.7)$) {\text{W[1]-h}};
      
      \node[problemvariant] (tw11) at ($(tw1) + (0, -0.7)$) {\cite{GuptaSZ17}};
      \node[problemvariant] (tw22) at ($(tw2) + (0, -0.7)$) {\cite{GuptaSZ17}};
      \node[problemvariant] (tw33) at ($(tw3) + (0, -0.7)$) {\cite{GuptaSZ17}};

      \node[problemvariant] (td1) at ($(tdm) + (0, -0.7)$) {\text{W[1]-h}};
      \node[problemvariant] (td2) at ($(tdp) + (0, -0.7)$) {\text{W[1]-h}};
      \node[problemvariant] (td3) at ($(tde) + (0, -0.7)$) {\text{W[1]-h}};

      \node[problemvariant] (tdr1) at ($(td1) + (0, -0.7)$) {\text{T.~\ref{thm:W-hard-td-fvn}}};
      \node[problemvariant] (tdr2) at ($(td2) + (0, -0.7)$) {\text{T.~\ref{thm:W-hard-td-fvn}}};
      \node[problemvariant] (tdr3) at ($(td3) + (0, -0.7)$) {\text{C.~\ref{cor:W-hard-existence}}};

      \node[problemvariant] (vc1) at ($(vcm) + (0, -0.7)$) {FPT};
      \node[problemvariant] (vc2) at ($(vcp) + (0, -0.7)$) {FPT};
      \node[problemvariant] (vc3) at ($(vce) + (0, -0.7)$) {FPT};

      \node[problemvariant] (vcr1) at ($(vc1) + (0, -0.7)$) {\cite{Adil2018}};
      \node[problemvariant] (vcr2) at ($(vc2) + (0, -0.7)$) {\cite{Adil2018}};
      \node[problemvariant] (vcr3) at ($(vc3) + (0, -0.7)$) {\cite{Adil2018}};

      \node[problemvariant] (twd1) at ($(twdm) + (0, -0.7)$) {FPT};
      \node[problemvariant] (twd2) at ($(twdp) + (0, -0.7)$) {FPT};
      \node[problemvariant] (twd3) at ($(twde) + (0, -0.7)$) {FPT};

      \node[problemvariant] (twd1) at ($(twd1) + (0, -0.7)$) {\cite{Adil2018}};
      \node[problemvariant] (twd2) at ($(twd2) + (0, -0.7)$) {\cite{Adil2018}};
      \node[problemvariant] (twd3) at ($(twd3) + (0, -0.7)$) {\cite{Adil2018}};

      \node[problemvariant] (tcw1) at ($(tcwm) + (0, -0.7)$) {\text{W[1]-h}$^\Diamond$};
      \node[problemvariant] (tcw2) at ($(tcwp) + (0, -0.7)$) {FPT};
      \node[problemvariant] (tcw3) at ($(tcwe) + (0, -0.7)$) {FPT};

      \node[problemvariant] (tcwr1) at ($(tcw1) + (0, -0.7)$) {\text{T.~\ref{twhtcw}}};
      \node[problemvariant] (tcwr2) at ($(tcw2) + (0, -0.7)$) {\text{T.~\ref{tfpttcw}}};
      \node[problemvariant] (tcwr3) at ($(tcw3) + (0, -0.7)$) {\text{T.~\ref{tfpttcw}}};

      \node[problemvariant] (fes1) at ($(fesm) + (0, -0.7)$) {FPT};
      \node[problemvariant] (fes2) at ($(fesp) + (0, -0.7)$) {FPT};
      \node[problemvariant] (fes3) at ($(fese) + (0, -0.7)$) {FPT};

      \node[problemvariant] (fesr1) at ($(fes1) + (0, -0.7)$) {\text{T.~\ref{tfptfes}}};
      \node[problemvariant] (fesr2) at ($(fes2) + (0, -0.7)$) {\text{T.~\ref{tfptfes}}};
      \node[problemvariant] (fesr3) at ($(fes3) + (0, -0.7)$) {\text{T.~\ref{tfptfes}}};

      \node[problemvariant] (fvs1) at ($(fvsm) + (0, -0.7)$) {\text{W[1]-h}};
      \node[problemvariant] (fvs2) at ($(fvsp) + (0, -0.7)$) {\text{W[1]-h}};
      \node[problemvariant] (fvs3) at ($(fvse) + (0, -0.7)$) {\text{W[1]-h}};

      \node[problemvariant] (fvsrm) at ($(fvs1) + (-0, -0.7)$) {T.~\ref{thm:W-hard-td-fvn}};
      \node[problemvariant] (fvsrp) at ($(fvs2) + (-0, -0.7)$) {T.~\ref{thm:W-hard-td-fvn}};
      \node[problemvariant] (fvsre) at ($(fvs3) + (-0, -0.7)$) {C.~\ref{cor:W-hard-existence}};

      \begin{scope}[on background layer]

        \draw[arrow] (tdr2) -- (tw);
        \draw[arrow] (vc) -- (td);
        \draw[arrow] (tcwr2) -- (tw);
        \draw[arrow] (twd) -- (tcw);
        \draw[arrow] (fvsrp) -- (tw);
        \draw[arrow] (fes) -- (fvs);
      \end{scope}

      \draw[arrow] (vcr3) -- (fvs);
      \draw[arrow] (fesr3) -- (tcw);
  \end{tikzpicture}

  \end{center}
  \caption{\label{fig:SRTIResultsOverview}%
  \looseness -1
  Results for graph-structural parameterizations of \textsc{Stable Roommates with Ties and Incomplete Lists}. Max means \textsc{Max-SRTI}, Perfect means \textsc{Perfect-SRTI}, and $\exists $ means \textsc{SRTI-Existence}.
  All W[1]-hardness results for \textsc{Perfect-SRTI} also hold for \textsc{Perfect-SMTI}, i.e., the restriction of \textsc{Perfect-SRTI} to bipartite acceptability graphs.
  The symbol $^\Diamond$ indicates the existence of an FPT factor-$\frac{1}{2}$-approximation algorithm as well as an exact FPT algorithm on bipartite graphs (see \Cref{capx}).
  The arrows indicate dependencies between the different parameters. An arrow from a parameter~$p_1(G)$ to a parameter~$p_2 (G)$ means that there is a computable function~$f\colon \mathbb{N}\rightarrow\mathbb{N}$ such that for any graph $G$ we have $p_1 (G) \ge f(p_2 (G))$.
  Consequently, W[1]-hardness for~$p_1$ then implies W[1]-hardness for~$p_2$, and fixed-parameter tractability for~$p_2$ then implies fixed-parameter tractability for~$p_1$.
  }
\end{figure}

\section{Preliminaries}\label{sec:preliminaries}
  For a positive integer $n$ let $[n]:= \{ 1, 2, 3, \dots, n\}$.
  We write vectors $\bm{h}$ in boldface, and given a vector $\bm{h} \in X^E $ for a set $X$ and a finite set $E$, we access its entries (coordinates) via $\bm{h} (e)$ for some $e \in E$.

  For an undirected graph $G$ and a vertex $v\in V(G)$, let $\delta_G (v)$ be the set of edges incident to $v$.
  For a set of vertices $X \subseteq V(G)$, let $\delta_G(X) $ be the set of edges with one endpoint contained in~$X$ and the other endpoint not contained in $X$.
  If the graph~$G$ is clear from the context, then we may just write $\delta (v)$ or $\delta (X)$.
  For a subset of edges $M\subseteq E(G)$ and a vertex~$v\in V(G)$, let $\delta_M (v) := \delta_G (v) \cap M$.
  We denote the maximum degree in $G$ by $\Delta (G)$, i.e., $\Delta (G):= \max_{v\in V(G)} |\delta_G (v)|$.
  For a tree $T$ rooted at a vertex~$r$ and a vertex $v\in V(T)$, we denote by $T_v$ the subtree rooted at $v$.
  For a graph~$G$ and a subset of vertices $X$ (a subset of edges $F$), we define $G-X$ ($G-F$) to be the graph arising from $G$ by deleting all vertices in $X$ and all edges incident to a vertex from $X$ (deleting all edges in~$F$).
  For a graph~$G$ and a set of vertices $X\subseteq V(G)$, the graph arising by \emph{contracting} $X$ is denoted by $G_{/X}$; it is defined by replacing the vertices in $X$ by a single vertex. Thus, we have $V(G_{/X}) := \bigl(V(G)\setminus X\bigr) \cup \{v_X\}$ and $E(G_{/X}):=\{\{v, w\}\in E(G): v, w\notin X\}\cup \{\{v, v_X\}:\{v, x\}\in E(G): v\notin X, x\in X\}$.
  Unless stated otherwise, $n := |V(G)|$ and $m := |E(G)|$.

  We define the directed graph $\overleftrightarrow{G}$ by replacing each edge $\{v, w\}\in E(G)$ by two directed ones in opposite directions, i.e., $(v, w) $ and $(w, v)$.

Note that the acceptability graph for a set of agents $V$ and a profile $\mathcal{P}$ is always simple, while a graph arising from a simple graph through the contraction of vertices does not need to be simple.

A parameterized problem consists of the problem instance $I$ (in our
setting the \textsc{Stable Roommate} instance) and a
parameter value~$k$ (in our case always a number measuring
some aspect in acceptability graph). An  FPT-algorithm for
a parameterized problem is an algorithm that runs in time
$f(k)|I|^{O(1)}$, where $f$ is some computable function. That is, an
FPT algorithm can run in exponential time, provided that the
exponential part of the running time depends on the parameter value
only. If such an algorithm exists, the parameterized problem is
called fixed-parameter tractable for the corresponding parameter.
There is also a theory of hardness of parameterized
problems that includes the notion of W[1]-hardness. If a problem is
W[1]-hard for a given parameter, then it is widely believed not to be
fixed-parameter tractable for the same parameter.

The typical approach to showing that a certain parameterized problem is
W[1]-hard is to reduce to it a known W[1]-hard problem, using
the notion of a parameterized reduction. In our case, instead of using
the full power of parameterized reductions, we use standard many-one
reductions that ensure that the value of the parameter in the output
instance is bounded by a function of the parameter of the input
instance.

The Exponential-Time Hypothesis (ETH) of Impagliazzo and Paturi~\cite{ImpagliazzoP01} asserts that there is a constant $c > 1$ such that there is no \( c^{o(n)} \) time algorithm solving the \textsc{Satisfiability} problem, where $n$ is the number of variables.
Chen et al.~\cite{CHEN2005216} showed that assuming ETH, there is no \( f(k) \cdot n^{o(k)} \) time algorithm solving $k$-\textsc{(Multicolored) Clique}, where $f$ is any computable function and $k$ is the size of the clique we are looking for.
For further notions related to parameterized complexity and ETH we refer the reader to~\cite{CyganFKLMPPS15}.

\subsection{Profiles and Preferences}
Let $V$ be a set of agents.
A \emph{preference list}~$\mathcal{P}_v$ for an agent $v$ is a subset $P_v \subseteq V\setminus \{v\}$ together with an ordered partition $(P_v^1,  \dots, P_v^k)$ of $P_v$.
A set $P_v^i$ with $|P_v^i|>1$ is called a \emph{tie}. The \emph{size of a tie} $P_v^i$ is its cardinality, i.e.,~$|P_v^i|$.
For an agent $v\in V$, the \emph{rank function} is $\rk_v \colon P_v\cup\{v\} \rightarrow \mathbb{N}\cup \{\infty\}$ with $\rk_v(x) := i$ for $x\in P_v^i$, and $\rk_v (v) = \infty$.

We say that $v$ \emph{prefers $x\in P_v$ to $y\in P_v$} if $\rk_v(x)< \rk_v (y)$.
In this case, we also write~$x \succ_v y$.
If $\rk_v(x)=\rk_v(y)$, then~$v$ \emph{ties} $x$ and $y$.
In this case, we also write~$x \sim_v y$.
For a set $V$ of agents, a set $\mathcal{P}= (\mathcal{P}_v)_{v \in V}$ of preference lists is called a \emph{profile}.
The corresponding \emph{acceptability graph} $G$ consists of vertex set $V(G):= V$ and edge set~$E(G):=\{\{v, w\}: v\in P_w \}$.
Recall that we assume that acceptability is symmetric, i.e., $v\in P_w $ if and only if $ w\in P_v$ for every two agents $v, w \in V$.

A subset $M\subseteq E(G)$ of pairwise non-intersecting edges is called a matching.
If~$\{x,y\} \in M$, then we denote the corresponding partner~$y$ of~$x$ by~$M(x)$ and set $M(x):=x$ if $x$~is unmatched, that is, if $\{y\in V(G) : \{x, y\} \in M\}=\emptyset$.
An edge~$\{v, w\}\in E(G)$ is \emph{blocking for~$M$} if $\rk_v (w) < \rk_v (M(v))$ and $\rk_w (v) < \rk_w (M(w))$; we say that $v,w$ constitutes a \emph{blocking pair for $M$}.
A matching $M\subseteq E(G)$ is \emph{stable} if there are no blocking pairs, i.e., for all $\{v, w\}\in E(G)$, we have $\rk_v(w) \ge \rk_v (M(v))$ or $\rk_w(v)\ge \rk_v(M(w))$.

Note that the literature contains several different stability notions for a matching in the presence of ties. Our stability definition is frequently called \emph{weak stability}.\footnote{Manlove \cite{Manlove02} discusses other types of stability---strong stability and super-strong stability.}

\subsection{Structural Graph Parameters}

We consider the (graph-theoretic) parameters treewidth, tree-cut width, treedepth, disjoint paths modulator number, feedback edge number, vertex cover number, and the combined parameter `treewidth + maximum vertex degree' (also called degree-treewidth in the literature).

  A set of edges $F \subseteq E(G)$ is a \emph{feedback edge set} if $G-F $ is a forest.
  We define the \emph{feedback edge number} $\text{fes} (G)$ to be the cardinality of a minimum feedback edge set of~$G$.
  The \emph{disjoint paths modulator number} $\operatorname{dpm} (G)$ is the minimum cardinality of a set~$X$ sucht that $G-X$ is the disjoint union of paths.
  A \emph{vertex cover} is a set of vertices intersecting with every edge of~$G$, and the \emph{vertex cover number} $\vc (G)$ is the size of a minimum vertex cover.
  The \emph{treedepth} $\td (G) $ is the smallest height of a rooted tree~$T$ with vertex set $V(G)$ such that for each $\{v, w\}\in V(G)$ we have that either $v$ is a descendant of $w$ in $T$ or $w$ is a descendant of $v$ in $T$.

\emph{Treewidth} intuitively measures the tree-likeness of a graph. It can be defined via structural decompositions of a graph into pieces of bounded size, which are connected in a tree-like fashion, called \emph{tree decompositions}.
As we do not need treewidth to describe our results, we do not give a formal definition here.
For a precise definition and some properties of treewidth, we refer to Kloks~\cite{Kloks94}.

\subparagraph{Tree-cut Width}
Tree-cut width has been introduced by Wollan~\cite{Wollan15} as tree-likeness measure
between treewidth and treewidth combined with maximum degree.
A family of subsets $X_1, \dots, X_k$ of a finite set $X$ is a \emph{near-partition} of $X$ if $X_i\cap X_j = \emptyset$ for all $i\neq j$ and $\bigcup_{i=1}^k X_i = X$. Note that $X_i =\emptyset$ is possible (even for several distinct~$i$).
A \emph{tree-cut decomposition} of a graph~$G$ is a pair $(T, \mathcal{X})$ which consists of a tree $T$ and a near-partition $\mathcal{X} = \{X_t\subseteq V(G): t\in V(T)\}$ of $V(G)$.
A set in the family $\mathcal{X}$ is called a \emph{bag} of the tree-cut decomposition.
An example for a tree-cut decompsition is given in \Cref{fig:tcw}.
 \begin{figure}[bt]
\begin{center}
    \begin{tikzpicture}[xscale=0.8, yscale=0.75]
       \node[vertex, label = 90:$r_{1}$] (r1) at (0.5, 0.5) {};
       \node[vertex, label = 90:$r_{2}$] (r2) at (1.5, 0.5) {};

       \node[vertex, label = 180:$v_{11}$] (m1) at (-1, -2.5) {};
       \node[vertex, label = 225:$v_{12}$] (m2) at (0.5, -2.5) {};
       \node[vertex, label = 180:$v_{13}$] (m3) at (1.5, -1) {};
       \node[vertex, label = 90:$v_{14}$] (m4) at (1.5, -2.5) {};
       \node[vertex, label = 00:$v_{15}$] (m5) at (2.5, -2.5) {};

       \node[vertex, label = 180:$v_{21}$] (ll1) at (-1, 0.5) {};
       \node[vertex, label = 180:$v_{22}$] (ll2) at (-1, -1) {};

       \node[vertex, label = 0:$v_{41}$] (lr1) at (3.7, -3) {};
       \node[vertex, label = 0:$v_{42}$] (lr2) at (3.7, -1) {};
       \node[vertex, label = 0:$v_{43}$] (lr3) at (3.7, 0.5) {};

       \draw (m1) edge (m2);
       \draw (m2) edge (m4);
       \draw (m4) edge (m5);
       \draw (m5) edge (m3);

       \node[vertex, label = 270:$v_{3}$] (lm) at (1, -4) {};
       \draw (lm) edge (m2);
       \draw (lm) edge (m4);
       \draw (lm) edge (m5);

       \draw (ll1) edge (ll2);
       \draw (ll2) edge (m1);
       \draw (ll1) edge (r1);

       \draw (lr2) edge (m5);
       \draw (lr1) edge  (lr2);
       \draw (lr2) edge (lr3);
       \draw (lr1) edge[bend right=55] (lr3);
       \draw (lr2) edge (m3);

       \draw (m3) edge (r2);
     \end{tikzpicture}

     \begin{tikzpicture}
       \node (shift) at (0, 0.35) {};
       \draw (0, 0) rectangle (3, 1);

       \draw ($(-1, -1.5) + (shift)$) rectangle ($(3.5, -3.5) + (shift)$);

       \draw ($(-5, -5) + 2*(shift)$) rectangle ($(-2, -6) + 2*(shift)$);
       \draw ($(-0.5, -5) + 2*(shift)$) rectangle ($(1.5, -6) + 2*(shift)$);
       \draw ($(3.5, -5) + 2*(shift)$) rectangle ($(7, -6) + 2*(shift)$);

       \draw[ultra thick] ($(1, 0)$) -- ($(1, -1.5) + (shift)$);
       \draw[ultra thick] ($(-1, -3.5) + (shift)$) -- ($(-2, -5) + 2*(shift)$);
       \draw[ultra thick] ($(1, -3.5) + (shift)$) -- ($(1, -5) + 2*(shift)$);
       \draw[ultra thick] ($(3.5, -3.5) + (shift)$) -- ($(4, -5) + 2*(shift)$);

       \node[vertex, label = 90:\Large{$r_{1}$}] (r1) at (0.5, 0.45) {};
       \node[vertex, label = 90:\Large{$r_{2}$}] (r2) at (1.5, 0.45) {};

       \node[vertex, label = 90:\Large{$v_{11}$}] (m1) at ($(-0.5, -2.1) + (shift)$) {};
       \node[vertex, label = 90:\Large{$v_{12}$}] (m2) at ($(0.5, -2.5) + (shift)$) {};
       \node[vertex, label = 180:\Large{$v_{13}$}] (m3) at ($(1.7, -2) + (shift)$) {};
       \node[vertex, label = 90:\Large{$v_{14}$}] (m4) at ($(1.5, -3) + (shift)$) {};
       \node[vertex, label = 00:\Large{$v_{15}$}] (m5) at ($(2.7, -2.7) + (shift)$) {};

       \draw (m1) edge node[pos=0.2, fill=white, inner sep=2pt] {\footnotesize $1$}  node[pos=0.76, fill=white, inner sep=2pt] {\footnotesize $1$} (m2);
       \draw (m2) edge node[pos=0.2, fill=white, inner sep=2pt] {\footnotesize $1$}  node[pos=0.76, fill=white, inner sep=2pt] {\footnotesize $3$} (m4);
       \draw (m4) edge node[pos=0.2, fill=white, inner sep=2pt] {\footnotesize $1$}  node[pos=0.76, fill=white, inner sep=2pt] {\footnotesize $1$} (m5);
       \draw (m5) edge node[pos=0.2, fill=white, inner sep=2pt] {\footnotesize $2$}  node[pos=0.76, fill=white, inner sep=2pt] {\footnotesize $2$} (m3);

       \node[vertex, label = 270:\Large{$v_{3}$}] (lm) at ($(1, -5.5) + 2*(shift)$) {};
       \draw (lm) edge node[pos=0.2, fill=white, inner sep=2pt] {\footnotesize $1$}  node[pos=0.76, fill=white, inner sep=2pt] {\footnotesize $2$} (m2);
       \draw (lm) edge node[pos=0.2, fill=white, inner sep=2pt] {\footnotesize $2$}  node[pos=0.76, fill=white, inner sep=2pt] {\footnotesize $2$} (m4);
       \draw (lm) edge node[pos=0.2, fill=white, inner sep=2pt] {\footnotesize $3$}  node[pos=0.76, fill=white, inner sep=2pt] {\footnotesize $2$} (m5);

       \node[vertex, label = 270:\Large{$v_{21}$}] (ll1) at ($(-3.9, -5.5) + 2*(shift)$) {};
       \node[vertex, label = 270:\Large{$v_{22}$}] (ll2) at ($(-2.5, -5.5) + 2*(shift)$) {};

       \draw (ll1) edge node[pos=0.2, fill=white, inner sep=2pt] {\footnotesize $1$}  node[pos=0.76, fill=white, inner sep=2pt] {\footnotesize $2$} (ll2);
       \draw (ll2) edge[bend left=15] node[pos=0.2, fill=white, inner sep=2pt] {\footnotesize $1$}  node[pos=0.8, fill=white, inner sep=2pt] {\footnotesize $2$} (m1);
       \draw (ll1) edge[bend left] node[pos=0.2, fill=white, inner sep=2pt] {\footnotesize $1$}  node[pos=0.76, fill=white, inner sep=2pt] {\footnotesize $1$} (r1);

       \node[vertex, label = 270:\Large{$v_{41}$}] (lr1) at ($(3.8, -5.5) + 2*(shift)$) {};
       \node[vertex, label = 270:\Large{$v_{42}$}] (lr2) at ($(4.8, -5.5) + 2*(shift)$) {};
       \node[vertex, label = 270:\Large{$v_{43}$}] (lr3) at ($(5.8, -5.5) + 2*(shift)$) {};

       \draw (lr2) edge[bend right = 25] node[pos=0.2, fill=white, inner sep=2pt] {\footnotesize $1$}  node[pos=0.76, fill=white, inner sep=3pt] {\footnotesize $2$} (m5);
       \draw (lr1) edge node[pos=0.2, fill=white, inner sep=2pt] {\footnotesize $1$}  node[pos=0.76, fill=white, inner sep=2pt] {\footnotesize $3$} (lr2);
       \draw (lr2) edge node[pos=0.2, fill=white, inner sep=2pt] {\footnotesize $2$}  node[pos=0.76, fill=white, inner sep=2pt] {\footnotesize $1$} (lr3);
        \draw (lr1) edge[bend left=35] node[pos=0.2, fill=white, inner sep=2pt] {\footnotesize $2$}  node[pos=0.76, fill=white, inner sep=2pt] {\footnotesize $1$} (lr3);
       \draw (lr2) edge[bend right = 60] node[pos=0.2, fill=white, inner sep=2pt] {\footnotesize $2$}  node[pos=0.76, fill=white, inner sep=2pt] {\footnotesize $1$} (m3);

        \draw (m3) edge node[pos=0.2, fill=white, inner sep=2pt] {\footnotesize $1$}  node[pos=0.76, fill=white, inner sep=2pt] {\footnotesize $1$} (r2);

       \node[] (lr) at (2.5, 0.5) {\LARGE{$r$}};
       \node (lm) at ($(-0.5, -3) + (shift)$) {\LARGE{$t_1$}};
       \node (lll) at ($(-4.5, -5.5) + 2*(shift)$) {\LARGE{$t_2$}};
       \node (llm) at ($(0, -5.5) + 2*(shift)$) {\LARGE{$t_3$}};
       \node (llr) at ($(6.7, -5.5) + 2*(shift)$) {\LARGE{$t_4$}};
     \end{tikzpicture}
\end{center}
   \caption{\label{fig:tcw}%
   An example of a graph $G$ (upper part) and its nice tree-cut decomposition~$(T, \mathcal{X})$ (not of minimal width) (lower part).
   The vertices of $G$ are the circles, while the nodes of $T$ are the rectangles.
   For a node $t\in V(T)$, bag $X_t$ contains exactly the vertices inside the rectangle.
   In the lower picture, the thick edges are the edges of $T$, while the thin edges are from $G$.
   The root of this tree-cut decomposition is $r$.
   We have $\adh (r) = 0$, $\adh (t_1) =1$, $\adh (t_2) = 2$, $\adh (t_3) = 3$, and $\adh (t_4) = 2$.
   Furthermore, we have $\tor (r) = 2$, $\tor (t_1) = 6$, $\tor (t_2) = 2$, $\tor (t_3) = 2$, and $\tor (t_4) = 3$.
   Therefore, the width of this tree-cut decomposition is $6$.
   The nodes~$t_1$ and $t_4$ are light, while $t_2$ (because there is an edge connecting a vertex in $t_2$ to a vertex in~$r$) and $t_3$ (because $\adh (t_3) =3$) are heavy.
   }
 \end{figure}
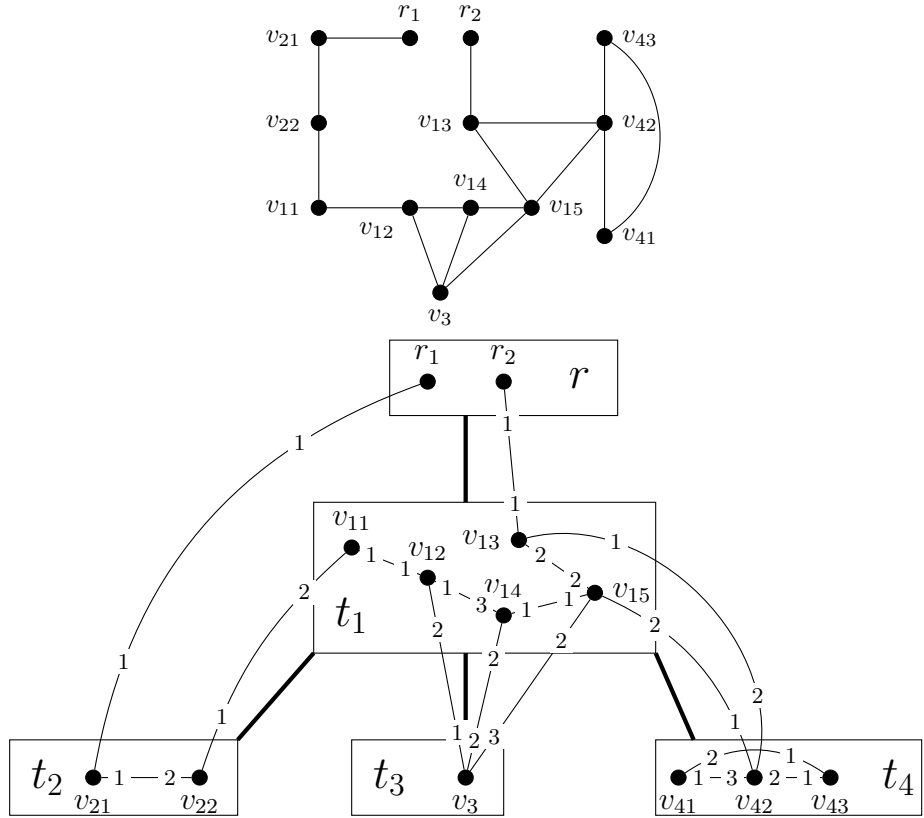
 
   Given a tree node~$t$, let $T_t$~be the subtree of~$T$ rooted at~$t$.
For a node~$t\in V(T)$, we denote by $Y_t$ the set of vertices induced by $T_t$, i.e., $Y_t:= \bigcup_{v\in V(T_t)} X_v$.
The graph induced by these vertices is denoted by $G_t:= G[Y_t]$.

  For an edge $e= \{u, v\}\in E(T)$, we denote by $T_u^{\{u, v\}}$ and $T_v^{\{u, v\}}$ the two connected components in $T-e$ which contain $u$ respectively $v$.
  These define a partition \[(\bigcup_{t\in T_u^{\{u, v\}}} X_t, \bigcup_{t\in T_v^{\{u, v\}}} X_t)\] of $V(G)$.
  We denote by $\cut (e)\subseteq E(G)$ the set of edges of $G$ with one endpoint in $\bigcup_{t\in T_u^{\{u, v\}}} X_t$ and the other one in $\bigcup_{t\in T_v^{\{u, v\}}} X_t$.

  A tree-cut decomposition is called \emph{rooted} if one of its nodes is called the root $r$. For any node $t\in V(T)\setminus\{r\}$, we denote by $e(t)$ the unique edge incident to $t$ on the $r$-$t$-path in $T$. The \emph{adhesion} $\adh_T (t)$ is defined as $|\cut (e(t))|$ for each $t\neq r$, and $\adh_T (r):= 0$.

  The \emph{torso of a tree-cut decomposition $(T, \mathcal{X})$ at a node $t$}, denoted by $H_t$, can be constructed from $G$ as follows:
  If $T$ consists of a single node, then the torso of $t\in V(T)$ is $G$.
  Else let $C_{t}^1, \dots, C_{t}^\ell$ be the connected components of $T-t$.
  Let $Z_i:=\bigcup_{v\in V(C_{t}^i)} X_v$.
  Then, the torso arises from $G$ by contracting each $Z_i\subseteq V(G)$ for $1\le i\le \ell$.

  The operation of \emph{suppressing a vertex} $v$ of degree at most two consists of deleting $v$ and, if $v$ has degree exactly two, then adding an edge between the two neighbors of $v$.
  The torso-size $\textnormal{tor} (t)$ is defined as the number of vertices of the graph arising from the torso $H_t$ by exhaustively suppressing all vertices of degree at most two.

  The \emph{width of a tree-cut decomposition} $(T, \mathcal{X})$ is defined as the maximum adhesion or torso-size of a node, i.e., $\max_{t\in V(T)} \{\adh (t), \textnormal{tor} (t)\}$.
  The \emph{tree-cut width} $\tcw (G)$ of a graph $G$ is the minimum width of a tree-cut decomposition of $G$.

  \subparagraph{Nice Tree-cut Decompositions}
  Similarly to nice tree decompositions~\cite{Kloks94}, each tree-cut decomposition can be transformed into a nice tree-cut decomposition.
  Nice tree-cut decompositions have additional properties which help simplifying algorithm design.
  Besides the definition of nice tree-cut decompositions, in the following we provide some of its properties.\footnote{The properties used here are stated (without a proof) by Ganian et al.~\cite{Ganian2015}; a proof is given by Ganian et al.~\cite{GanianKO21}.}
  
  \begin{definition}[\cite{Ganian2015}]
    Let $(T, \mathcal{X})$ be a tree-cut decomposition.
    A node $t\in V(T)$ is called \emph{light} if $\adh (t) \le 2$ and all outgoing edges from $Y_t$ end in $X_p$, where $p$ is the parent of $t$, and \emph{heavy} otherwise.
  \end{definition}

   \begin{theorem}[{\cite[Theorem~2]{Ganian2015}}]
     Let $G$ be a graph with $\tcw( G) = k$. Given a tree-cut decomposition of $G$ of width $k$, one can compute a nice tree-cut decomposition $(T, \mathcal{X})$ of $G$ of width $k$ with at most $2|V(G)|$ nodes in cubic time.
   \end{theorem}

  \begin{lemma}[{\cite[Lemma~2]{Ganian2015}}]\label{laatcw}
    Each node $t$ in a nice tree-cut decomposition of width $k$ has at most $2k+1$ heavy children.
  \end{lemma}

  \looseness -1
  In what follows, we will assume that a nice tree-cut decomposition of the input graph is given.
  Computing the tree-cut width of a graph is NP-hard, but there exists an algorithm that, given a graph $G$ and an integer $k$, either finds a tree-cut decomposition of width at most $2k$ or decides that $\tcw (G) > k$ in time $2^{O(k^2\log k)} n^2$ \cite{KimOPST18}.
  Furthermore, Giannopoulou et al.~\cite{GiannopoulouKRT19} gave a constructive proof of the existence of an algorithm deciding whether the tree-cut width of a given graph $G$ is at most $k$ in $f(k) n$ time, where $f$ is a computable recursive function.
  Very recently, Ganian et al.~\cite{Ganian-ALENEX2019} performed experiments on computing optimal tree-cut decompositions using SAT-solvers.

  To get an intuition about tree-cut width, we start with simple observations about the tree-cut width.

\begin{lemma}\label{ltcwt}
  Let $T$ be a forest. Then $\tcw (T) = 1$.
\end{lemma}

\begin{proof}
  As clearly $\tcw (T) \le \tcw (T + F)$ for any set of edges $F$, we may assume without loss of generality that $T$ is a tree.

  We define $X_t = \{t\}$ for all $t\in V(T)$, and consider the tree-cut decomposition $(T, \mathcal{X})$, and pick an arbitrary vertex $r$ to be the root of $T$.

  As $T$ is a tree, we have $\adh (t) = 1$ for all $t\neq r$.

  Furthermore, for each $t\in V(T)$, all vertices but $t$ contained in the torso of $t$ can be suppressed, and thus $\tor (t) \le 1$.
\end{proof}

\begin{lemma}\label{ltcwea}
  Let $G$ be a graph. Then $\tcw (G+e) \le \tcw (G) + 2$ for any edge $e$.
\end{lemma}

\begin{proof}
  Consider a tree-cut decomposition $(T, \mathcal{X})$ of $G$. This is also a tree-cut decomposition of $G+e$.

  Clearly, the adhesion of any node of $T$ can increase by at most 1.

  The torso-size of a vertex can also increase by at most 2, as $e$ can prevent at most both of its endpoints from being suppressed.
\end{proof}

\begin{corollary}
  Let $G$ be a graph, and $k$ be its feedback edge number. Then $\tcw (G) \le 2k + 1$.
\end{corollary}

\begin{proof}
  This directly follows from \Cref{ltcwt,ltcwea}.
\end{proof}

\section{W[1]-hardness of Perfect-SMTI and SRTI-Existence with Respect to Treedepth and Disjoint Paths Modulator Number}
\label{sec:MaxSRTIhardnessTDandFVS}\label{appaddhard}\label{ssmti}

All our hardness result are based on parameterized reductions from the problems \textsc{Clique} or \textsc{Multicolored Biclique}, two well-known W[1]-hard problems.
The so-called vertex selection gadgets are somewhat similar to those of Gupta et al.~\cite{GuptaSZ17}; however, the other gadgets in our reductions are different.
Next, we only discuss the main dissimilarities of the reductions we present here and the one of Gupta et al.~\cite{GuptaSZ17}.
We use one gadget for each edge whereas the reduction presented by Gupta et al.~\cite{GuptaSZ17} uses a single gadget for all edges between two color classes.
This subtle difference allows us to bound not only the treewidth of the resulting graph but rather both treedepth and the size of a disjoint paths modulator number.
Recall that both treedepth and disjoint paths modulator number upper-bound the treewidth of any graph.
It is worth noting that it is not clear whether the reduction of Gupta et al.~\cite{GuptaSZ17} can, with some additional changes and work, yield hardness for these parameters as well or not.
Furthermore, in our reduction all vertices have either strictly ordered preferences or a tie between (the only) two agents they find acceptable.

We reduce from \textsc{Multicolored Biclique}:

\defProblemTask{\textsc{Multicolored Biclique}}
{A bipartite graph $G$ with bipartition $V(G) = A \cupdot B$, a partition of $A$ into $k$ sets $A_1, \dots, A_k$, and a partition of $B$ into $k$ sets $B_1, \dots, B_k$.}
{Find a set $\{a_1, \dots, a_k\} \cup\{ b_1, \dots, b_k\}\subseteq V(G)$ such that for each $i\in [k]$ we have $a_i \in A_i$ and $b_i \in B_i$, and for each $i, j\in [k]$ we have $\{a_i, b_j\} \in E(G)$, or decide that no such set exists.}
We may assume without loss of generality that $|A_i| = |B_j | = n$ for all $i,j\in [k]$ for some~$n\in \mathbb{N}$, as \textsc{Multicolored Biclique} is W[1]-hard parameterized by solution size~$k$ even if $|A_i| = |B_j|$ for all $i,j\in [k]$.
Let $A_i = \{a_i^1, \dots, a_i^n\}$ and $B_j = \{b_j^1, \dots, b_j^n\}$.

The general idea of the reduction is as follows.
For each set $Z \in \{A_1, \dots, A_k, \allowbreak B_1, \dots, B_k\}$, we add a vertex-selection gadget, encoding which vertex of~$Z$ is part of the biclique.
For each edge from $G$, we add an edge gadget.
Furthermore, we add two incidence vertex $a_{i,j}$ and $b_{i,j}$ for each pair $(A_i, B_j)$.
A stable matching can match the two incidence vertices only if there exists an edge between the vertices selected by the vertex-selection gadgets for $A_i$ and $B_j$.

We will first describe the gadgets used in the reduction, starting with the vertex-selection gadget in \Cref{sec:fvs-vsg} and followed by the edge gadget and incidence vertices in \Cref{sEdgeGadget}.
Afterwards, we show how to connect the gadgets and we prove the correctness of the reduction in \Cref{sec:fvs-correctness}.

\subsection{Vertex-selection Gadget}
\label{sec:fvs-vsg}
A vertex-selection gadget for a set of vertices $Z_i = A_i$ or $Z_i =B_i$ has four vertices $c^Z_i$, $\bar{c}^Z_i$, $d^Z_i$, and~$\bar{d}^Z_i$, and, additionally for each $z\in Z_i$, four vertices $s^z$, $\bar{s}^z$, $t^z$, and $\bar{t}^z$.
We denote a vertex-selection gadget for the vertices $Z_i$ by $S(Z_i)$.
The preferences of these vertices over the other vertices of the vertex-seleciton gadget look as follows (where $v : v_1 \succ v_2 \succ v_3 \succ \dots$ means that vertex~$v$ prefers $v_1 $ to $v_2$ to $v_3$ and so on).
\begin{align*}
  c^Z_i & : s^{z^1_i} \succ t^{z^1_i} \succ s^{z^2_i} \succ t^{z_i^2} \succ \dots \succ s^{z_i^n} \succ t^{z_i^n},\\
  \bar{c}^Z_i & : \bar{s}^{z_i^n} \succ \bar{s}^{z_i^{n-1}} \succ \bar{s}^{z_i^{n-2}} \succ \dots \succ \bar{s}^{z^1_i},\\
  d^Z_i & : {t}^{z_i^n} \succ {t}^{z_i^{n-1}} \succ {t}^{z_i^{n-2}} \succ \dots \succ {t}^{z^1_i},\\
  \bar{d}^Z_i & :\bar t^{z^1_i} \succ \bar s^{z^1_i} \succ \bar t^{z^2_i} \succ \bar s^{z_i^2} \succ \dots \succ \bar t^{z_i^n} \succ \bar s^{z_i^n},\\
  s^v &: \bar{s}^v \sim c^Z_i, \\
  t^v &: \bar{t}^v \succ c^Z_i \succ d^Z_i, \\
  \bar{s}^v &: s^v \succ \bar{d}^Z_i \succ \bar{c}^Z_i, \\
  \bar{t}^v &: t^v \sim \bar{d}^Z_i.
\end{align*}
See \cref{fvsgm} for an example.
Note that $c^Z_i$ and $d^Z_i$ have also neighbors outside the vertex-selection gadget; we will describe the preferences of $c^Z_i$ and $d^Z_i$ over these neighbors later.

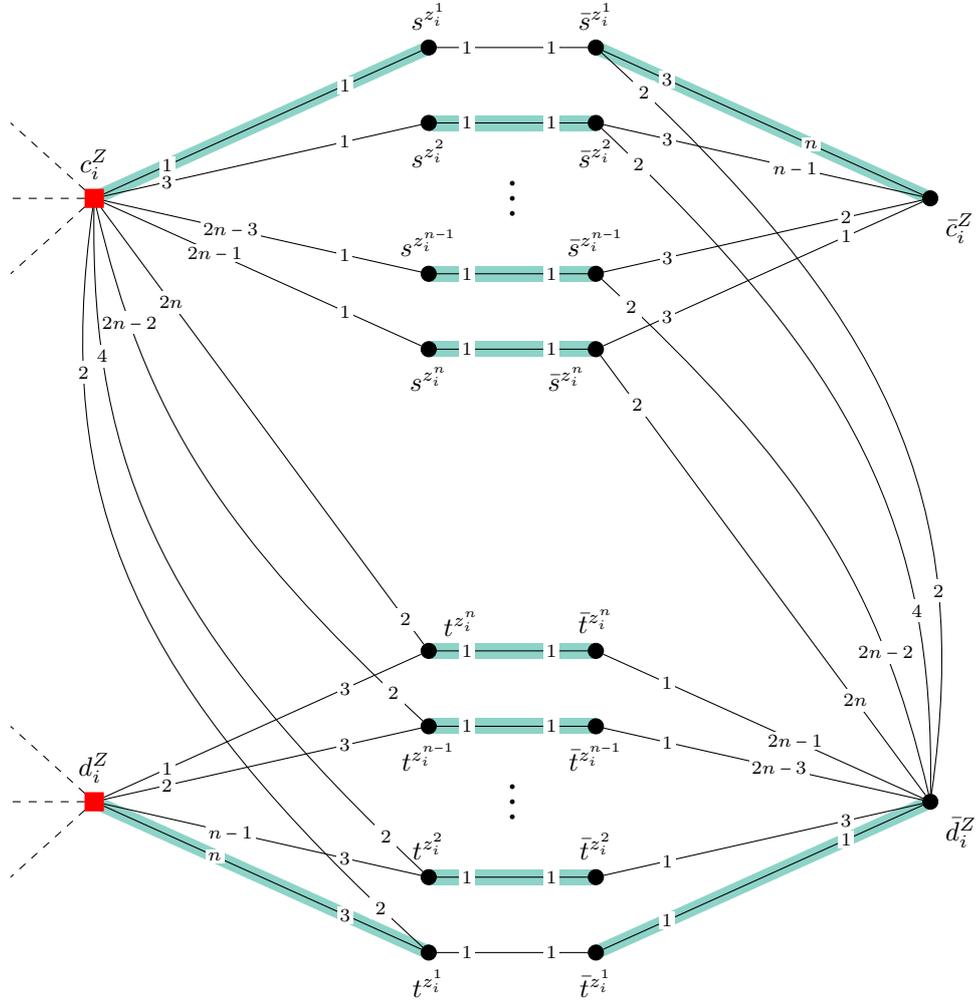
\begin{figure}
    \begin{center}
      \begin{tikzpicture}[xscale = 2.2]

        \node[squaredvertex, label=90:$c^Z_i$] (ci) at (-2, 1) {};
        \node[vertex, label=315:$\bar{c}^Z_i$] (cih) at (3, 1) {};
        \node[vertex, label=90:$s^{z_i^1}$] (a1) at (0, 3) {};
        \node[vertex, label=90:$\bar{s}^{z_i^1}$] (a2) at (1, 3) {};
        \node[vertex, label=270:$s^{z_i^2}$] (b1) at (0, 2) {};
        \node[vertex, label=270:$\bar{s}^{z_i^2}$] (b2) at (1, 2) {};
        \node[vertex, label=90:$s^{z_i^{n-1}}$] (c1) at (0, 0) {};
        \node[vertex, label=90:$\bar{s}^{z_i^{n-1}}$] (c2) at (1, 0) {};
        \node[vertex, label=270:$s^{z_i^n}$] (d1) at (0, -1) {};
        \node[vertex, label=265:$\bar{s}^{z_i^n}$] (d2) at (1, -1) {};

        \draw (ci) edge node[pos=0.2, fill=white, inner sep=1pt] {\scriptsize $1$}  node[pos=0.76, fill=white, inner sep=1pt] {\scriptsize $1$} (a1);
        \draw (a1) edge node[pos=0.2, fill=white, inner sep=1pt] {\scriptsize $1$}  node[pos=0.76, fill=white, inner sep=1pt] {\scriptsize $1$} (a2);
        \draw (a2) edge node[pos=0.2, fill=white, inner sep=1pt] {\scriptsize $3$}  node[pos=0.65, fill=white, inner sep=1pt] {\scriptsize $n$} (cih);

        \draw (ci) edge node[pos=0.2, fill=white, inner sep=1pt] {\scriptsize $3$}  node[pos=0.76, fill=white, inner sep=1pt] {\scriptsize $1$} (b1);
        \draw (b1) edge node[pos=0.2, fill=white, inner sep=1pt] {\scriptsize $1$}  node[pos=0.76, fill=white, inner sep=1pt] {\scriptsize $1$} (b2);
        \draw (b2) edge node[pos=0.2, fill=white, inner sep=1pt] {\scriptsize $3$}  node[pos=0.6, fill=white, inner sep=1pt] {\scriptsize $n-1$} (cih);

        \draw (ci) edge node[pos=0.4, fill=white, inner sep=1pt] {\scriptsize $2n - 3$}  node[pos=0.76, fill=white, inner sep=1pt] {\scriptsize $1$} (c1);
        \draw (c1) edge node[pos=0.2, fill=white, inner sep=1pt] {\scriptsize $1$}  node[pos=0.76, fill=white, inner sep=1pt] {\scriptsize $1$} (c2);
        \draw (c2) edge node[pos=0.2, fill=white, inner sep=1pt] {\scriptsize $3$}  node[pos=0.76, fill=white, inner sep=1pt] {\scriptsize $2$} (cih);

        \draw (ci) edge node[pos=0.35, fill=white, inner sep=1pt] {\scriptsize $2n - 1$}  node[pos=0.76, fill=white, inner sep=1pt] {\scriptsize $1$} (d1);
        \draw (d1) edge node[pos=0.2, fill=white, inner sep=1pt] {\scriptsize $1$}  node[pos=0.76, fill=white, inner sep=1pt] {\scriptsize $1$} (d2);
        \draw (d2) edge node[pos=0.2, fill=white, inner sep=1pt] {\scriptsize $3$}  node[pos=0.76, fill=white, inner sep=1pt] {\scriptsize $1$} (cih);

        \node[draw, circle, fill, inner sep = 0.5] (dot1) at (0.5, 1.2) {};
        \node[draw, circle, fill, inner sep = 0.5] (dot2) at (0.5, 1.0) {};
        \node[draw, circle, fill, inner sep = 0.5] (dot3) at (0.5, 0.8) {};

        \begin{scope}[on background layer]
          \newcommand{\colorBetweenTwoNodes}[3]{
            \fill[#1] ($(#2) + (0, .1)$) to ($(#2) - (0, .1)$) to ($(#3) - (0,.1)$) to ($(#3) + (0,.1)$) -- cycle;
          }
          \colorBetweenTwoNodes{mygreen}{ci}{a1}
          \colorBetweenTwoNodes{mygreen}{cih}{a2}
          \colorBetweenTwoNodes{mygreen}{b1}{b2}
          \colorBetweenTwoNodes{mygreen}{c1}{c2}
          \colorBetweenTwoNodes{mygreen}{d1}{d2}
        \end{scope}

        \def\shift{8}

        \node[squaredvertex, label=90:$d^Z_i$] (cib) at ($(ci) - (0, \shift)$) {};
        \node[vertex, label=315:$\bar{d}^Z_i$] (cihb) at ($(cih) - (0, \shift)$) {};
        \node[vertex, label=45:$t^{z_i^n}$] (a1b) at ($(a1) - (0, \shift)$) {};
        \node[vertex, label=90:$\bar{t}^{z_i^n}$] (a2b) at ($(a2) - (0, \shift)$) {};
        \node[vertex, label=270:$t^{z_i^{n-1}}$] (b1b) at ($(b1) - (0, \shift)$) {};
        \node[vertex, label=270:$\bar{t}^{z_i^{n-1}}$] (b2b) at ($(b2) - (0, \shift)$) {};
        \node[vertex, label=90:$t^{z_i^{2}}$] (c1b) at ($(c1) - (0, \shift)$) {};
        \node[vertex, label=90:$\bar{t}^{z_i^{2}}$] (c2b) at ($(c2) - (0, \shift)$) {};
        \node[vertex, label=270:$t^{z_i^1}$] (d1b) at ($(d1) - (0, \shift)$) {};
        \node[vertex, label=270:$\bar{t}^{z_i^1}$] (d2b) at ($(d2) - (0, \shift)$) {};

        \draw (cib) edge node[pos=0.2, fill=white, inner sep=1pt] {\scriptsize $1$}  node[pos=0.76, fill=white, inner sep=1pt] {\scriptsize $3$} (a1b);
        \draw (a1b) edge node[pos=0.2, fill=white, inner sep=1pt] {\scriptsize $1$}  node[pos=0.76, fill=white, inner sep=1pt] {\scriptsize $1$} (a2b);
        \draw (a2b) edge node[pos=0.2, fill=white, inner sep=1pt] {\scriptsize $1$}  node[pos=0.6, fill=white, inner sep=1pt] {\scriptsize $2n-1$} (cihb);

        \draw (cib) edge node[pos=0.2, fill=white, inner sep=1pt] {\scriptsize $2$}  node[pos=0.76, fill=white, inner sep=1pt] {\scriptsize $3$} (b1b);
        \draw (b1b) edge node[pos=0.2, fill=white, inner sep=1pt] {\scriptsize $1$}  node[pos=0.76, fill=white, inner sep=1pt] {\scriptsize $1$} (b2b);
        \draw (b2b) edge node[pos=0.2, fill=white, inner sep=1pt] {\scriptsize $1$}  node[pos=0.55, fill=white, inner sep=1pt] {\scriptsize $2n-3$} (cihb);

        \draw (cib) edge node[pos=0.4, fill=white, inner sep=1pt] {\scriptsize $n - 1$}  node[pos=0.76, fill=white, inner sep=1pt] {\scriptsize $3$} (c1b);
        \draw (c1b) edge node[pos=0.2, fill=white, inner sep=1pt] {\scriptsize $1$}  node[pos=0.76, fill=white, inner sep=1pt] {\scriptsize $1$} (c2b);
        \draw (c2b) edge node[pos=0.2, fill=white, inner sep=1pt] {\scriptsize $1$}  node[pos=0.76, fill=white, inner sep=1pt] {\scriptsize $3$} (cihb);

        \draw (cib) edge node[pos=0.35, fill=white, inner sep=1pt] {\scriptsize $n$}  node[pos=0.76, fill=white, inner sep=1pt] {\scriptsize $3$} (d1b);
        \draw (d1b) edge node[pos=0.2, fill=white, inner sep=1pt] {\scriptsize $1$}  node[pos=0.76, fill=white, inner sep=1pt] {\scriptsize $1$} (d2b);
        \draw (d2b) edge node[pos=0.2, fill=white, inner sep=1pt] {\scriptsize $1$}  node[pos=0.76, fill=white, inner sep=1pt] {\scriptsize $1$} (cihb);

        \node[draw, circle, fill, inner sep = 0.5] (dot1b) at ($(dot1) - (0, \shift)$) {};
        \node[draw, circle, fill, inner sep = 0.5] (dot2b) at ($(dot2) - (0, \shift)$) {};
        \node[draw, circle, fill, inner sep = 0.5] (dot3b) at ($(dot3) - (0, \shift)$) {};

        \begin{scope}[on background layer]
          \newcommand{\colorBetweenTwoNodes}[3]{
            \fill[#1] ($(#2) + (0, .1)$) to ($(#2) - (0, .1)$) to ($(#3) - (0,.1)$) to ($(#3) + (0,.1)$) -- cycle;
          }
          \colorBetweenTwoNodes{mygreen}{cib}{d1b}
          \colorBetweenTwoNodes{mygreen}{cihb}{d2b}
          \colorBetweenTwoNodes{mygreen}{b1b}{b2b}
          \colorBetweenTwoNodes{mygreen}{c1b}{c2b}
          \colorBetweenTwoNodes{mygreen}{a1b}{a2b}
        \end{scope}

        \node[squaredvertex] (xi) at (ci) {};
        \node[squaredvertex] (xij) at (cib) {};

        \draw (xi) edge[bend right = 0] node[pos=0.2, fill=white, inner sep=2pt] {\scriptsize $2n$}  node[pos=0.95, fill=white, inner sep=2pt] {\scriptsize $2$} (a1b);
        \draw (xi) edge[bend right = 10] node[pos=0.2, fill=white, inner sep=2pt] {\scriptsize $2n-2$}  node[pos=0.95, fill=white, inner sep=2pt] {\scriptsize $2$} (b1b);
        \draw (xi) edge[bend right = 13] node[pos=0.2, fill=white, inner sep=2pt] {\scriptsize $4$}  node[pos=0.95, fill=white, inner sep=2pt] {\scriptsize $2$} (c1b);
        \draw (xi) edge[bend right = 15] node[pos=0.2, fill=white, inner sep=2pt] {\scriptsize $2$}  node[pos=0.95, fill=white, inner sep=2pt] {\scriptsize $2$} (d1b);

        \draw (cihb) edge[bend right = 15] node[pos=0.25, fill=white, inner sep=2pt] {\scriptsize $2$}  node[pos=0.95, fill=white, inner sep=2pt] {\scriptsize $2$} (a2);
        \draw (cihb) edge[bend right = 13] node[pos=0.25, fill=white, inner sep=2pt] {\scriptsize $4$}  node[pos=0.95, fill=white, inner sep=2pt] {\scriptsize $2$} (b2);
        \draw (cihb) edge[bend right = 10] node[pos=0.25, fill=white, inner sep=2pt] {\scriptsize $2n-2$}  node[pos=0.95, fill=white, inner sep=2pt] {\scriptsize $2$} (c2);
        \draw (cihb) edge[bend right = 0] node[pos=0.2, fill=white, inner sep=2pt] {\scriptsize $2n$}  node[pos=0.9, fill=white, inner sep=2pt] {\scriptsize $2$} (d2);

        \draw (ci) edge[dashed] ($(ci) + (-0.5, 0)$);
        \draw (ci) edge[dashed] ($(ci) + (-0.5, 1)$);
        \draw (ci) edge[dashed] ($(ci) + (-0.5, -1)$);

        \draw (cib) edge[dashed] ($(cib) + (-0.5, 1)$);
        \draw (cib) edge[dashed] ($(cib) + (-0.5, 0)$);
        \draw (cib) edge[dashed] ($(cib) + (-0.5, -1)$);
      \end{tikzpicture}

    \end{center}
    \caption{A vertex selection gadget for the vertex set $Z_i = \{z_i^1, \dots, z_i^n\}$ for some $Z\in \{A, B \}$ and $i \in [k]$.
    For an edge $\{v, w\}$, the number on this edge closer to $v$ indicates how $v$ ranks $w$.
    For example, vertex $\bar{c}^A_i$ likes $\bar{s}^{a_i^n}$ most, $\bar s^{a_i^{n-1}}$ second-most, vertex $\bar s^{a_i^2}$ is at the $n-1$st position in the preferences of $\bar c^A_i$, and $\bar s^{a_i^1}$ is at the $n$th position in the preferences of $\bar c^A_i$.
    The red, squared vertices $c^A_i$ and $d^A_i$ are the only vertices having neighbors not depicted in the picture.
    Dashed edges are not part of the vertex-selection gadget, but will be added later in the reduction.
    The green edges form a stable matching.}\label{fvsgm}

\end{figure}

\begin{observation}\label{operf}
  If a graph $H$ contains a vertex-selection gadget $S(Z_i)$, then any perfect matching $M$ matches $c^Z_i$ and $d^Z_i$ to vertices from the vertex-selection gadget.
\end{observation}

\begin{proof}
  Note that the vertex-selection gadget is bipartite with $U:=\{s^v, t^v : v\in Z_i\} \cup \{\bar{c}^Z_i, \bar{d}^Z_i\}$ and $W:= \{\bar{s}^v, \bar{t}^v : v\in Z_i\} \cup \{c^Z_i, d^Z_i\}$.
  Since $|U| = n + 2 = |W|$, the observation follows.
\end{proof}

Now we show that any perfect stable matching must match $c^Z_i$ and $d^Z_i$ to the vertices $s^v$ and $t^v$ for some $v\in Z_i$.

\begin{lemma}\label{lvsgm}
  Let $H$ be a graph containing a vertex selection gadget $S(Z_i)$ such that $c^Z_i$ and~$d^Z_i$ are the only vertices from $S(Z_i)$ with neighbors outside the vertex selection gadget.

  In any perfect stable matching $M$, the vertices $\bar{d}$ and $\bar{c}$ are matched to vertices $\bar{t}^v$ and $\bar{s}^v$, and $d$ and $c$ are matched to $t^v$ and $s^v$ for some $v\in [n]$.
\end{lemma}

\begin{proof}
  Since $M$ is perfect, we have $\{\bar{s}^v, \bar{c}^Z_i\} \in M$ for some vertex~$v\in Z_i$.
  As $M$ is perfect and the only neighbors of $s^v$ are $\bar{s}^v$ and $c^Z_i$, it follows that $\{s^v, c^Z_i\} \in M$.

  Similarly, by~\cref{operf}, vertex~$d^Z_i$ is matched to a vertex~$t^w$.
  As $M$ is perfect and the only neighbors of $\bar{t}^w$ are $t^w$ and $\bar{d}^Z_i$, it follows that $\{t^w, \bar{d}^Z_i\} \in M$.

  Let $v= z_i^\ell$, and $w = z_i^p$.
  It remains to show $\ell = p$.
  If $\ell < p$, then edge~$\{\bar{d}^Z_i, \bar{s}^v\}$ blocks~$M$, a contradiction.
  If $\ell > p $, then edge~$\{c^Z_i, t^w\}$ blocks $M$, a contradiction.
  Thus, we have $v =w $, and the lemma follows.
\end{proof}

If a matching contains edges $\{c_i^z, s^v\}$, $\{\bar c_i^z, \bar s^v\}$, $\{d^Z_i, t^v\}$, $\{\bar d^Z_i, \bar t^v\}$, and $\{s^w, \bar s^w\}$ and $\{t^w, \bar t^w\}$ for all $w\in Z_i \setminus \{v\}$, then we say that the vertex selection gadget $S(Z_i)$ \emph{selects} the vertex $v$.

We now show that selecting a vertex from $X$ also yields a stable matching inside the vertex selection gadget.

\begin{lemma}\label{lpmivsg}
  For each vertex-selection gadget $S(Z_i)$ with $Z\in \{A, B\}$ and each $v\in Z_i$, there exists a perfect stable matching $M$ inside the edge gadget containing the edges $\{c^Z_i, s^v\}$ and~$\{d^Z_i, t^v\}$.
\end{lemma}

\begin{proof}
  Let $M = \{\{c^Z_i, s^v\}, \{d^Z_i, t^v\}, \{\bar{c}^Z_i, \bar{s}^v\}, \{\bar{d}^Z_i, \bar{t}^v\} \cup \{\{s^x, \bar{s}^x\}, \{t^x, \bar{t}^x\} : x\in V\setminus \{v\}\}$.
  We show that $M$ is stable, proving the lemma.

  Every vertex of the form $s^x$, $\bar{s}^x$, $t^x$, or $\bar{t}^x$ for $x\neq v$ cannot be contained in a blocking pair, as they are matched to one of their top choices.
  Thus, only edges $\{c^Z_i, t^v\}$ and $\{\bar{d}^Z_i, \bar{s}^v\}$ can be blocking.
  However, $c^Z_i$ prefers $M(c^Z_i) = s^v$ to $t^v$, and $\bar{d}^Z_i$ prefers $M(\bar{d}^Z_i) = \bar{t}^v$ to $\bar{s}^v$.
  Thus, $M$ is stable.
\end{proof}

We now turn to incidence vertices and edge gadgets which shall ensure that the vertices selected by a perfect stable matching indeed form a biclique.

\subsection{Incidence Vertices and Edge Gadgets}\label{sEdgeGadget}
For each pair $(i, j)$ with $i, j\in [k]$, we add two vertices $a_{ij}$ and $b_{ij}$ (which we call \emph{incidence vertices}) to check whether the vertices selected by~$S(A_i)$ and $S(B_j)$ are adjacent.
Each incidence vertex ranks its neighbors in an arbitrary way.

An edge gadget for an edge $e = \{a_i^\ell, b_j^p\}$ consists of a path of length three and has the vertices $a^e_1$, $a^e_2$, $b^e_1$, and $b^e_2$.
The inner two vertices (i.e., $a^e_2$ and $b^e_2$) tie both their neighbors.
Vertex $a^e_1$ ($b^e_1$) is connected to the vertices $c^A_i$ and $d^A_i$ of the vertex selection gadget $S(A_i)$ ($c^B_j$ and $d^B_j$ of $S(B_j)$), and to incidence vertex $a_{ij}$ ($b_{ij}$).
The preference lists look as follows.
\begin{align*}
  a^e_1 &: a_2^e \succ c^A_i \succ d^A_i \succ a_{ij},\\
  a^e_2 &: (a^e_1, b^e_2),\\
  b^e_2 &: (a^e_2, b^e_1),\\
  b^e_1 &: b^e_2 \succ c^B_j \succ d^B_j \succ b_{ij}.
\end{align*}
See \Cref{fedgegadgetm} for an example.

\begin{figure}
  \begin{center}
    \begin{tikzpicture}[yscale = 1.5]
      \node[vertex, label=180:$c^A_i$] (ai) at (-1.5, 0) {};
      \node[vertex, label=180:$d^A_i$] (aib) at (-1.5, -2) {};
      \node[vertex, label=0:$c^B_j$] (bi) at (9.5, 0) {};
      \node[vertex, label=0:$d^B_j$] (bib) at (9.5, -2) {};

      \node[vertex, label=90:$a_{ij}$] (v1) at (3, 0) {};
      \node[vertex, label=90:$a^e_1$] (v3) at (1, -1) {};
      \node[vertex, label=270:$a^e_2$] (v5) at (3, -2) {};
      \node[vertex, label=270:$b^e_2$] (v7) at (5, -2) {};
      \node[vertex, label=90:$b^e_1$] (v8) at (7, -1) {};
      \node[vertex, label=90:$b_{ij}$] (v10) at (5, 0) {};

      \draw (ai) edge node[pos=0.2, fill=white, inner sep=1pt] {\scriptsize $2\ell - 1$}  node[pos=0.76, fill=white, inner sep=1pt] {\scriptsize $2$} (v3);
      \draw (aib) edge node[pos=0.3, fill=white, inner sep=1pt] {\scriptsize $2(n + 1 - \ell) - 1$}  node[pos=0.76, fill=white, inner sep=1pt] {\scriptsize $3$} (v3);
      \draw (bi) edge node[pos=0.2, fill=white, inner sep=1pt] {\scriptsize $2p - 1$}  node[pos=0.76, fill=white, inner sep=1pt] {\scriptsize $2$} (v8);
      \draw (bib) edge node[pos=0.2, fill=white, inner sep=1pt] {\scriptsize $2(n + 1 - p) -1$}  node[pos=0.76, fill=white, inner sep=1pt] {\scriptsize $3$} (v8);

      \draw (v1) edge node[pos=0.2, fill=white, inner sep=1pt] {\scriptsize $7$}  node[pos=0.76, fill=white, inner sep=1pt] {\scriptsize $4$} (v3);
      \draw (v3) edge node[pos=0.2, fill=white, inner sep=1pt] {\scriptsize $1$}  node[pos=0.76, fill=white, inner sep=1pt] {\scriptsize $1$} (v5);
      \draw (v5) edge node[pos=0.2, fill=white, inner sep=1pt] {\scriptsize $1$}  node[pos=0.76, fill=white, inner sep=1pt] {\scriptsize $1$} (v7);
      \draw (v7) edge node[pos=0.2, fill=white, inner sep=1pt] {\scriptsize $1$}  node[pos=0.76, fill=white, inner sep=1pt] {\scriptsize $1$} (v8);
      \draw (v8) edge node[pos=0.2, fill=white, inner sep=1pt] {\scriptsize $4$}  node[pos=0.76, fill=white, inner sep=1pt] {\scriptsize $5$} (v10);

        \begin{scope}[on background layer]
          \newcommand{\colorBetweenTwoNodes}[3]{
            \fill[#1] ($(#2) + (0, .1)$) to ($(#2) - (0, .1)$) to ($(#3) - (0,.1)$) to ($(#3) + (0,.1)$) -- cycle;
          }
          \colorBetweenTwoNodes{mygreen}{v1}{v3}
          \colorBetweenTwoNodes{mygreen}{v5}{v7}
          \colorBetweenTwoNodes{mygreen}{v8}{v10}
          \colorBetweenTwoNodes{mylila}{v3}{v5}
          \colorBetweenTwoNodes{mylila}{v7}{v8}
        \end{scope}
    \end{tikzpicture}

  \end{center}
  \caption{An edge gadget for the edge $e = \{a_i^\ell, b^p_j\}$ with all its neighbors.
    For an edge $\{v, w\}$, the number on this edge closer to $v$ indicates how $v$ ranks $w$.
    For example, vertex $a^e_1$ prefers $a^e_2$ to $c^A_i$, vertex~$c^A_i$ to $d^A_i$, and $d^A_i$ to $a_{ij}$.
  The green and purple edges represent the two different sets of edges a perfect stable matching can contain.}
  \label{fedgegadgetm}

\end{figure}
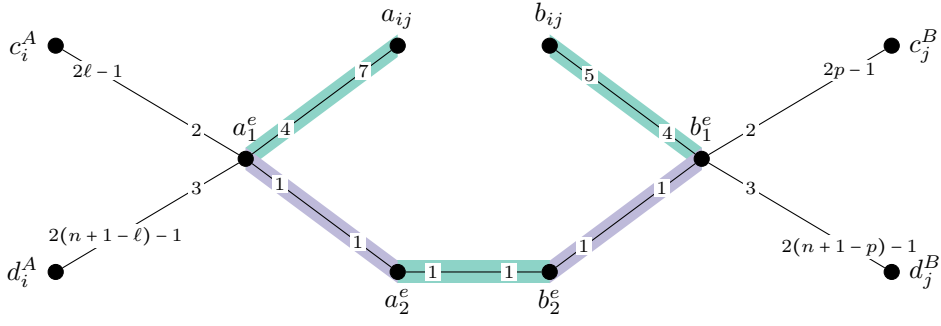

\subsection{The Reduction}
\label{sec:fvs-correctness}

Given an instance $(G, A_1, \dots, A_k, B_1, \dots, B_k)$ of \textsc{Multicolored Biclique}, we construct a \textsc{Perfect-SMTI}-instance as follows:
For each set $A_i$, we add a vertex selection gadget~$S(A_i)$.
Similary, for each set $B_i$, we add a vertex selection gadget $S(B_i)$.
We add~$k^2$ incidence vertices $a_{ij}$ and $k^2$ incidence vertices $b_{ij}$.
For each edge $e=\{a_i^\ell, b_j^p\}\in E(G)$, we add an edge gadget.
We connect the vertex $a^e_1$ to~$a_{ij}$,~$c^A_i$, and $d^A_i$, where for the ranks we set $\rk_{c^A_i} (a^e_1): = 2 \ell$ and $\rk_{d^A_i} (a^e_1): = 2(n + 1 -\ell )$.
Similarly, $b^e_1$ is connected to $b_{ij}$, $c^B_j$, and $d^B_j$, where for the ranks we set $\rk_{c^B_j} (b^e_1) := 2 p$ and $\rk_{d^B_j} (b^e_1) := 2 (n + 1 - p)$.

We call the resulting graph $H$.
First, we show that $H$ is bipartite.

\begin{lemma}\label{lbip}
  Graph $H$ is bipartite.
\end{lemma}

\begin{proof}
  Let
  \[
    U:= \left\{\bar{c}^A_i, \bar{d}^A_i, s^a, t^a, a^e_1, b^e_2, b_{ij}, c^B_j, d^B_j, \bar{s}^b, \bar{t}^b : i, j\in [k], e\in E(G), a\in \bigcup_{\ell=1}^k A_\ell, b\in \bigcup_{k=1}^k B_\ell\right\}
  \]
  and
  \[
    W:= \left\{c^A_i, d^A_i, \bar{s}^a, \bar{t}^a, a^e_2, b^e_1, a_{ij}, \bar{c}^B_j, \bar{d}^B_j, s^b, t^b : i, j\in [k], e\in E(G), a\in \bigcup_{\ell=1}^k A_\ell, b\in \bigcup_{k=1}^k B_\ell\right\}\,.
  \]

  It is easy to check that $U \uplus W$ is a bipartition of $H$.
\end{proof}

Next, we show that the parameters treedepth and disjoint paths modulator number of~$H$ are indeed bounded by a function of~$k$.

\begin{lemma}\label{ltd}
  Graph $H$ has treedepth $O(k)$ and disjoint paths modulator number size $O(k^2)$.
\end{lemma}

\begin{proof}
  First, we show that the treedepth of $G$ is bounded linearly in $k$.
  Consider the following rooted tree $T$ on $V(G)$:
  The root of~$T$ is $c^A_1$.
  Then, $T$ contains path $c^{A}_1$-$c^A_2$-$\dots$-$c^A_k$-$\bar c^A_1$-$\bar c^A_2$-$\dots$-$\bar c^A_k$-$d^A_1$-$\dots$-$d^A_k$-$\bar d^A_1$-$\dots$-$\bar d^A_k$-$c^{B}_1$-$\dots$-$c^B_k$-$\bar c^B_1$-$\dots$-$\bar c^B_k$-$d^B_1$-$\dots$-$d^B_k$-$\bar d^B_1$-$\dots$-$\bar d^B_k$.
  For every $i,j \in [k]$, vertex $a_{ij}$ is a child of $d^B_k$, and $b_{ij} $ is a child of $a_{ij}$.
  For every edge $e= \{a^p_i, b^q_j\} \in E$, vertex $a^e_1$ is a child of $b_{ij}$, vertex $a^e_2$ is a child of $a^e_2$, vertex $b^e_2$ is a child of $a^e_2$, and $b^e_1$ is a child of $b^e_1$.
  It is easy to verify that for each edge $e = \{v, w\}\in E(H)$, either $v$ is an ancestor of $w$ or $w$ is an ancestor of $v$ in $T$.
  As the depth of $T$ is $8k + 6$, the treedepth of $H$ is at most $8k+6$.

  Now we show that disjoint path modulator number of $H$ is bounded by~$O(k^2)$.  
  Consider the set $X:=\{c^A_i, d^A_i, \bar{c}^A_i, \bar{d}^A_i, c^B_i, d^B_i, \bar{c}^B_i, \bar{d}^B_i : 1~\le~i\le~k\} \cup\{a_{ij}, b_{ij}: i, j\in [k]\}$.
  Note that $|X| = O(k^2)$.
  After removing this $X$, only paths of length one (inside the vertex-selection gadgets) and paths of length three (inside the edge gadgets) remain.
  Thus, we have that $H$ has disjoint path modulator number $O(k^2)$. 
\end{proof}

We now turn to the correctness of the reduction.
We start by showing that a multicolored biclique implies a perfect stable matching.

\begin{lemma}\label{lh}
  If $G$ contains a multicolored biclique $K_{k, k}$, then $H$ contains a perfect stable matching.
\end{lemma}

\begin{proof}
  Let $\{a_1^{\ell_1}, \dots, a_k^{\ell_k}\}\cup \{b_1^{p_1}, \dots, b_k^{p_k}\}$ be the multicolored biclique.
  We construct a perfect stable matching~$M$.
  The vertex selection gadget $S(A_i)$ respectively~$S(B_i)$ shall select~$a_i^{\ell_i}\in A_i$ respectively~$b_i^{p_i}\in B_i$.
  For each edge~$e = \{a_i^{\ell_i}, b_j^{p_j}\}$ for $i,j \in [k]$, we add edges $\{a_{ij}, a^e_1\}$, $\{a^e_2, b^e_2\}$, and $\{b_1^e, b_{ij}\}$ to $M$, while for all other edges~$e'$, we add the edges $\{a^{e'}_1, a^{e'}_2\}$ and~$\{b^{e'}_1, b^{e'}_2\}$ to~$M$.
  We first show that $M$ is perfect, and then that $M$ is stable.

  \begin{claim*}
  The matching $M$ is perfect.
  \end{claim*}
  \begin{claimproof}
    Each incidence vertex is matched, as $M$ contains the edges $\{a_{ij}, a^e_1\}$ and $\{b_{ij}, b^e_1\}$ for the edge $e = \{a_i^{\ell_i}, b_j^{p_j}\}$.
    All vertices contained in vertex selection gadgets are matched by \cref{lpmivsg}, and all vertices contained in edge gadgets are clearly matched.
  \end{claimproof}

  \begin{claim*}
  The matching $M$ is stable.
  \end{claim*}
  \begin{claimproof}
    No blocking pair is contained inside a vertex-selection gadget by \cref{lpmivsg}.
    For any edge $e\in E(G)$, vertices~$a_2^e$ and $b_2^e$ are matched to one of their top choices, and thus are not part of any blocking pair.
    Unless $e$ is an edge of the multicolored biclique, also the vertices $a_1^e$ and $b_1^e$ are matched to their top choice and thus not part of a blocking pair.

    For every incidence vertex~$a_{ij}$ or $b_{ij}$, all vertices adjacent to $a_{ij}$ or $b_{ij}$ except for $M(a_{ij})$ or $M(b_{ij})$ are matched to their top choice.
    Thus, no incidence vertex is part of a blocking pair.

    Thus, the only possible blocking pairs are $\{c^A_i, a^e_1\}$ and $\{d^A_i, a_1^e\}$ when $\{a_{ij}, a_1^e\}\in M$ and $\{c^B_j, b^e_1\}$ and $\{d^B_j, b_1^e\}$ when $\{b_{ij}, b_1^e\}\in M$.
    However, since $e$ is incident to vertex~$a^{\ell_i}_i$ selected by $S(A_i)$, vertex~$c^A_i$ respectively~$d^A_i$ prefers $M(c^A_i) = s^{a^{\ell_i}}_i$ respectively $M(d^A_i) = t^{a^{\ell_i}}$ to $a_1^e$.
    Pairs~$\{c^B_j, b^e_1\}$ and $\{d^B_j, b_1^e\}$ are not blocking by symmetric arguments.
  \end{claimproof}

  This finishes the proof of \Cref{lh}.
\end{proof}

We now show the reverse direction, i.e., that every perfect stable matching implies a multicolored biclique.
In order to do so, we start by observing that incidence vertices $a_{ij}$ and $b_{ij}$ are matched to vertices from the same edge gadget in every perfect stable matching.

\begin{lemma}\label{leg}
  Let $M$ be any perfect stable matching in $H$, and $e = \{a_i^\ell, b_j^p\}$.
  If $\{a_{ij}, a^e_1\}\in M$, then also $\{b_{ij}, b^e_1\}\in M$.
\end{lemma}

\begin{proof}
  Assume that $\{a_{ij}, a^e_1\}\in M$.
  Then $M$ also contains the edge $\{a^e_2, b^e_2\}$, as $M$ is perfect.
  Since the vertices $c^B_j$ and $d^B_j$ are matched inside the vertex selection gadget by \cref{lvsgm}, the vertex $b^e_1$ must be matched to $b_{ij}$.
\end{proof}

\begin{lemma}\label{lr}
  If $H$ contains a perfect stable matching $M$, then $G$ contains a biclique $K_{k, k}$.
\end{lemma}

\begin{proof}
  By \cref{leg}, $M$ matches the incidence vertices $a_{ij}$ and $b_{ij}$ to the same edge gadget for all~$i, j\in [k]$.
  By \cref{lvsgm}, any vertex selection gadget $S(A_i)$ respectively~$S(B_i)$ indeed selects a vertex $a_i^{\ell_i}$ respectively $b_i^{p_i}$.

  It is enough to show that $G$ contains the edge $\{a_i^{\ell_i}, b_j^{p_j}\}$ for all $i, j\in[k]$.
  Let $e= \{a_i^q, b\}\in E(G)$ with $a^i_q\in A_i$ and $b\in B_j$ be the edge such that $\{a_{ij}, a_1^e\}\in M$ (this edge exists since $M$ is perfect).
  Then $a^e_1$ prefers both $c_i^{A}$ and $d_i^A$ over $M$.
  By the stability of $M$, it follows that $\rk_{c_i^A} (M(c_i^A)) = 2 \ell_i - 1\le 2 q -1$, and $\rk_{d_i^A} (M(d_i^A) ) = 2 ( n + 1 -\ell_i) - 1 \le 2 (n + 1 - q\ell) -1 $.
  This implies $\ell_i = q$.
  Since $M$ is perfect, it also contains the edge $\{b_{ij}, b_1^e\}$.
  By symmetric arguments, it follows that $b_j = b$, and thus, $G$ contains a biclique of size $k$.
\end{proof}

Alltogether, we have proven the W[1]-hardness of \textsc{Perfect-SMTI} (and therefore also \textsc{Max-SMTI}) parameterized by treedepth and disjoint paths modulator number.

\begin{theorem}\label{thm:W-hard-td-fvn}
  \textsc{Perfect-SMTI} and \textsc{Max-SMTI} are W[1]-hard for the parameter treedepth plus disjoint paths modulator number.
\end{theorem}

\begin{proof}
  Directly follows from \cref{ltd,lbip,lh,lr}.
\end{proof}

W[1]-hardness for \textsc{SRTI-Existence} parameterized by treedepth and disjoint paths modulator number now easily follows.

\begin{corollary}\label{cor:W-hard-existence}
  \textsc{SRTI-Existence} parameterized by the combined parameter treedepth plus disjoint paths modulator number is W[1]-hard.
\end{corollary}

\begin{proof}
  We reduce from \textsc{Perfect-SMTI} parameterized by treedepth plus disjoint paths modulator number.
  Let $\mathcal{I}$ be an instance of \textsc{Perfect-SMTI}, and let~$H$ be the acceptability graph of this instance.

  We add four vertices $\alpha_1$, $\alpha_2$, $\beta_1$, and $\beta_2$ with the following preferences (see \Cref{f3c} for an example):
  \begin{align*}
    \alpha_1 &: (V(H)) \succ \alpha_2,\\
    \alpha_2 &: \alpha_1 \succ \beta_2 \succ \beta_1 \succ (V(H)),\\
    \beta_1 &: \alpha_2 \succ \beta_2,\\
    \beta_2 &: \beta_1 \succ \alpha_2.
  \end{align*}
  For each $v\in V(H)$, we append ``$\succ \alpha_2 \succ\alpha_1$'' at the end of the preferences of $v$, i.e., $v$ prefers all its neighbors in $H$ to both $\alpha_1$ and $\alpha_2$ and prefers $\alpha_2 $ to $\alpha_1$.
  We call the resulting instance $\mathcal{I}'$.
  As we added only four vertices, the treedepth and disjoint paths modulator number of the acceptability graph of $\mathcal{I}$ can be at most by four larger than the one of $H$.

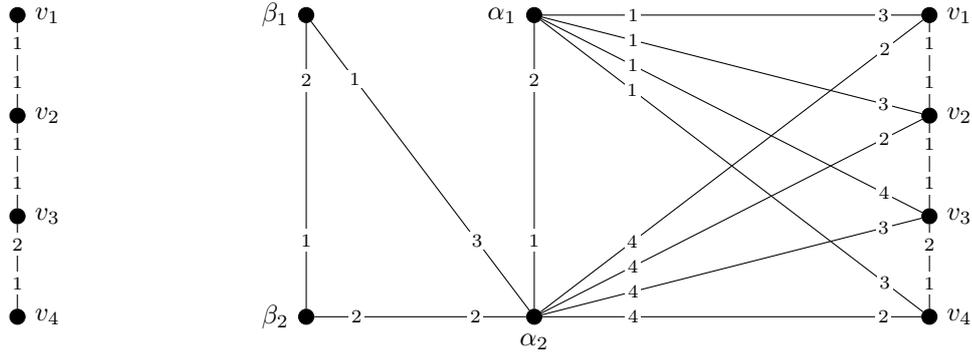
\begin{figure}
  \begin{center}    \begin{tikzpicture}[xscale = 4]

        \node[vertex, label=0:$v_1$] (aa1) at (-3, 3) {};
        \node[vertex, label=0:$v_2$] (ba1) at (-3, 1.667) {};
        \node[vertex, label=0:$v_3$] (ca1) at (-3, 0.333) {};
        \node[vertex, label=0:$v_4$] (da1) at (-3, -1) {};

        \draw (da1) edge node[pos=0.3, fill=white, inner sep=2pt] {\scriptsize $1$}  node[pos=0.76, fill=white, inner sep=2pt] {\scriptsize $2$} (ca1);
        \draw (ca1) edge node[pos=0.3, fill=white, inner sep=2pt] {\scriptsize $1$}  node[pos=0.76, fill=white, inner sep=2pt] {\scriptsize $1$} (ba1);
        \draw (ba1) edge node[pos=0.3, fill=white, inner sep=2pt] {\scriptsize $1$}  node[pos=0.76, fill=white, inner sep=2pt] {\scriptsize $1$} (aa1);

        \node[vertex, label=0:$v_1$] (a1) at (0, 3) {};
        \node[vertex, label=0:$v_2$] (b1) at (0, 1.667) {};
        \node[vertex, label=0:$v_3$] (c1) at (0, 0.333) {};
        \node[vertex, label=0:$v_4$] (d1) at (0, -1) {};

        \node[vertex, label=180:$\alpha_1$] (alpha1) at (-1.3, 3) {};
        \node[vertex, label=270:$\alpha_2$] (alpha2) at (-1.3, -1) {};

        \node[vertex, label=180:$\beta_1$] (beta1) at ($(alpha1) + (-0.75, 0)$) {};
        \node[vertex, label=180:$\beta_2$] (beta2) at ($(alpha2)  + (-0.75, 0)$) {};

        \draw (a1) edge node[pos=0.1, fill=white, inner sep=1pt] {\scriptsize $3$}  node[pos=0.76, fill=white, inner sep=1pt] {\scriptsize $1$} (alpha1);
        \draw (b1) edge node[pos=0.1, fill=white, inner sep=1pt] {\scriptsize $3$}  node[pos=0.76, fill=white, inner sep=1pt] {\scriptsize $1$} (alpha1);
        \draw (c1) edge node[pos=0.1, fill=white, inner sep=1pt] {\scriptsize $4$}  node[pos=0.76, fill=white, inner sep=1pt] {\scriptsize $1$} (alpha1);
        \draw (d1) edge node[pos=0.1, fill=white, inner sep=1pt] {\scriptsize $3$}  node[pos=0.76, fill=white, inner sep=1pt] {\scriptsize $1$} (alpha1);
        \draw (a1) edge node[pos=0.1, fill=white, inner sep=1pt] {\scriptsize $2$}  node[pos=0.76, fill=white, inner sep=1pt] {\scriptsize $4$} (alpha2);
        \draw (b1) edge node[pos=0.1, fill=white, inner sep=1pt] {\scriptsize $2$}  node[pos=0.76, fill=white, inner sep=1pt] {\scriptsize $4$} (alpha2);
        \draw (c1) edge node[pos=0.1, fill=white, inner sep=1pt] {\scriptsize $3$}  node[pos=0.76, fill=white, inner sep=1pt] {\scriptsize $4$} (alpha2);
        \draw (d1) edge node[pos=0.1, fill=white, inner sep=1pt] {\scriptsize $2$}  node[pos=0.76, fill=white, inner sep=1pt] {\scriptsize $4$} (alpha2);

        \draw (alpha1) edge node[pos=0.2, fill=white, inner sep=2pt] {\scriptsize $2$}  node[pos=0.76, fill=white, inner sep=2pt] {\scriptsize $1$} (alpha2);

        \draw (beta1) edge node[pos=0.2, fill=white, inner sep=2pt] {\scriptsize $2$}  node[pos=0.76, fill=white, inner sep=2pt] {\scriptsize $1$} (beta2);

        \draw (beta1) edge node[pos=0.2, fill=white, inner sep=1pt] {\scriptsize $1$}  node[pos=0.76, fill=white, inner sep=1pt] {\scriptsize $3$} (alpha2);
        \draw (beta2) edge node[pos=0.2, fill=white, inner sep=1pt] {\scriptsize $2$}  node[pos=0.76, fill=white, inner sep=1pt] {\scriptsize $2$} (alpha2);

        \draw (d1) edge node[pos=0.3, fill=white, inner sep=2pt] {\scriptsize $1$}  node[pos=0.76, fill=white, inner sep=2pt] {\scriptsize $2$} (c1);
        \draw (c1) edge node[pos=0.3, fill=white, inner sep=2pt] {\scriptsize $1$}  node[pos=0.76, fill=white, inner sep=2pt] {\scriptsize $1$} (b1);
        \draw (b1) edge node[pos=0.3, fill=white, inner sep=2pt] {\scriptsize $1$}  node[pos=0.76, fill=white, inner sep=2pt] {\scriptsize $1$} (a1);
    \end{tikzpicture}

  \end{center}
  \caption{An example for the reduction from \textsc{Perfect-SMTI} to \textsc{SRTI-Existence}.
  The \textsc{Perfect-SMTI} instance is depicted on the left, and the \textsc{SRTI-Existence} instance created by the reduction is depicted on the left.}\label{f3c}
\end{figure}

  Given a perfect stable matching~$M$ in $\mathcal{I}$, we define $M' := M \cup \{\{\alpha_1, \alpha_2\}, \{\beta_1, \beta_2\}\}$, and claim that $M' $ is stable in $\mathcal{I}'$.
  Since $M$ is perfect and stable, no vertex from $V(H)$ is part of a blocking pair.
  Vertices $\alpha_2$ and $\beta_2$ are matched to their first choice and thus also not part of a blocking pair.
  As there is no edge between the only two remaining vertices ($\alpha_1$ and $\alpha_2$), it follows that $M'$ is stable.

  Vice versa, given a stable matching $M'$ in $\mathcal{I}'$, we construct a perfect stable matching~$M : = M' \setminus \{\{\alpha_1, \alpha_2\}, \{\beta_1, \beta_2\}\}$.
  First, note that if $M'$ does not contain edge~$\{\alpha_1, \alpha_2\}$, then two vertices from $\alpha_2$, $\beta_1$, and $\beta_2$ from a blocking pair.
  Furthermore, every vertex from $V(H)$ is matched in $M'$, as otherwise there exists an unmatched vertex~$v$ and two vertices from $v$, $\alpha_1$, and $\alpha_2$ form a blocking pair.
  It follows that $M$ is a perfect matching.
  Matching~$M$ is also stable, as any blocking pair for $M$ would also be a blocking pair for $M'$.
\end{proof}

Note that by arbitrarily breaking all ties involving agent $c^Z_i$ or $d^Z_i$ for all $Z\in \{A, B\}$ and $i \in [k]$, the hardness results in this section extend even in the restricted case that the preferences of every vertex are either strictly ordered or a tie of size two:
\Cref{operf} implies that $c^Z_i$ or $d^Z_i$ cannot be matched to a vertex it ties with another vertex in a stable matching, and thus, breaking these ties arbitrarily does not change the set of stable matchings.

\section{W[1]-hardness of Max-SRTI for the Parameter Tree-cut Width}

In this section, we show the W[1]-hardness of \textsc{Max-SRTI} for the parameter tree-cut width. We will reduce from \textsc{Clique}, which is defined as follows.
\defProblemTask{\textsc{Clique}}
{A graph $G$ and a natural number $k$.}
{Find a clique of size $k$, i.e., a set~$\{v_1, \dots, v_k\}$ of vertices such that for each $i, j\in [k]$ we have $\{v_i, v_j\} \in E(G)$, or decide that no such set exists.}
Thus, let $(G, k)$ be an instance of \textsc{Clique}.
The basic idea of the reduction is as follows.
First, the vertices of $V(G) = \{v_1, \dots, v_n\}$ are ordered in an arbitrary way.
The parameterized reduction contains several gadgets:
$k$ vertex-selection gadgets which ``select'' the vertices which are contained in the clique, and one edge gadget for each edge of $G$ and each $i\neq j \in [k]$.
Furthermore, there is one incidence vertex~$c_{i,j}$ for each $i \neq j \in [k]$.
Every vertex-selection gadget only one vertex~$c_i$ which is connected to vertices outside the edge gadget.
The larger a stable matching is inside the vertex-selection gadget, the worse is vertex~$c_i$ matched.
Every incidence vertex $c_{ij}$ now has to be matched to the edge gadget of an edge incident to the vertex selected by the $i$-th vertex-selection gadget in any sufficiently large stable matching.
To ensure this, we need to model parallel edges by gadgets which we call ``parallel-edges gadget''.
If now $c_{ij} $ and $c_{ji}$ match to the same edge gadget, then this increases the size of a maximum stable matching by one.
Thus, a maximum stable matching allows us to decide whether $G$ contains a clique of size $k$.

We will describe the gadgets used in the reduction first, then the reduction, and finally prove its correctness.

\subsection{The Gadgets}

As in the W[1]-hardness proof for treedepth and disjoint paths modulator number, the reduction contains vertex-selection gadgets and edge gadgets.
However, they work differently than those for treedepth and disjoint path modulator number.
Consistency gadgets are not required, but there will be ``incidence vertices'' and a ``parallel-edges gadget''.
We start describing the reduction by describing the gadgets, starting with the parallel-edges gadget.

\subsubsection{The Parallel-edges Gadget}
  Our reduction will use parallel edges; however, parallel edges are not allowed in the \textsc{Max-SRTI}-problem because it is defined on simple graphs.
  We will use a result of Cechl\'{a}rov\'{a} and Fleiner \cite{Cechlarova2005}, showing that parallel edges can modelled in an \textsc{Max-SRTI}-instance:

  \begin{lemma}\label{lpeg}
    One can model parallel edges in an \textsc{(Max)-SR(T)I}-instance by 6-cycles with two outgoing edges (see Theorem 2.1 from \cite{Cechlarova2005}).
  \end{lemma}

  The gadget one can use to replace such an edge is depicted in \Cref{fpeg}; we call it a parallel-edges gadget.

  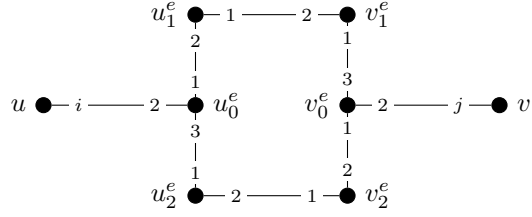
\begin{figure}[bt]
    \begin{center}
      \begin{tikzpicture}[xscale =2 , yscale = 1.2]
        \node[vertex, label=180:$u$] (u) at (0, 0) {};
        \node[vertex, label=0:$v$] (v) at (3, 0) {};

        \node[vertex, label=0:$u_0^e$] (u0) at (1, 0) {};
        \node[vertex, label=180:$v_0^e$] (v0) at (2, 0) {};
        \node[vertex, label=180:$u_1^e$] (u1) at (1, 1) {};
        \node[vertex, label=0:$v_1^e$] (v1) at (2, 1) {};
        \node[vertex, label=180:$u_2^e$] (u2) at (1, -1) {};
        \node[vertex, label=0:$v_2^e$] (v2) at (2, -1) {};

        \draw (u) edge node[pos=0.2, fill=white, inner sep=2pt] {\scriptsize $i$}  node[pos=0.76, fill=white, inner sep=2pt] {\scriptsize $2$} (u0);
        \draw (u0) edge node[pos=0.2, fill=white, inner sep=2pt] {\scriptsize $1$}  node[pos=0.76, fill=white, inner sep=2pt] {\scriptsize $2$} (u1);
        \draw (u1) edge node[pos=0.2, fill=white, inner sep=2pt] {\scriptsize $1$}  node[pos=0.76, fill=white, inner sep=2pt] {\scriptsize $2$} (v1);
        \draw (v1) edge node[pos=0.2, fill=white, inner sep=2pt] {\scriptsize $1$}  node[pos=0.76, fill=white, inner sep=2pt] {\scriptsize $3$} (v0);
        \draw (v0) edge node[pos=0.2, fill=white, inner sep=2pt] {\scriptsize $1$}  node[pos=0.76, fill=white, inner sep=2pt] {\scriptsize $2$} (v2);
        \draw (v2) edge node[pos=0.2, fill=white, inner sep=2pt] {\scriptsize $1$}  node[pos=0.76, fill=white, inner sep=2pt] {\scriptsize $2$} (u2);
        \draw (u2) edge node[pos=0.2, fill=white, inner sep=2pt] {\scriptsize $1$}  node[pos=0.76, fill=white, inner sep=2pt] {\scriptsize $3$} (u0);
        \draw (v0) edge node[pos=0.2, fill=white, inner sep=2pt] {\scriptsize $2$}  node[pos=0.76, fill=white, inner sep=2pt] {\scriptsize $j$} (v);
      \end{tikzpicture}

    \end{center}
    \caption{How parallel edges can be modeled \cite[Theorem 2.1]{Cechlarova2005}: any edge $e=\{u, v\}$ can be replaced by the above gadget, where $u$ ranks $u_0^e$ at the position on which it ranked $v$ in the original instance (position $i$ in this example), and $v$ ranks $v_0^e$ at the position on which it ranked $w$ in the original instance (position $j$ in this example).}\label{fpeg}
  \end{figure}

  By \Cref{lpeg}, we can use parallel edges, and we will do so from now on.

\subsubsection{The Vertex-selection Gadget}

We now describe the vertex-selection gadget.
The gadget has a special vertex $c_i$, which is the only vertex which has connections outside the vertex-selection gadget.
The vertex-selection gadget also contains a vertex-gadget for each vertex $u\in V(G)$.
In any stable matching, $c_i$~has to be connected to one of the vertex-gadgets.
The worse $c_i$ is matched, the more edges will be contained in the vertex-gadget.

We first describe the vertex-gadget.

 \begin{lemma}\label{lvg2}
   For any $j, \ell\in \mathbb{N}$ and any vertex $c_i$, one can construct a gadget with one outgoing edge $\{c_i, w\}$ to $c_i$ with $\rk_{c_i} (w) = \ell$ not contained in the gadget such that the maximum size of a stable matching is
   \begin{enumerate}[label=(\roman*)]
     \item $j+2$ if $\{c_i, w\}$ is not contained in the matching; in this case, $\{c_i, w\}$ is no blocking pair even if $c_i$ is unmatched and

     \item $2j+2$ otherwise (i.e. if $\{c_i, w\}$ is contained in the matching), counting also the edge~$\{c_i, w\}$.
   \end{enumerate}

   Furthermore, the gadget can be constructed in~$O(j)$ time and has tree-cut width at most four.

 \end{lemma}

 We call such a gadget a $(j, c_i, \ell)$-vertex gadget.

 \begin{figure}
   \begin{center}
     \begin{tikzpicture}[xscale = 1.8, yscale = 1.5]
        \draw[dotted] (1.23, 0.5) circle (1.13);
        \draw[dotted] (4., 1.6) ellipse (2 and 0.5);
        \draw[dotted] (5, 0.6) ellipse (2 and 0.5);
        \draw[dotted] (4., -1.6) ellipse (2 and 0.5);
        \node[vertex, label=270:$c_i$] (v) at (0, 00) {};
        \node[vertex, label=270:$w$] (w) at (1, 0) {};
        \node[vertex, label=90:$w'$] (w1) at (0.5, 1) {};
        \node[vertex, label=270:$c$] (c) at (2, 0) {};
        \node[vertex, label=90:$w''$] (w2) at (1.5, 1) {};
        \node[vertex, label=90:$p_{11}$] (p11) at (2.5, 1.5) {};
        \node[vertex, label=90:$p_{12}$] (p12) at (3.5, 1.5) {};
        \node[vertex, label=90:$p_{13}$] (p13) at (4.5, 1.5) {};
        \node[vertex, label=90:$p_{14}$] (p14) at (5.5, 1.5) {};
        \node[vertex, label=90:$p_{21}$] (p21) at (3.5, 0.5) {};
        \node[vertex, label=90:$p_{22}$] (p22) at (4.5, 0.5) {};
        \node[vertex, label=90:$p_{23}$] (p23) at (5.5, 0.5) {};
        \node[vertex, label=90:$p_{24}$] (p24) at (6.5, 0.5) {};
        \node[vertex, label=270:$p_{j1}$] (p41) at (2.5, -1.5) {};
        \node[vertex, label=270:$p_{j2}$] (p42) at (3.5, -1.5) {};
        \node[vertex, label=270:$p_{j3}$] (p43) at (4.5, -1.5) {};
        \node[vertex, label=270:$p_{j4}$] (p44) at (5.5, -1.5) {};

        \draw (v) edge node[pos=0.2, fill=white, inner sep=1pt] {\scriptsize $\ell$}  node[pos=0.76, fill=white, inner sep=1pt] {\scriptsize $1$} (w);
        \draw (w) edge node[pos=0.2, fill=white, inner sep=1pt] {\scriptsize $1$}  node[pos=0.76, fill=white, inner sep=1pt] {\scriptsize $2$} (c);
        \draw (c) edge node[pos=0.2, fill=white, inner sep=1pt] {\scriptsize $1$}  node[pos=0.76, fill=white, inner sep=1pt] {\scriptsize $2$} (p12);
        \draw (p12) edge node[pos=0.2, fill=white, inner sep=1pt] {\scriptsize $1$}  node[pos=0.76, fill=white, inner sep=1pt] {\scriptsize $1$} (p13);
        \draw (p13) edge node[pos=0.2, fill=white, inner sep=1pt] {\scriptsize $1$}  node[pos=0.76, fill=white, inner sep=1pt] {\scriptsize $1$} (p14);
        \draw (c) edge node[pos=0.2, fill=white, inner sep=1pt] {\scriptsize $1$}  node[pos=0.76, fill=white, inner sep=1pt] {\scriptsize $2$} (p22);
        \draw (p22) edge node[pos=0.2, fill=white, inner sep=1pt] {\scriptsize $1$}  node[pos=0.76, fill=white, inner sep=1pt] {\scriptsize $1$} (p23);
        \draw (p23) edge node[pos=0.2, fill=white, inner sep=1pt] {\scriptsize $1$}  node[pos=0.76, fill=white, inner sep=1pt] {\scriptsize $1$} (p24);
        \draw (c) edge node[pos=0.2, fill=white, inner sep=1pt] {\scriptsize $1$}  node[pos=0.76, fill=white, inner sep=1pt] {\scriptsize $2$} (p42);
        \draw (p42) edge node[pos=0.2, fill=white, inner sep=1pt] {\scriptsize $1$}  node[pos=0.76, fill=white, inner sep=1pt] {\scriptsize $1$} (p43);
        \draw (p43) edge node[pos=0.2, fill=white, inner sep=1pt] {\scriptsize $1$}  node[pos=0.76, fill=white, inner sep=1pt] {\scriptsize $1$} (p44);
        \draw (w) edge node[pos=0.2, fill=white, inner sep=1pt] {\scriptsize $3$}  node[pos=0.76, fill=white, inner sep=1pt] {\scriptsize $1$} (w1);
        \draw (w1) edge node[pos=0.2, fill=white, inner sep=1pt] {\scriptsize $2$}  node[pos=0.76, fill=white, inner sep=1pt] {\scriptsize $1$} (w2);
        \draw (w2) edge node[pos=0.2, fill=white, inner sep=1pt] {\scriptsize $2$}  node[pos=0.76, fill=white, inner sep=1pt] {\scriptsize $2$} (w);
        \draw (p11) edge node[pos=0.2, fill=white, inner sep=1pt] {\scriptsize $1$}  node[pos=0.76, fill=white, inner sep=1pt] {\scriptsize $3$} (p12);
        \draw (p21) edge node[pos=0.2, fill=white, inner sep=1pt] {\scriptsize $1$}  node[pos=0.76, fill=white, inner sep=1pt] {\scriptsize $3$} (p22);
        \draw (p41) edge node[pos=0.2, fill=white, inner sep=1pt] {\scriptsize $1$}  node[pos=0.76, fill=white, inner sep=1pt] {\scriptsize $3$} (p42);

        \node[draw, circle, fill, inner sep = 0.8] (dot1) at (4, -0.2) {};
        \node[draw, circle, fill, inner sep = 0.8] (dot2) at (4, -0.5) {};
        \node[draw, circle, fill, inner sep = 0.8] (dot3) at (4, -0.8) {};

        \begin{scope}[on background layer]
          \newcommand{\colorBetweenTwoNodes}[3]{
            \fill[#1] ($(#2) + (0, .08)$) to ($(#2) - (0, .08)$) to ($(#3) - (0,.08)$) to ($(#3) + (0,.08)$) -- cycle;
          }
          \colorBetweenTwoNodes{mygreen}{v}{w}
          \colorBetweenTwoNodes{mygreen}{c}{p12}
          \colorBetweenTwoNodes{mygreen}{p13}{p14}
          \colorBetweenTwoNodes{mygreen}{p21}{p22}
          \colorBetweenTwoNodes{mygreen}{p23}{p24}
          \colorBetweenTwoNodes{mygreen}{p41}{p42}
          \colorBetweenTwoNodes{mygreen}{p43}{p44}

          \colorBetweenTwoNodes{myyellow}{w}{c}
          \colorBetweenTwoNodes{myyellow}{p12}{p13}
          \colorBetweenTwoNodes{myyellow}{p22}{p23}
          \colorBetweenTwoNodes{myyellow}{p42}{p43}

          \colorBetweenTwoNodes{mylila}{w1}{w2}
        \end{scope}
     \end{tikzpicture}

   \end{center}
    \caption{A $(j, c_i, \ell)$-vertex gadget.
    The vertex $c_i$ ranks $w$ at position $\ell$.
    The green edges (including the violet edge) form a stable matching, and so do the yellow edges (also containing the violet edge). The dotted ellipses indicate the near-partition of a tree-cut decomposition of width 4.}\label{fvgtcw}
 \end{figure}
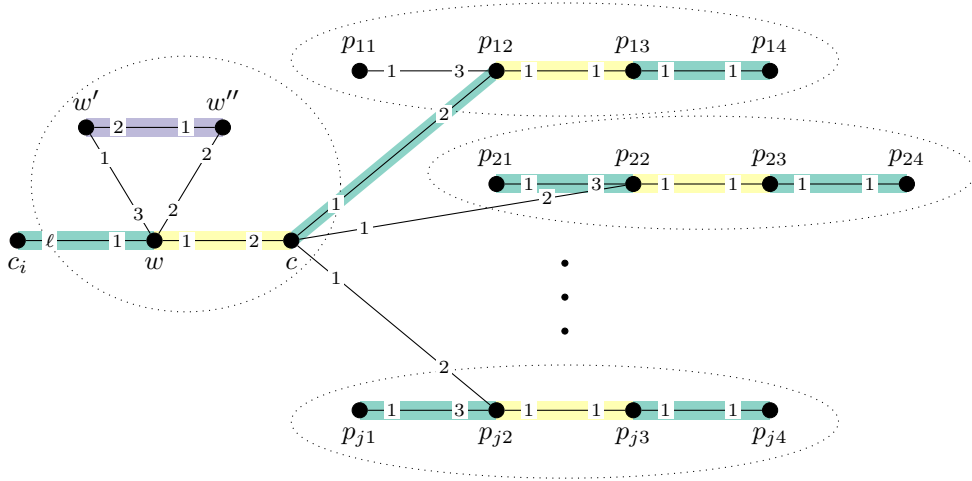

 \begin{proof}
  The gadget can be seen in \Cref{fvgtcw}.
  It contains a 3-cycle containing $w$ and two other vertices $w'$ and $w''$, and a vertex~$c$ connected to $w$.
  The vertex $c$ again is connected to the vertex $p_{q2} $ of the 3-path $p_{q1},p_{q2},p_{q3},p_{q4}$ for $q\in [j]$.
  The preferences are shown in \Cref{fvgtcw}.

  Let $M$ be a maximum stable matching.

  First, note that avoiding a blocking pair in the 3-cycle forces vertex $w$ to be matched to~$c$ or $c_i$, and that the vertices $w'$ and $w''$ have to be matched together, i.e., $\{w', w''\}\in M$.

   \begin{enumerate}[label=(\roman*)]
    \item By the above observation, the matching must contain the edge $\{w, c\}$.
    As this is one of $w$'s top choices, $\{c_i, w\}$ is not a blocking pair.
    To avoid a blocking pair $\{c, p_{i2}\}$ for $1\le i \le j$, we have to take edge $\{p_{i2}, p_{i3}\}$ for all $1\le i\le j$, yielding in total $j+2$ edges (the yellow and violet edges in \Cref{fvgtcw}).
    As this matching is maximal, no stable matching not containing $\{c_i, w\}$ can contain more edges.

    \item If $\{c_i, w\}$ is contained in the matching, then $M$ can contain the edges $\{p_{i1}, p_{i2}\}$ and~$\{p_{i3}, p_{i4}\}$ for $i\ge 2$ and the edges $\{c, p_{12}\}$ and $\{p_{13}, p_{14}\}$ (the green and violet edges in \Cref{fvgtcw}).
    One easily checks that this results in a stable matching leaving only $p_{11}$ unmatched, and thus, no matching can contain more edges. Therefore, $M$ contains $2j+2$ edges.
   \end{enumerate}

   The size of the gadget is obviously linear in $j$. It remains to show that the tree-cut width is at most four. To see this, we give a tree-cut decomposition. The tree is a star with~$j$~leaves, where the center contains $\{w, w', w'', c\}$, while the $j$-th leaf contains the vertices $\{p_{q1}, p_{q2}, p_{q3}, p_{q4}\}$ for $q\in [j]$ (see \Cref{fvgtcw}).
   Thus, we have $\adh (t) \le 1$ for all $t\in V(T)$.
   Furthermore, the 3-center of each torso is just the bag itself. Thus, the tree-cut width is at most four.
 \end{proof}

 From this, we construct a vertex-selection gadget as follows.
 For each $j\in [n]$, we add a $(2j-1, c_i, C+ j(k-1))$-vertex gadget for a sufficiently large number $C$, depending polynomially on~$n$, ensuring that $c_i$ has to match to one of those vertex gadgets.
 More precisely, let $C:=(14 + 6n) m k (k-1) + 8 k (k-1) (n -1)$.

 It is easy to see that the tree-cut width of a vertex-selection gadget is at most four since one can construct an according tree-cut decomposition by taking a bag $\{c_i\}$ as the root, and adding for each vertex-gadget a tree-cut decomposition of width four, where the root of this tree-cut decomposition is a child of $\{c_i\}$.

 Next, we turn to the edge gadgets:

 \subsubsection{The Edge Gadget}

 The idea behind the edge gadgets is that if the vertex $c_i$ of a vertex-selection gadget is ranked ``bad'', then there are so-called incidence vertices $c_{ij}$ which have to be ranked ``good'' in order to avoid a blocking pair $\{c_i, c_{ij}\}$, and thus, a stable matching contains ``few'' edges inside this edge gadget, compensating the fact that the matching contains ``many'' edges inside the vertex-selection gadget.
 Finally, if $c_{ij}$ and $c_{ji}$ do not match to the same edge gadget, then this matching contains one edge less than if they would match to the same gadget.

 We now prove that the edge gadget has the desired properties.

 \begin{lemma}\label{leg2}
   For any $k_1, k_2\in \mathbb{N}$, we can construct a gadget with tree-cut width at most 10 with two outgoing edges $e_1 =\{v^1, w^1\}$ and $e_2 = \{v^2, w^2\}$ such that

   \begin{enumerate}[label=(\roman*)]
    \item if $e_1$ and $e_2$ are contained in a maximum stable matching $M$, then $M$ contains $6+2k_1+2k_2$ edges with at least one endpoint in the gadget,

    \item if $e_i$ is contained in a maximum stable matching $M$ but $e_{3-i}$ is not, then $M$ contains $5+2k_i+k_{3-i}$ edges with at least one endpoint in the gadget, and

    \item if neither $e_1 $ nor $e_2$ is contained in a maximum stable matching $M$, then $M$ contains $5+k_1+k_2$ edges with at least one endpoint in the gadget.
   \end{enumerate}

   Furthermore, the gadget can be constructed in $O(k_1+k_2)$ time, and $\{v^i, w^i\}$ is not a blocking pair in any of the three cases, even if $v^i$ is unmatched.
 \end{lemma}

  We call such a gadget a $(v^1, j_1, k_1, v^2, j_2, k_2)$-edge gadget. The vertex $v^1$ (which is not part of the gadget) is called \emph{left} neighbor, while vertex $v^2$ (which is also not part of the gadget) is called \emph{right} neighbor.

  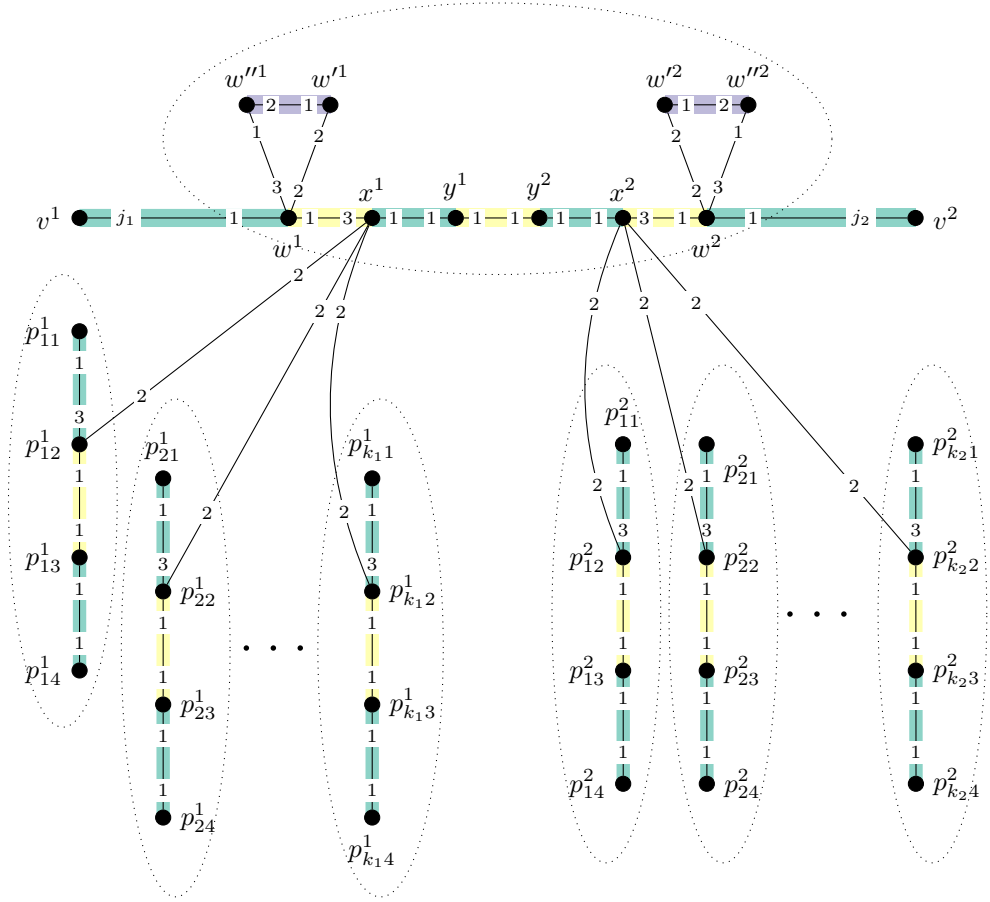
\begin{figure}[bt]
    \begin{center}
      \begin{tikzpicture}[xscale = 1.1, yscale = 1.5]
        \draw[dotted] (3.5, -0.3) ellipse (4 and 1.2);
        \node[vertex, label=180:$v^1$] (v2) at (-1.5, -1) {};
        \node[vertex, label=270:$w^1$] (v3) at (1, -1) {};
        \node[vertex] (v1) at (v3) {};
        \node[vertex, label=90:$x^1$] (v4) at (2, -1) {};
        \node[vertex, label=90:$y^1$] (v5) at (3, -1) {};
        \node[vertex, label=90:$y^2$] (v6) at (4, -1) {};
        \node[vertex, label=90:$x^2$] (v7) at (5, -1) {};
        \node[vertex, label=270:$w^2$] (v8) at (6, -1) {};
        \node[vertex, label=0:$v^2$] (v9) at (8.5, -1) {};
        \node[vertex] (cji) at (v8) {};
        \node[vertex, label=180:$p_{14}^1$] (p31) at ($(v4) + (-3.5, -4)$) {};
        \node[vertex, label=180:$p_{13}^1$] (p32) at ($(p31) + (0, 1)$) {};
        \node[vertex, label=180:$p_{12}^1$] (p33) at ($(p32) + (0, 1)$) {};
        \node[vertex, label=180:$p_{11}^1$] (p34) at ($(p33) + (0, 1)$) {};
        \draw[dotted] ($(p31) + (-0.2, 1.5)$) ellipse (0.65 and 2.);

        \node[vertex, label=0:$p_{24}^1$] (pl31) at ($(v4) + (-2.5, -5.3)$) {};
        \node[vertex, label=0:$p_{23}^1$] (pl32) at ($(pl31) + (0, 1)$) {};
        \node[vertex, label=0:$p_{22}^1$] (pl33) at ($(pl32) + (0, 1)$) {};
        \node[vertex, label=90:$p_{21}^1$] (pl34) at ($(pl33) + (0, 1)$) {};
        \draw[dotted] ($(pl31) + (0.15, 1.5)$) ellipse (0.65 and 2.2);

        \node[vertex, label=270:$p_{k_14}^1$] (p1) at ($(v4) + (0, -5.3)$) {};
        \node[vertex, label=0:$p_{k_13}^1$] (p2) at ($(p1) + (0, 1)$) {};
        \node[vertex, label=0:$p_{k_12}^1$] (p3) at ($(p2) + (0, 1)$) {};
        \node[vertex, label=90:$p_{k_11}^1$] (p4) at ($(p3) + (0, 1)$) {};
        \draw[dotted] ($(p1) + (0.1, 1.5)$) ellipse (0.75 and 2.2);
        
        \draw (p1) edge node[pos=0.2, fill=white, inner sep=2pt] {\scriptsize $1$}  node[pos=0.76, fill=white, inner sep=2pt] {\scriptsize $1$} (p2);
        \draw (p2) edge node[pos=0.2, fill=white, inner sep=2pt] {\scriptsize $1$}  node[pos=0.76, fill=white, inner sep=2pt] {\scriptsize $1$} (p3);
        \draw (p3) edge node[pos=0.2, fill=white, inner sep=2pt] {\scriptsize $3$}  node[pos=0.76, fill=white, inner sep=2pt] {\scriptsize $1$} (p4);
        \draw (p3) edge[bend left] node[pos=0.2, fill=white, inner sep=2pt] {\scriptsize $2$}  node[pos=0.76, fill=white, inner sep=2pt] {\scriptsize $2$} (v4);

        \node[vertex, label=90:$p_{11}^2$] (p61) at ($(v7) + (-0, -2)$) {};
        \node[vertex, label=180:$p_{12}^2$] (p62) at ($(p61) + (0, -1)$) {};
        \node[vertex, label=180:$p_{13}^2$] (p63) at ($(p62) + (0, -1)$) {};
        \node[vertex, label=180:$p_{14}^2$] (p64) at ($(p63) + (0, -1)$) {};
        \draw[dotted] ($(p61) + (-0.2, -1.5)$) ellipse (0.65 and 2.2);

        \node[vertex, label=350:$p_{21}^2$] (p261) at ($(v7) + (1, -2)$) {};
        \node[vertex, label=0:$p_{22}^2$] (p262) at ($(p261) + (0, -1)$) {};
        \node[vertex, label=0:$p_{23}^2$] (p263) at ($(p262) + (0, -1)$) {};
        \node[vertex, label=0:$p_{24}^2$] (p264) at ($(p263) + (0, -1)$) {};
        \draw[dotted] ($(p261) + (0.2, -1.5)$) ellipse (0.65 and 2.2);

        \node[vertex, label=0:$p_{k_21}^2$] (p361) at ($(v7) + (3.5, -2)$) {};
        \node[vertex, label=0:$p_{k_22}^2$] (p362) at ($(p361) + (0, -1)$) {};
        \node[vertex, label=0:$p_{k_23}^2$] (p363) at ($(p362) + (0, -1)$) {};
        \node[vertex, label=0:$p_{k_24}^2$] (p364) at ($(p363) + (0, -1)$) {};
        \draw[dotted] ($(p361) + (0.2, -1.5)$) ellipse (0.65 and 2.2);
        
        \draw (p31) edge node[pos=0.2, fill=white, inner sep=2pt] {\scriptsize $1$}  node[pos=0.76, fill=white, inner sep=2pt] {\scriptsize $1$} (p32);
        \draw (p32) edge node[pos=0.2, fill=white, inner sep=2pt] {\scriptsize $1$}  node[pos=0.76, fill=white, inner sep=2pt] {\scriptsize $1$} (p33);
        \draw (p33) edge node[pos=0.2, fill=white, inner sep=2pt] {\scriptsize $3$}  node[pos=0.76, fill=white, inner sep=2pt] {\scriptsize $1$} (p34);
        \draw (p33) edge node[pos=0.2, fill=white, inner sep=1pt] {\scriptsize $2$}  node[pos=0.76, fill=white, inner sep=1pt] {\scriptsize $2$} (v4);
        \draw (pl31) edge node[pos=0.2, fill=white, inner sep=2pt] {\scriptsize $1$}  node[pos=0.76, fill=white, inner sep=2pt] {\scriptsize $1$} (pl32);
        \draw (pl32) edge node[pos=0.2, fill=white, inner sep=2pt] {\scriptsize $1$}  node[pos=0.76, fill=white, inner sep=2pt] {\scriptsize $1$} (pl33);
        \draw (pl33) edge node[pos=0.2, fill=white, inner sep=2pt] {\scriptsize $3$}  node[pos=0.76, fill=white, inner sep=2pt] {\scriptsize $1$} (pl34);
        \draw (pl33) edge node[pos=0.2, fill=white, inner sep=2pt] {\scriptsize $2$}  node[pos=0.76, fill=white, inner sep=2pt] {\scriptsize $2$} (v4);

        \draw (p63) edge node[pos=0.2, fill=white, inner sep=2pt] {\scriptsize $1$}  node[pos=0.76, fill=white, inner sep=2pt] {\scriptsize $1$} (p64);
        \draw (p63) edge node[pos=0.2, fill=white, inner sep=2pt] {\scriptsize $1$}  node[pos=0.76, fill=white, inner sep=2pt] {\scriptsize $1$} (p62);
        \draw (p62) edge node[pos=0.2, fill=white, inner sep=2pt] {\scriptsize $3$}  node[pos=0.76, fill=white, inner sep=2pt] {\scriptsize $1$} (p61);
        \draw (p62) edge[bend left] node[pos=0.2, fill=white, inner sep=2pt] {\scriptsize $2$}  node[pos=0.76, fill=white, inner sep=2pt] {\scriptsize $2$} (v7);

        \draw (p263) edge node[pos=0.2, fill=white, inner sep=2pt] {\scriptsize $1$}  node[pos=0.76, fill=white, inner sep=2pt] {\scriptsize $1$} (p264);
        \draw (p263) edge node[pos=0.2, fill=white, inner sep=2pt] {\scriptsize $1$}  node[pos=0.76, fill=white, inner sep=2pt] {\scriptsize $1$} (p262);
        \draw (p262) edge node[pos=0.2, fill=white, inner sep=2pt] {\scriptsize $3$}  node[pos=0.76, fill=white, inner sep=2pt] {\scriptsize $1$} (p261);
        \draw (p262) edge node[pos=0.2, fill=white, inner sep=2pt] {\scriptsize $2$}  node[pos=0.76, fill=white, inner sep=2pt] {\scriptsize $2$} (v7);

        \draw (p363) edge node[pos=0.2, fill=white, inner sep=2pt] {\scriptsize $1$}  node[pos=0.76, fill=white, inner sep=2pt] {\scriptsize $1$} (p364);
        \draw (p363) edge node[pos=0.2, fill=white, inner sep=2pt] {\scriptsize $1$}  node[pos=0.76, fill=white, inner sep=2pt] {\scriptsize $1$} (p362);
        \draw (p362) edge node[pos=0.2, fill=white, inner sep=2pt] {\scriptsize $3$}  node[pos=0.76, fill=white, inner sep=2pt] {\scriptsize $1$} (p361);
        \draw (p362) edge node[pos=0.2, fill=white, inner sep=2pt] {\scriptsize $2$}  node[pos=0.76, fill=white, inner sep=2pt] {\scriptsize $2$} (v7);

        \node[vertex, label=90:$w^{\prime 1}$] (v11) at (1.5, 0) {};
        \node[vertex, label=90:$w^{\prime \prime 1}$] (v12) at (0.5, 0) {};
        \node[vertex, label=90:$w^{\prime 2}$] (v61) at (5.5, 0) {};
        \node[vertex, label=90:$w^{\prime \prime 2}$] (v62) at (6.5,0) {};

        \draw (v1) edge node[pos=0.2, fill=white, inner sep=1pt] {\scriptsize $2$}  node[pos=0.76, fill=white, inner sep=1pt] {\scriptsize $2$} (v11);
        \draw (v11) edge node[pos=0.2, fill=white, inner sep=1pt] {\scriptsize $1$}  node[pos=0.76, fill=white, inner sep=1pt] {\scriptsize $2$} (v12);
        \draw (v12) edge node[pos=0.2, fill=white, inner sep=1pt] {\scriptsize $1$}  node[pos=0.76, fill=white, inner sep=1pt] {\scriptsize $3$} (v1);
        \draw (cji) edge node[pos=0.2, fill=white, inner sep=1pt] {\scriptsize $2$}  node[pos=0.76, fill=white, inner sep=1pt] {\scriptsize $2$} (v61);
        \draw (v61) edge node[pos=0.2, fill=white, inner sep=1pt] {\scriptsize $1$}  node[pos=0.76, fill=white, inner sep=1pt] {\scriptsize $2$} (v62);
        \draw (v62) edge node[pos=0.2, fill=white, inner sep=1pt] {\scriptsize $1$}  node[pos=0.76, fill=white, inner sep=1pt] {\scriptsize $3$} (cji);
        \draw (v2) edge node[pos=0.2, fill=white, inner sep=2pt] {\scriptsize $j_1$}  node[pos=0.76, fill=white, inner sep=1pt] {\scriptsize $1$} (v1);
        \draw (v3) edge node[pos=0.2, fill=white, inner sep=1pt] {\scriptsize $1$}  node[pos=0.76, fill=white, inner sep=1pt] {\scriptsize $3$} (v4);
        \draw (v4) edge node[pos=0.2, fill=white, inner sep=1pt] {\scriptsize $1$}  node[pos=0.76, fill=white, inner sep=1pt] {\scriptsize $1$} (v5);
        \draw (v5) edge node[pos=0.2, fill=white, inner sep=1pt] {\scriptsize $1$}  node[pos=0.76, fill=white, inner sep=1pt] {\scriptsize $1$} (v6);
        \draw (v6) edge node[pos=0.2, fill=white, inner sep=1pt] {\scriptsize $1$}  node[pos=0.76, fill=white, inner sep=1pt] {\scriptsize $1$} (v7);
        \draw (v7) edge node[pos=0.2, fill=white, inner sep=1pt] {\scriptsize $3$}  node[pos=0.76, fill=white, inner sep=1pt] {\scriptsize $1$} (v8);
        \draw (v8) edge node[pos=0.2, fill=white, inner sep=1pt] {\scriptsize $1$}  node[pos=0.76, fill=white, inner sep=2pt] {\scriptsize $j_2$} (v9);

        \begin{scope}[on background layer]
          \newcommand{\colorBetweenTwoNodes}[3]{
            \fill[#1] ($(#2) + (0, .08)$) to ($(#2) - (0, .08)$) to ($(#3) - (0,.08)$) to ($(#3) + (0,.08)$) -- cycle;
          }
          \colorBetweenTwoNodes{mygreen}{v1}{v2}
          \colorBetweenTwoNodes{mygreen}{v5}{v4}
          \colorBetweenTwoNodes{mygreen}{v6}{v7}
          \colorBetweenTwoNodes{mygreen}{v8}{v9}

          \colorBetweenTwoNodes{mylila}{v11}{v12}
          \colorBetweenTwoNodes{mylila}{v61}{v62}

          \colorBetweenTwoNodes{myyellow}{v4}{v3}
          \colorBetweenTwoNodes{myyellow}{v6}{v5}
          \colorBetweenTwoNodes{myyellow}{v7}{v8}
        \end{scope}
        \begin{scope}[on background layer]
          \newcommand{\colorBetweenTwoNodes}[3]{
            \fill[#1] ($(#2) + (0.08, .0)$) to ($(#2) - (0.08, .0)$) to ($(#3) - (0.08,.0)$) to ($(#3) + (0.08,.0)$) -- cycle;
          }
          \colorBetweenTwoNodes{mygreen}{p31}{p32}
          \colorBetweenTwoNodes{mygreen}{p33}{p34}
          \colorBetweenTwoNodes{mygreen}{pl31}{pl32}
          \colorBetweenTwoNodes{mygreen}{pl33}{pl34}
          \colorBetweenTwoNodes{mygreen}{p61}{p62}
          \colorBetweenTwoNodes{mygreen}{p63}{p64}
          \colorBetweenTwoNodes{mygreen}{p261}{p262}
          \colorBetweenTwoNodes{mygreen}{p263}{p264}
          \colorBetweenTwoNodes{mygreen}{p361}{p362}
          \colorBetweenTwoNodes{mygreen}{p363}{p364}

          \colorBetweenTwoNodes{myyellow}{v3}{v4}
          \colorBetweenTwoNodes{myyellow}{v5}{v6}
          \colorBetweenTwoNodes{myyellow}{v7}{v8}
          \colorBetweenTwoNodes{myyellow}{p62}{p63}
          \colorBetweenTwoNodes{myyellow}{p262}{p263}
          \colorBetweenTwoNodes{myyellow}{p362}{p363}
          \colorBetweenTwoNodes{myyellow}{pl32}{pl33}
          \colorBetweenTwoNodes{myyellow}{p32}{p33}
          
          \colorBetweenTwoNodes{mygreen}{p1}{p2}
          \colorBetweenTwoNodes{mygreen}{p3}{p4}
          \colorBetweenTwoNodes{myyellow}{p2}{p3}
        \end{scope}
        
        \node (dots) at ($0.25*(pl32)+0.25*(pl33)+0.25*(p2)+0.25*(p3) + (0.13,0)$) {\Huge\dots};
        
        \node (dots2) at ($0.25*(p262)+0.25*(p263)+0.25*(p362)+0.25*(p363) + (0.13,0)$) {\Huge\dots};
      \end{tikzpicture}
      \caption{A $(v^1, j_1, k_1,v^2, j_2, k_2)$-edge gadget.
      The dotted ellipses describe a near-partition of the gadget. The green edges together with the violet edges and the yellow edges together with the violet edges form a stable matching.}
    \label{fedgegadgettcw}
    \end{center}
    \end{figure}

 \begin{proof}
  A $(v^1, j_1, k_1, v^2, j_2, k_2)$-edge gadget is depicted in \Cref{fedgegadgettcw}.
  The gadget consists of a path $w^1, x^1, y^1, y^2, x^2, w^2$ of length five, where the end vertices $w^1$ and $w^2$ of the path are connected to the outgoing edges $e_1$ and $e_2$.
  For each $w^i$, there exists also two vertices~$w'^i$ and $w''^i$ which form a triangle together with $w^i$.
  The vertex $x^i$ is connected to $k_i$ paths $p^i_{j1}, p^i_{j2},p^i_{j3},p^i_{j4}$ of length three.
  More precisely, $x^i$ is connected to $p^i_{j2}$ for $1\le j \le k_i$.
  The preferences can be seen in \Cref{fedgegadgettcw}.
  Obviously, the gadget can be constructed in time linear in $k_1 + k_2$.

  Let $M^*$ be a maximum stable matching.
  If $M^*(w^i)\notin \{v^i, x^i\}$, then one of the vertices~$w^i$, $w'^i$, and $w''^i$ is unmatched in $M^*$. It is easy to see that this unmatched vertex forms a blocking pair with another vertex of this triangle, so we may assume that $w^i$ is matched to~$x^i$ or $v^i$ in $M^*$.

  We define a matching for each of the three cases as follows:
  \begin{enumerate}[label=(\roman*)]
    \item $
            M_1 := \{\{w^{\prime i}, w^{\prime\prime i}\}, \{v^i, w^i\}, \{x^i, y^i\}:i\in [2]\} \cup \{\{p^i_{j1}, p^i_{j2}\}, \{p^i_{j3}, p^i_{j4}\} : i\in [2], j\in [k_i]\}
          $ (green and violet edges in \Cref{fedgegadgettcw}).
    \item We assume without loss of generality that $e_1$ is contained in a maximum stable matching.
    \begin{align*}
      M_2 := &\{\{w^{\prime i}, w^{\prime\prime i}\}:i\in [2]\}\cup \{\{v^1, w^1\}, \{x^1, y^1\}, \{x^2, w^2\}\} \\
      &\cup \{\{p^1_{j1}, p^1_{j2}\}, \{p^1_{j3}, p^1_{j4}\} : j\in [k_1]\} \cup \{\{p^2_{j2}, p^2_{j3}\} : j \in [k_2]\}.
    \end{align*}
    \item $M_3:= \{\{w^{\prime i}, w^{\prime\prime i}\}, \{w^i, x^i\}: i\in [2]\} \cup \{\{y_1, y_2\}\} \cup \{\{p^i_{j2}, p^i_{j3}\}: i\in [2], j\in [k_i]\}$ (yellow and violet edges in \Cref{fedgegadgettcw}).
  \end{enumerate}

  One easily checks that $M_1$, $M_2$, and $M_3$ are stable.

  If $M^* (w^i) = u^i$, then $M^*$ must contain the edge $\{p^i_{j2}, p^i_{j3}\}$ for all $j\in [k_i]$ to avoid the blocking pair $\{u^i, p^i_{j2}\}$ (yellow edges in \Cref{fedgegadgettcw}), yielding in total $k_i$ edges.

  Otherwise, $M^* (w^i) = x^i$ holds.
  Thus, $M^*$ can contain the edges $\{p^i_{j1}, p^i_{j2}\}$ and $\{p^i_{j3}, p^i_{j4}\}$ for~$j\in [k_i]$ (green edges in \Cref{fedgegadgettcw}), yielding in total $2k_i$ edges.

  If $\{v^i, w^i\}\notin M$ for some $i$, then $M$ can contain at most three edges inside the (unique)~$v^1$-$v^2$-path inside the edge gadget, while otherwise $M$ can contain the four edges $\{v^i, w^i\}$, $\{x^i, y^i\}$ with $1\le i\le 2$.

  Thus, $M$ cannot contain more than
  \begin{itemize}
    \item $6+2k_1 + 2k_2$ edges with at least one endpoint in the gadget in case $(i)$,
    \item $5+2k_i+k_{3-i}$ edges with at least one endpoint in the gadget in case $(ii)$, and
    \item $5 + k_1 + k_2$ edges with at least one endpoint in the gadget in case $(iii)$.
  \end{itemize}

  As the sizes of $M_1$, $M_2$ and $M_3$ match these upper bounds, the upper bounds are tight.

  It remains to show that the resulting graph has bounded tree-cut width. This can be seen by considering the following tree-cut decomposition (see \Cref{fedgegadgettcw}): $T$ is a star, and the center corresponds to $\{w^i, w'^i, w''^i, x^i, y^i : i\in [2]\}$, while each 3-path is a leaf, showing that the tree-cut width is at most 10.
 \end{proof}

  \subsection{The Reduction}

  We are now ready to describe our reduction from \textsc{Clique}. Let $(G, k)$ be a \textsc{Clique}-instance. To simplify notation, we assume that $V(G) = [n]$.
  The \textsc{Max-SRTI}-instance contains $k$ vertices $c_1,\dots , c_k$, each of which is connected to $k-1$ incidence vertices $c_{ij}$ for $i\neq j$.

  Each vertex $c_i$ is connected to incidence vertex $c_{ij}$ by $n- 1$ parallel edges for each $j\in [k]$ (see \Cref{fvic}).
  They are ranked at positions $2l$ at $c_i$ and at positions $2(n-l)$ at $c_{ij}$ for~$1\le l \le n-1$.

  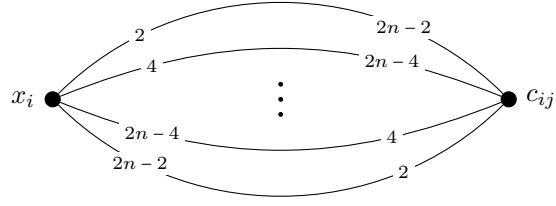
\begin{figure}[bt]
    \begin{center}
      \begin{tikzpicture}

        \node[vertex, label=180:$x_i$] (xi) at (-3, 0) {};
        \node[vertex, label=0:$c_{ij}$] (xij) at (3, 0) {};

        \draw[bend right = 45] (xi) edge node[pos=0.2, fill=white, inner sep=2pt] {\scriptsize $2n-2$}  node[pos=0.76, fill=white, inner sep=2pt] {\scriptsize $2$} (xij);
        \draw[bend right = 22.5] (xi) edge node[pos=0.2, fill=white, inner sep=2pt] {\scriptsize $2n-4$}  node[pos=0.76, fill=white, inner sep=2pt] {\scriptsize $4$} (xij);
        \draw[bend right = -22.5] (xi) edge node[pos=0.2, fill=white, inner sep=2pt] {\scriptsize $4$}  node[pos=0.76, fill=white, inner sep=2pt] {\scriptsize $2n-4$} (xij);
        \draw[bend right = -45] (xi) edge node[pos=0.2, fill=white, inner sep=2pt] {\scriptsize $2$}  node[pos=0.76, fill=white, inner sep=2pt] {\scriptsize $2n-2$} (xij);
        \node[draw, circle, fill, inner sep = 0.5] (dot1) at (0, 0.2) {};
        \node[draw, circle, fill, inner sep = 0.5] (dot2) at (0, 0) {};
        \node[draw, circle, fill, inner sep = 0.5] (dot3) at (0, -0.2) {};
      \end{tikzpicture}

    \end{center}
    \caption{The connection between a vertex $x_i$ and the incidence vertex $c_{ij}$.}\label{fvic}
  \end{figure}

  Two incidence vertices $c_{ij}$ and $c_{ji}$ with $i<j$ are connected by $2m$ edge gadgets (for each edge $e=(v, w)\in E(\overleftrightarrow{G})$, a $(c_{ij}, 2(n-v)+1, n-v, c_{ji}, 2(n-w)+1, n-w)$-edge gadget). This means that for any edge $e=(v, w)\in E(\overleftrightarrow{G})$ and each $1\le i < j \le k$, we add an edge gadget, where $c_{ij}$ ranks the edge gadget at position $2(n-v)+1$, and $c_{ji}$ ranks the edge gadget at position $2(n-w)+1$.

  For each $i\in [k] $ and $j\in [n]$, we add a $(2j-1, c_i, C+ j(k-1))$-vertex gadget for a sufficiently large number $C$, depending polynomially on $n$, ensuring that $c_i$ has to match to one of those vertex gadgets.
  More precisely, we let $C:=(14 + 6n) m k (k-1) + 8 k (k-1) (n -1)$, as this is an upper bound on the number of edges not contained in vertex gadgets (the first summand arises from the edge gadgets, and the second summand from the parallel-edges gadgets).

  We call the resulting graph $H$.

  The intuition behind the reduction is the following: As $C$ is sufficiently large, each $c_i$ matches to a vertex gadget, corresponding to a vertex $x_i$. This corresponds to selecting $x_i$ to be part of the clique.
  The incidence vertices $c_{ij}$ then have to match to edge gadgets.
  By \Cref{leg2}, it is better to match $c_{ij}$ and $c_{ji}$ to the same edge gadget than to different ones, so the matching is only large enough if $c_{ij}$ and $c_{ji}$ match to the same edge gadget. The parallel edges between $c_i$ and $c_{ij}$ ensure that $c_{ij}$ is matched at rank at most $2(n+1-x_i)$. Thus, $c_{ij}$ must match to an edge gadget corresponding to an edge $e=\{v, w\}$ with $v\le x$ (remember that $V(G)=[n]$, and thus there is a natural order on $V(G)$). If $v < x$, then the edge gadget contains less edges, and therefore in any stable matching which is large enough, we have $v = x$.

  To keep the calculations simpler, we define $\kappa$ to be the sum over all edge, parallel-edges and vertex gadgets of the minimum number of edges any maximum stable matching contains inside them, assuming that at least one stable matching exists (i.e., for any $j$-vertex gadget, we add $2j+2$ to $\kappa$, for each parallel-edges gadget, we add $3$, and for each $(k_1, k_2)$-edge gadget, we add $5 + k_1+k_2$).
  Thus, if we define $\kappa_{(k_1, k_2)}$ to be the number of $(v^1, j_1, k_1, v^2, j_2, k_2)$-edge gadgets contained in $H$, then we have \[\kappa:=k\Bigl(n(C+2)+(k-1)\frac{n(n+1)}{2}\Bigr) + 3 k (n-1) + \sum_{(k_1, k_2)\in \mathbb{N}^2} (5 + k_1 +k_2)\kappa_{(k_1, k_2)}\]
  where the first summand arises from the vertex gadgets, the second from the parallel-edges gadgets (a stable matching in such a gadget contains at least three edges) and the third from the edge gadgets.

  If a gadget contains by $\ell$ more edges than the minimum number of edges any stable matching must contain inside the gadget, then we say that the gadget has $\ell$ \emph{additional} edges.

  The key part of the correctness of the reduction will be the following lemma which gives an upper bound on the size of a stable matching in $H$.

  \begin{lemma}\label{lmnoe}
    Let $M$ be a stable matching in $H$. Then $|M|\le \kappa + k(k-1) (n + 0.5) + kC$, and equality holds if and only if

    \begin{enumerate}
      \item every $c_i$ is matched to a vertex gadget,
      \item if $c_i$ is matched at rank $2l - 1$, then $c_{ij}$ is matched at rank $2(n-l) +1$, and
      \item for each $i< j$, $c_{ij}$ and $c_{ji}$ match to the same edge gadget.
    \end{enumerate}

  \end{lemma}

  \begin{proof}
    Let $M$ be a stable matching of size at least $\kappa+k(k-1) (n+ 0.5) + kC$.

    As $C$ is lower-bounded by the number of edges not contained in vertex gadgets, there are less than $C$ additional edges outside vertex gadgets. Thus, any stable matching containing at least $kC$ additional edges must match any vertex $c_i$ to a vertex gadget. Thus, we can assume that $c_i$ is matched to the vertex gadget corresponding to some vertex $x_i$. This yields $C + x_i (k-1)$ additional edges for vertex $c_i$.

    Consider any $(c_{ij}, 2k_1 +1, k_1, c_{ji}, 2k_2 +1, k_2)$-edge gadget between two incidence vertices~$c_{ij}$ and $c_{ji}$ with $i<j$. If the edge $\{c_{ij}, w^1\}\in M$, then we charge $k_1$ edges to $c_i$, and if $\{c_{ji}, w^2\}\in M$, we charge $k_2$ edges to $c_j$.
    If both $\{c_{ij}, w^1\}$ and $\{c_{ji}, w^2\}$ are contained in~$M$, then we additionally charge half an edge to $c_{ij}$ and $c_{ji}$.
    Then clearly the sum over all additional edges in edge gadgets equals the sum of charged edges.

    We know that $c_i$ is matched at rank $2x_i -1$, and yields $C+x_i(k-1)$ additional edges.
    As there exists a parallel-edges gadget which $c_i$ ranks at position~$2x_i -2$ and $c_{ij}$ ranks at position~$2 (n -\ell) + 2$, we know that $c_{ij}$ is matched at rank at most~$2(n-x_i)+1$.
    Thus, the incidence vertex $c_{ij}$ is matched to a $(c_{ij}, 2k_1 +1, k_1, c_{ji}, 2k_2 + 1, k_2)$-edge gadget with $k_1 \le (n-x_i)$ and $k_2\le n-x_j$.
    Therefore, there are at most $n-x_i$ edges charged to any $c_{ij}$, and $n-x_j$ edges charged to $c_{ji}$.
    Thus, for any fixed $i$, the sum of edges charged to~$c_i$ and $c_{ij}$ for $j\neq i$ is at most $x_i(k-1) + (k-1) (n-x_i) =(k-1) n$, where equality holds if and only if all $c_{ij}$ are matched at rank $2x_i -1$.

    The number of edges charged to incidence vertices is at most $\frac{k(k-1)}{2}$, and this is tight if and only if $c_{ij}$ and $c_{ji}$ match to the same edge gadget for all $i\neq j$.

    Thus, the total number of edges in $M$ is at most $\kappa + k (k-1) n + \frac{k(k-1)}{2} + kC$, with equality if and only if 1.-3. hold.
  \end{proof}

 \begin{theorem}\label{twhtcw}
   \textsc{Max-SRTI} parameterized by tree-cut width is W[1]-hard.
 \end{theorem}

 \begin{proof}
   The reduction is already described above.
   We now show $G$ contains a clique of size $k$ if and only if $H$ contains a stable matching of size at least $\kappa + k (k-1 ) n + \frac{k (k-1)}{2} +k C$.
   First, we show that a clique of size $k$ implies a stable matching of size at least $\kappa + k (k-1) n + \frac{k (k-1)}{2} + kC$, and afterwards, we will show the reverse direction.

  \begin{claim}
    If $G$ has a clique of size $k$, then $H$ has a stable matching of size at least $\kappa + k (k-1) n + \frac{k(k-1)}{2} + kC$.
  \end{claim}

  \begin{claimproof}
    Assume that $G$ contains a clique $\{x_1, \dots, x_k\}$.
    We match $c_i$ to the vertex-selection gadget corresponding to $x_i$, and the incidence vertices $c_{ij}$ to the edge gadgets corresponding to the edge~$(x_i, x_j)\in E(\overleftrightarrow{G})$ between $c_{ij}$ and $c_{ji}$.
    Inside each gadget, we apply \Cref{leg2} $(i)$ and $(iii)$ and take the maximum number of edges of a stable matching inside the gadget, assuming that the outgoing edges are contained if and only if they have already been added to the matching.

    By \Cref{lmnoe}, the matching has size $\kappa + k (k-1) n + \frac{k(k-1)}{2} + kC$ if it is stable, so it remains to show that $M$ is stable.

    Note that the vertex and edge gadgets are constructed in such a way that none of their vertices is part of a blocking pair.
    Thus, the blocking pair must contain a vertex from a parallel-edges gadget. We picked the edges inside the parallel-edges gadgets in a way such that no blocking pair contains two vertices inside a parallel-edges gadget.
    Thus, a parallel-edge gadget for an edge $e=\{c_i, c_{ij}\}$ contains a blocking pair if and only if both $c_i$ and $c_{ij}$ prefer $e$ over their partner. But this is not blocking as we connected $c_{ij}$ to an edge gadget corresponding to an edge ending in the vertex selected by $c_i$.
  \end{claimproof}

  \begin{claim}
    If $H$ contains a stable matching of size $\kappa + k (k-1) n + \frac{k(k-1)}{2} + kC$, then $G$ contains a clique of size $k$.
  \end{claim}

  \begin{claimproof}
    By \Cref{lmnoe}, we know that each $c_i$ is matched to a vertex gadget, each incidence vertex is matched to an edge gadget whose endpoint corresponds to the vertex selected by the vertex selection gadget, and the incidence vertices $c_{ij}$ and $c_{ji}$ are matched to the same edge gadget.
    Let $x_i $ be the vertex selected by $c_i$.
    Thus, for each $i\neq j$, the incidence vertex $c_{ij}$ matches to an edge gadget corresponding to an edge adjacent to $x_i$, and the vertex~$c_{ji}$ matches to an edge gadget corresponding to an edge adjacent to $x_j$, and these edge gadgets are the same, implying that $G$ contains the edge~$\{x_i, x_j\}$.
    Thus, $\{x_1, \dots, x_k\}$ forms a clique.
  \end{claimproof}
  
  Having shown the correctness of the reduction, it remains to show that the graph has bounded tree-cut width.
  So consider the following tree-cut decomposition: Let $T$ be a star containing a leaf for each gadget. The bag corresponding to the center consists of $\{c_1, \dots, c_k\}\cup\{c_{ij}: i\neq j\}$.
    Replace all leaves by the tree-cut decompositions of the gadgets.
    The center shall be the root of this tree-cut decomposition.

    We claim that the resulting tree-cut decomposition has width at most $k':=\max\{k^2, 10\}$. Obviously each bag contains at most $k'$ vertices. For the center, the torso size is $k^2$. For each other node, it is at most 10 (the maximum tree-cut width of a gadget), as the component containing the center gets contracted or suppressed. The adhesion of the center is 0, and for all other nodes, the adhesion is bounded by $10$. Thus, we have $\tcw (H) \le k'$.
 \end{proof}

 \begin{remark}
    By \Cref{lmnoe}, every stable matching of size $\kappa + k (k-1) n + \frac{k (k-1)}{2} + k C$ does not match any vertex to a tie of size at least three.
    As breaking ties cannot increase the size of a maximum stable matching, it follows that the instance arising by breaking each tie of size at least three contains a stable matching of size $\kappa + k (k-1) n + \frac{k (k-1)}{2}$ if and only if $H$ does. 
    Therefore, by breaking all ties of length at least three in an arbitrary way, one gets that \textsc{Max-SRTI} parameterized by tree-cut width is W[1]-hard even if each preference list is strictly ordered or a tie of length two.
 \end{remark}

\section{Perfect-SRTI and SRTI-Existence are Fixed-parameter Tractable for the Parameter Tree-cut Width}\label{sec:SRTIFPTTreeCutWidth}
 In this section,
 we present a dynamic programming FPT-algorithm for \textsc{Perfect-SRTI} and \textsc{SRTI-Existence} parameterized by tree-cut width.
 Recall that a tree-cut decomposition consists of a rooted tree $T$, and an assignment of the nodes of $T$ to a near-partition~$\mathcal{X} = \{X_t :{t\in V(T)}\}$ of $V(G)$.
 This naturally assigns to each node $t\in V(T)$ the union~$Y_t$ of the vertex sets assigned to nodes in the subtree of $T$ rooted at $t$.
 In order to design an FPT-algorithm with respect to tree-cut width, one approach is to compute a partial ``solution(s)'' on the subgraph $G_t$ of $G$ induced by $Y_t$ (somehow taking into account the at most $\tcw (G)$ many edges leaving $Y_t$) via bottom-up induction on $T$.
 As $Y_t$ contains at most $\tcw (G)$ vertices for every leaf $t\in V(T)$, the solution(s) for every leaf usually can be computed efficiently.
 The difficulty lies in the induction step, where one needs compute for some node~$t\in T$ solution(s) for $G_t $ given solution(s) to $G_c$ for every child $c$ of $t$.
 Here, the difficulty usually does not arise because $Y_t $ contains the vertices from $X_t$ which are not contained in $Y_c$ for every child $c$ of $t$ (as there are at most $\tcw (G)$ vertices contained in $X_t$), but that there may be an unbounded number of children of~$t$.
 Assuming that the tree-cut decomposition is nice helps in so far as it allows to bound the number of heavy children (i.e., children~$c$ such that at least three edges leave $Y_c$), but the number of light children (i.e., children $c$ such that at most two edges leave $Y_c$) is unbounded.
 
 The section is divided into three parts:
 In \Cref{sec:idea}, we give the high-level idea of the algorithm.
 In \Cref{sind}, we elaborate the details of the algorithm and prove its correctness.
 Finally, in \Cref{sec:Max-SMTI}, we show how to adopt the algorithm to solve \textsc{Max-SMTI} (i.e., \textsc{Max-SRTI} on bipartite graphs).
 
 Throughout this section, we fix a \textsc{SRTI-Existence} instance with acceptability graph~$G$ as well as a nice tree-cut decomposition $(T, \mathcal{X})$ of $G$ of width $k:= \tcw (G)$.
 
 \subsection{High-level Idea}
 \label{sec:idea}
 
 Given a nice tree-cut decomposition of the acceptability graph, we use dynamic programming to decide whether a solution exists or not.
 In the dynamic programming (DP) table for a node~$t \in V(T)$ we store information whether there exists a matching $M$ for the set $Y_t$~of
 vertices from the bags of the subtree of the tree-cut decomposition rooted in~$t$.
 We allow that $M$ is not stable in $G$ but require that the blocking pairs are incident to vertices outside $Y_t$, and for some of the edges~$\{v, w\}$ in $\cut (t)$ (recall that $\cut (t) $ is the set of edges in $G$ which leave $Y_t$), we require the endpoint $v$ in $Y_t$ to be matched at least as good as it ranks $w$.

\subparagraph{DP Tables}
Before we describe the idea behind the table entries we store in
our dynamic programming procedure, we introduce the following relaxation of
matching stability.

 \begin{definition}
   Let $(T, \mathcal{X})$ be a nice tree-cut decomposition of $G$.
   For a node $t\in V(T)$, the \emph{closure $\clos(t)$ of $t$} is the set of vertices in $Y_t$ together with their neighbors, that is, $\clos (t):= Y_t \cup \N_G (Y_t)$.
   We say that a matching $M$ on $\clos (t)$ for some $t\in V(T)$ \emph{complies with a vector $\bm{h}\in \{\oldminusone, \oldzero, \oldone\}^{\cut (t)}$} if the following conditions hold:
   \begin{itemize}
     \item for each edge $e \in \cut (t)$, we have $e\in M$ if and only if $\bm{h} (e) = \oldzero$;

     \item for each $e = \{v, y\}\in \cut (t)$ with $y\in Y_t$ and $\bm{h} (e) = \oldone$, we have $\rk_y (M(y)) \le \rk_y(v)$, i.e.,~$y$~ranks its partner (not being $v$ by the previous condition) in $M$ at least as good as~$v$; and

     \item every blocking pair contains a vertex from $V(G)\setminus Y_t$ not matched in~$M$.
   \end{itemize}
 \end{definition}
Intuitively, if we set $\bm{h} (e) = \oldone$ for an edge $e = \{v,y\}$ in $\cut(t)$ with $y \in Y_t$, then we are searching for a matching $M$ (in $G[\clos(t)]$) for which we can guarantee that $\rk_y(M(y)) \le \rk_y(v)$.
Consequently, we know that $e$ will not be blocking in an extension of such a matching.
On the contrary, if we set $\bm{h}_t(e) = \oldminusone$, then we allow $y$ to prefer $v$ over its partner (in particular, $y$ may remain unmatched).
Thus, for an extension of such a matching in order to maintain stability we have to secure that $\rk_v(M(v)) \le \rk_v(y)$, since otherwise $e$ will be blocking.
Observe that if a matching complies with $\bm{h}$ for a vector $\bm{h} \in \{\oldminusone,\oldzero,\oldone\}^{\cut(t)}$ with $\bm{h}(e) = \oldone$ for some edge~$e \in \cut(t)$, then it complies with $\bm{\hat{h}} \in \{\oldminusone,\oldzero,\oldone\}^{\cut(t)}$ which is the same as $\bm{h}$ but for $e$ is set to~$\oldminusone$ (formally, $\bm{\hat{h}}(f) = \bm{h}(f)$ for $f \neq e$ and $\bm{\hat{h}}(e) = \oldminusone$).
Clearly, any matching complying with~$\bm{h}$ complies with $\bm{\hat{h}}$, since the latter is more permissive.

 For a node $t \in V(T)$, its dynamic programming table is $\tau_t$ and it contains an entry for every $\bm{h} \in \{ \oldminusone, \oldzero, \oldone \}^{\cut(t)}$.
 An entry $\tau_t[\bm{h}]$ is a matching $M \subseteq E(G_t) \cup \cut(t)$ if $M$ complies with~$\bm{h}$.
 If no such matching for $\bm{h}$ exists, then we set $\tau_t[\bm{h}] = \square$.
 Note that the size of the table $\tau_t$ for a node $t$ is upper-bounded by $3^{\tcw(G)}$.

 \begin{example}
   A graph together with a corresponding tree-cut decomposition are depicted in \Cref{fextcwalgo}:
 \begin{figure}[bt]
\begin{center}
    \begin{tikzpicture}[xscale=0.8, yscale=0.75]
       \node[vertex, label = 90:$r_{1}$] (r1) at (0.5, 0.5) {};
       \node[vertex, label = 90:$r_{2}$] (r2) at (1.5, 0.5) {};

       \node[vertex, label = 180:$v_{11}$] (m1) at (-1, -2.5) {};
       \node[vertex, label = 225:$v_{12}$] (m2) at (0.5, -2.5) {};
       \node[vertex, label = 180:$v_{13}$] (m3) at (1.5, -1) {};
       \node[vertex, label = 90:$v_{14}$] (m4) at (1.5, -2.5) {};
       \node[vertex, label = 00:$v_{15}$] (m5) at (2.5, -2.5) {};

       \node[vertex, label = 180:$v_{21}$] (ll1) at (-1, 0.5) {};
       \node[vertex, label = 180:$v_{22}$] (ll2) at (-1, -1) {};

       \node[vertex, label = 0:$v_{41}$] (lr1) at (3.7, -3) {};
       \node[vertex, label = 0:$v_{42}$] (lr2) at (3.7, -1) {};
       \node[vertex, label = 0:$v_{43}$] (lr3) at (3.7, 0.5) {};

       \draw (m1) edge (m2);
       \draw (m2) edge  (m4);
       \draw (m4) edge (m5);
       \draw (m5) edge (m3);

       \node[vertex, label = 270:$v_{3}$] (lm) at (1, -4) {};
       \draw (lm) edge  (m2);
       \draw (lm) edge (m4);
       \draw (lm) edge (m5);

       \draw (ll1) edge (ll2);
       \draw (ll2) edge  (m1);
       \draw (ll1) edge  (r1);

       \draw (lr2) edge  (m5);
       \draw (lr1) edge (lr2);
       \draw (lr2) edge (lr3);
       \draw (lr1) edge[bend right=55] (lr3);
       \draw (lr2) edge (m3);

       \draw (m3) edge (r2);
     \end{tikzpicture}

     \begin{tikzpicture}
       \node (shift) at (0, 0.35) {};
       \draw (0, 0) rectangle (3, 1);

       \draw ($(-1, -1.5) + (shift)$) rectangle ($(3.5, -3.5) + (shift)$);

       \draw ($(-5, -5) + 2*(shift)$) rectangle ($(-2, -6) + 2*(shift)$);
       \draw ($(-0.5, -5) + 2*(shift)$) rectangle ($(1.5, -6) + 2*(shift)$);
       \draw ($(3.5, -5) + 2*(shift)$) rectangle ($(7, -6) + 2*(shift)$);

       \draw[ultra thick] ($(1, 0)$) -- ($(1, -1.5) + (shift)$);
       \draw[ultra thick] ($(-1, -3.5) + (shift)$) -- ($(-2, -5) + 2*(shift)$);
       \draw[ultra thick] ($(1, -3.5) + (shift)$) -- ($(1, -5) + 2*(shift)$);
       \draw[ultra thick] ($(3.5, -3.5) + (shift)$) -- ($(4, -5) + 2*(shift)$);

       \node[vertex, label = 90:\Large{$r_{1}$}] (r1) at (0.5, 0.45) {};
       \node[vertex, label = 90:\Large{$r_{2}$}] (r2) at (1.5, 0.45) {};

       \node[vertex, label = 90:\Large{$v_{11}$}] (m1) at ($(-0.5, -2.1) + (shift)$) {};
       \node[vertex, label = 90:\Large{$v_{12}$}] (m2) at ($(0.5, -2.5) + (shift)$) {};
       \node[vertex, label = 180:\Large{$v_{13}$}] (m3) at ($(1.7, -2) + (shift)$) {};
       \node[vertex, label = 90:\Large{$v_{14}$}] (m4) at ($(1.5, -3) + (shift)$) {};
       \node[vertex, label = 00:\Large{$v_{15}$}] (m5) at ($(2.7, -2.7) + (shift)$) {};

       \draw (m1) edge node[pos=0.2, fill=white, inner sep=2pt] {\footnotesize $1$}  node[pos=0.76, fill=white, inner sep=2pt] {\footnotesize $1$} (m2);
       \draw (m2) edge node[pos=0.2, fill=white, inner sep=2pt] {\footnotesize $1$}  node[pos=0.76, fill=white, inner sep=2pt] {\footnotesize $3$} (m4);
       \draw (m4) edge node[pos=0.2, fill=white, inner sep=2pt] {\footnotesize $1$}  node[pos=0.76, fill=white, inner sep=2pt] {\footnotesize $1$} (m5);
       \draw (m5) edge node[pos=0.2, fill=white, inner sep=2pt] {\footnotesize $2$}  node[pos=0.76, fill=white, inner sep=2pt] {\footnotesize $2$} (m3);

       \node[vertex, label = 270:\Large{$v_{3}$}] (lm) at ($(1, -5.5) + 2*(shift)$) {};
       \draw (lm) edge node[pos=0.2, fill=white, inner sep=2pt] {\footnotesize $1$}  node[pos=0.76, fill=white, inner sep=2pt] {\footnotesize $2$} (m2);
       \draw (lm) edge node[pos=0.2, fill=white, inner sep=2pt] {\footnotesize $2$}  node[pos=0.76, fill=white, inner sep=2pt] {\footnotesize $2$} (m4);
       \draw (lm) edge node[pos=0.2, fill=white, inner sep=2pt] {\footnotesize $3$}  node[pos=0.76, fill=white, inner sep=2pt] {\footnotesize $2$} (m5);

       \node[vertex, label = 270:\Large{$v_{21}$}] (ll1) at ($(-3.9, -5.5) + 2*(shift)$) {};
       \node[vertex, label = 270:\Large{$v_{22}$}] (ll2) at ($(-2.5, -5.5) + 2*(shift)$) {};

       \draw (ll1) edge node[pos=0.2, fill=white, inner sep=2pt] {\footnotesize $1$}  node[pos=0.76, fill=white, inner sep=2pt] {\footnotesize $2$} (ll2);
       \draw (ll2) edge[bend left=15] node[pos=0.2, fill=white, inner sep=2pt] {\footnotesize $1$}  node[pos=0.8, fill=white, inner sep=2pt] {\footnotesize $2$} (m1);
       \draw (ll1) edge[bend left] node[pos=0.2, fill=white, inner sep=2pt] {\footnotesize $1$}  node[pos=0.76, fill=white, inner sep=2pt] {\footnotesize $1$} (r1);

       \node[vertex, label = 270:\Large{$v_{41}$}] (lr1) at ($(3.8, -5.5) + 2*(shift)$) {};
       \node[vertex, label = 270:\Large{$v_{42}$}] (lr2) at ($(4.8, -5.5) + 2*(shift)$) {};
       \node[vertex, label = 270:\Large{$v_{43}$}] (lr3) at ($(5.8, -5.5) + 2*(shift)$) {};

       \draw (lr2) edge[bend right = 25] node[pos=0.2, fill=white, inner sep=2pt] {\footnotesize $1$}  node[pos=0.76, fill=white, inner sep=3pt] {\footnotesize $2$} (m5);
       \draw (lr1) edge node[pos=0.2, fill=white, inner sep=2pt] {\footnotesize $1$}  node[pos=0.76, fill=white, inner sep=2pt] {\footnotesize $3$} (lr2);
       \draw (lr2) edge node[pos=0.2, fill=white, inner sep=2pt] {\footnotesize $2$}  node[pos=0.76, fill=white, inner sep=2pt] {\footnotesize $1$} (lr3);
        \draw (lr1) edge[bend left=35] node[pos=0.2, fill=white, inner sep=2pt] {\footnotesize $2$}  node[pos=0.76, fill=white, inner sep=2pt] {\footnotesize $1$} (lr3);
       \draw (lr2) edge[bend right = 60] node[pos=0.2, fill=white, inner sep=2pt] {\footnotesize $2$}  node[pos=0.76, fill=white, inner sep=2pt] {\footnotesize $1$} (m3);

        \draw (m3) edge node[pos=0.2, fill=white, inner sep=2pt] {\footnotesize $1$}  node[pos=0.76, fill=white, inner sep=2pt] {\footnotesize $1$} (r2);

       \node[] (lr) at (2.5, 0.5) {\LARGE{$r$}};
       \node (lm) at ($(-0.5, -3) + (shift)$) {\LARGE{$t_1$}};
       \node (lll) at ($(-4.5, -5.5) + 2*(shift)$) {\LARGE{$t_2$}};
       \node (llm) at ($(0, -5.5) + 2*(shift)$) {\LARGE{$t_3$}};
       \node (llr) at ($(6.7, -5.5) + 2*(shift)$) {\LARGE{$t_4$}};
     \end{tikzpicture}
\end{center}

   \caption{\label{fextcwalgo}%
   An example of a graph $G$ (upper part) and its nice tree-cut decomposition~$(T, \mathcal{X})$ (not of minimal width) (lower part).
   The vertices of $G$ are the circles, while the nodes of $T$ are the rectangles.
   For a node $t\in V(T)$, bag $X_t$ contains exactly the vertices inside the rectangle.
   In the lower picture, the thick edges are the edges of $T$, while the thin edges are from $G$.
   The nodes~$t_1$ and $t_4$ are light, while $t_2$ (because there is an edge connecting a vertex in $t_2$ to a vertex in~$r$) and $t_3$ (because $\adh (t_3) =3$) are heavy.
   }
 \end{figure}
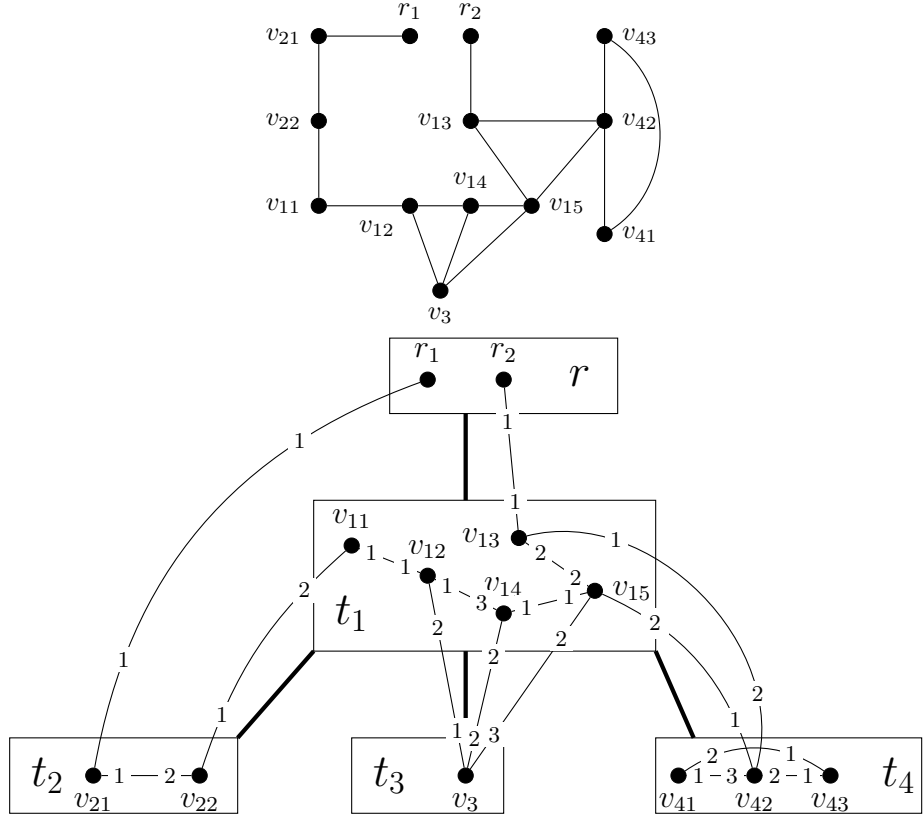

   For $t_1$ and the vector $\bm{h^1} \in \{\oldminusone , \oldzero, \oldone\}^{\cut (t_1)}$ with $\bm{h^1}(\{r_1, v_{21}\}) = \oldzero$ and $\bm{h^1}(\{v_{13}, r_2\}) = \oldone$, all stable matchings containing $\{r_1, v_{21}\}$ and $\{v_{13}, v_{42}\}$ comply with $\bm{h^1}$.
   For $t_2$ and the vector $\bm{h^1} \in \{\oldminusone , \oldzero, \oldone\}^{\cut (t_2)}$ with $\bm{h^2}(\{v_{21}, r_1\}) = \oldone$ and $\bm{h^2}(\{v_{22}, v_{11}\})=\oldminusone$, the matching~$M = \{ \{v_{21}, v_{22}\} \}$ complies with $\bm{h^2}$.
   For $t_3$ and any vector $\bm{h^3} \in \{\oldminusone , \oldzero, \oldone\}^{\cut (t_2)}$ with $\bm{h^3}(\{v_3, v_{12}\}) = \oldone$, no matching complies with $\bm{h^3}$:
   Such a matching would need to match~$v_3$ at least as good as~$\rk_{v_3} (v_{12})=1$, but not to $v_{12}$, which is impossible.
   For $t_4$ and the vector $\bm{h^1} \in \{\oldminusone , \oldzero, \oldone\}^{\cut (t_4)}$ with $\bm{h^4}(\{v_{42}, v_{15}\}) = \oldzero = \bm{h^4}(\{v_{42}, v_{13}\})$, no matching complying with $\bm{h^4}$ exists, as such a matching must match $v_{42}$ to both $v_{13} $ and $v_{15}$.
 \end{example}

For each leaf $t\in V(T)$, we have that $|\clos (t) | \le |Y_t  | + |\cut (t)| \le 2k$.
Therefore, we can find a matching complying with a vector $\bm{h} \in \{\oldminusone, \oldzero , \oldone\}^{ \cut (t)}$ in $2^{O(k \log k)}$ time if it exists by enumerating all matchings on $\clos (t)$.

For the root $r$ of $T$, it holds that $Y_r = V(G)$.
Therefore, a stable matching exists if and only if there exists a stable matching complying with (the unique) vector $\bm{h}\in \{\oldminusone, \oldzero, \oldone\}^{\cut (r)}$.

Thus, it remains to find a matching complying with a vector $\bm{h} \in \{\oldminusone, \oldzero, \oldone\}^{\cut (t)}$, given $\tau_c [ \bm{h^c}]$ for each child $c$ of $t$ and each $\bm{h^c} \in \{\oldminusone, \oldzero, \oldone\}^{\cut (c)}$.
We call this the \emph{induction step}.

\subparagraph{Induction Step}
In what follows, we sketch how to perform the induction step.
\defProblemTask{Induction Step}
 {
  The acceptability graph $G$, rank functions $\rk_v$ for all $v\in V(G)$, a nice tree-cut decomposition $(T, \mathcal{X})$, a node $t\in V(T)$, a vector $\bm{h}\in \{\oldminusone, \oldzero, \oldone\}^{\cut (t)}$, and for each child $c$ of $t$ and each $\bm{h^c} \in \{\oldminusone, \oldzero, \oldone\}^{\cut (c)}$ the value $\tau_c [\bm{h^c}]$.
 }
 {
  Compute $\tau_t [ \bm{h}]$.
 }

Before we give the proof idea, we first give the definition of light children classes.
Intuitively, two light children of a node $t$ are in the same class if, with respect to $t$, they behave in a similar way, that is, their neighborhoods in $X_t$ are the same and their table entries are compatible.
In order to properly define the later notion, we first need to introduce few auxiliary definitions.

 \begin{definition}
  Let $t\in V(T)$ be a node.
  Its \emph{signature} $\sig (t)$ is the set $\{\bm{h}\in \{\oldminusone, \oldzero, \oldone\}^{\cut (t)}: \tau_t [\bm{h} ]\neq \square\}$.

 Let $c, d$ be light children of $t$.
 Then, $c \Diamond d$ means $\sig (c) = \sig (d)$ and $N(c) = N(d)$, where $N(c) := N_G (Y_c)$ for each $c\in V(T)$.
\end{definition}

It follows immediately that $\Diamond$ is an equivalence relation on light children of $t$.
Furthermore, since each class of $\Diamond$ is identified by its signature and neighborhood in $X_t$, there are $O(k^2)$ classes of $\Diamond$.
Let $\mathcal{C}(c)$ denote the equivalence class of the light child $c$ of $t$ and let $N(\mathcal{C}) \subseteq X_t$ be the set of neighbors of the class $\mathcal{C}$ of $\Diamond$ (i.e., $N(\mathcal{C})$ is the set of neighbors $N(Y_c) \subseteq X_t$ for~$c \in \mathcal{C}$).
Furthermore, let $\sig(\mathcal{C})$ denote the signature of the class $\mathcal{C}$ and, similarly, let $\sig_x(\mathcal{C})$ denote the signature of $\mathcal{C}$ with respect to its neighbor $x \in N(\mathcal{C})$.

From the definition of the signature, it immediately follows that if a stable matching exists, then every class containing at least three children must contain $(\oldminusone, \oldminusone)$.

\begin{observation}\label{obs:smallClassOfChildren}
If $\mathcal{C}$ is a class of $\Diamond$ with $|\mathcal{C}| \ge 3$ and $(\oldminusone,\oldminusone) \notin \sig(\mathcal{C})$, then there is no stable matching in $G$.
\end{observation}

\begin{proof}
  Assume for a contradiction that there exists a stable matching $M$ in $G$ and a class~$\mathcal{C}$ with $|\mathcal{C}| \ge 3 $ but $(\oldminusone, \oldminusone) \notin \sig (\mathcal{C}|$.
  Let $N(c) = \{x_1, x_2\}$ for all $c \in \mathcal{C}$.
  As $|\mathcal{C}| \ge 3$, there exists some child $c^*\in \mathcal{C}$ such that neither $x_1$ nor $x_2$ is matched to a vertex in $Y_c$.
  It follows that $M$ restricted to $\clos (c^*)$ complies with the vector $\bm{(\oldminusone, \oldminusone}$ (otherwise there exists a blocking pair for $M$ in $Y_{c^*}$, contradicting the stability of $M$), a contradiction to $(\oldminusone, \oldminusone) \notin \sig (c^*)$.
\end{proof}

We now sketch the main ideas on how to perform the induction step.

\subparagraph*{Proof Idea (Induction Step)}
For proof details, see \Cref{sind}.
First, we will guess for each heavy child~$c$ of $t$ a vector $\bm{h^c}$ such that the matching restricted to $\clos (c)$ shall comply with $\bm{h^c}$. 
Since every node has at most $2 (k + 1) $ heavy children by \Cref{laatcw}, and every node has adhesion at most $k$, this can be done in $(3^k)^{O(k)}$.
(In fact, we will show that it is enough to consider $k^{O(k)}$ such guesses (\Cref{lheavy}).)
It is worth noting that these choices will result in some further constraints the matching in the light children must fulfill.

Then, for every class of the equivalence $\Diamond$ we guess whether its neighbor(s) in $X_t$ are matched to it (i.e., matched to a vertex in a child or two in this class) or not.
Let $\mathcal{N}$ denote the guessed matching.
Since each vertex $x \in X_t$ is ``adjacent'' to at most $k+1$ classes of $\Diamond$ (i.e., there are at most $k$ classes such that $x$ is adjacent to a vertex in all children contained in this class), choosing a class (or deciding not to be adjacent to any light child) for every vertex in $X_t$ results in~$k^{O(k)}$ guesses.
We show that if a class of $\Diamond$ with $N(\mathcal{C}) = \{ x \}$ is selected to be matched with its neighbor $x$, then it is possible to match $x$ to the best child in this class (the one containing the top choice for~$x$ among these children), provided that there exists a solution which is compatible with such a choice.
We do this by providing a rather simple exchange argument (\cref{lHsimple}).

Having resolved heavy children and light children with only one neighbor in $X_t$, it remains to deal with children with two neighbors.
To this end, we generalize the exchange argument we provide for classes with one neighbor (\cref{lcomplie}).
Then, we prove that some combinations of $\mathcal{N}$ and a signature of a class $\mathcal{C}$ allows us to reduce the number of children in $\mathcal{C}$ in which we have to search for a partner of a vertex in $N(\mathcal{C})$ to a constant (in fact, four).
We call such classes the \emph{good classes}.
However, there are classes where are not able to do so (call these the \emph{bad classes}).
Consequently, there are only $4^k$ possibilities how to match vertices in $X_t$ to good classes of children.
Finally, we show how to reduce the question of existence of a (perfect) matching complying with $\bm{h}$ and obeying all the constraints of heavy children to an instance of 2-SAT.
This basic idea of this reduction is similar to the reduction from \textsc{Stable Roommates} to \textsc{2-SAT} by Feder~\cite{Feder1992}, although we have to model several additional constraints.

\subsection{Details of the Induction Step}
 \label{sind}

 We now describe the induction step in detail and prove its correctness.
 Recall that in the induction step, we are given the \textsc{Max-SRTI} instance together with a nice tree-cut decomposition~$(T, \mathcal{X})$ of its acceptability graph, a node $t\in V(T)$, a vector~$\bm{h} \in \{\oldminusone, \oldzero, \oldone\}^{\cut (c)}$ as well as $\tau_c [\bm{h^c}]$ for every child $c$ of $t$ and $\bm{h^c} \in \{\oldminusone, \oldzero, \oldone\}^{\cut (c)}$, and the task is to decide whether there exists a matching complying with $\bm{h}$.
 In order to do so, we first guess how the matching shall be on the heavy children (\Cref{sec:heavy}).
 Recall that we partitioned the light children into $O(k^2)$ equivalence classes.
 We then guess for each class of light children whether their neighbor(s) are matched to a vertex contained in one of the light children.
 For each class~$\mathcal{C}$ for which we guessed that at least one neighbor~$x$ is matched to, we then try to compute a set $\mathcal{C}'\subseteq \mathcal{C}$ of constant size such that if there exists a matching complying with~$\bm{h}$, then there also exists one which matches $x$ to a vertex contained in $Y_c$ for some~$c\in \mathcal{C}'$.
 This is done in \Cref{applightchildren}.
 As there are at most $k$ classes, we can afford to guess for each class~$\mathcal{C}$ for which we could reduce to a constant-sized set~$\mathcal{C}'$ to which child in $\mathcal{C}'$ vertex $x$ is matched to.
 Finally, in \Cref{sec:2-SAT}, we show how to reduce the problem of finding a matching complying with $\bm{h}$ and fulfilling our guesses to 2-SAT.

 \subsubsection{The Heavy Children}
 \label{sec:heavy}

 We first deal with the heavy children.
 Since each node can have at most $k$ heavy children and for each heavy child $c$ we have $\cut (c) \le k$, it follows that there are at most $(3^k)^k = 3^{k^2}$ combinations of vectors for all heavy children together.
 However, many of these combinations can be ignored.
 For example, assume that the node $t$ has two heavy children~$c_1$ and $c_2$, and for some $x\in X_t$, there exists an edge $\{x, y_i\}$ for some $y_i \in Y_{c_i}$ for $i\in \{1, 2\}$ with $\rk_x(y_1 ) < \rk_x (y_2)$.
 If we now want to extend a matching complying with a vector $\bm{h^1}$ with $\bm{h^1} (\{x, y_1\}) = \oldminusone$ for $c_1$, and a vector $\bm{h^2}$ with $\bm{h^2} (\{x, y_2\}) = \oldone$ for $c_2$, then we could use a matching complying with a vector $\bm{h'^{2}}$ with $\bm{h'^{2}} (\{x, y_2\}) = \oldminusone$ for $c_2$ instead, and thus, we do not need to consider the combination of vectors $(\bm{h^1}, \bm{h^2})$ at all.

 Based on this idea, we reduce the number of combinations of vectors to consider from $3^{k^2}$ to $2^{O(k\log k)}$, thus improving the running time.

 \begin{lemma}\label{lheavy}
   Let $t\in V(T)$ be a non-leaf node, let $H = \{c_1, \dots, c_\ell\}$ be the set of its heavy children, and let $\bm{h}\in \{\oldminusone, \oldzero, \oldone\}^{\cut (t)}$ be given.

   There exists a family $\mathcal{H}$ of tuples $\left( \bm{g^{c_1}}, \dots, \bm{g^{c_\ell}} \right)$ with $\bm{g^{c_i}} \in \{\oldminusone, \oldzero, \oldone\}^{\cut (c_i)}$, where $k = \tcw(G)$, such that $|\mathcal{H}| \le k^{O(k)}$ and the following holds.
   If there exists a matching $\tilde{M} \subseteq E(G_t) \cup \cut(t)$ complying with $\bm{h}$, then there exists a matching $M$ complying with $\bm{h}$ for which there exists $(\bm{g^{c_1}}, \dots, \bm{g^{c_\ell}}) \in \mathcal{H}$ such that $\tau_c[\bm{g^c}] = M \cap \left( E(G_c) \cup \cut(c) \right)$ holds for every $c \in H$.

   The family $\mathcal{H}$ can be computed in $2^{O(k \log k)}$ time.
 \end{lemma}
 \begin{proof}
  Let $\cut (H):= \bigcup_{c\in H} \cut (c)$.

  We generate the family $\mathcal{H}$ as follows.
  Recall that for a set of vertices~$X$, we defined $\delta (X)$ to be the set of edges leaving~$X$, i.e., with one endpoint in $X$ and one outside $X$.
  First, we enumerate all possible subsets $F$ of edges from $\bigl(\bigcup_{c\in H} \cut (c)\bigr) \cap \delta (X_t)$ which can be part of the matching.
  For each such subset~$F$, we enumerate all functions $f_F: X\rightarrow \cut (H) \cup \{\square\}$ with $f_F (x) \in \delta (X) \cup \{\square\}$ for all~$x\in X$, and $f_F (x) = e$ for all $e = \{v, x\} \in F$.
  For each pair $(F, f_F)$, we add the tuple $t(F, f_F):= (\bm{g^{c_1}}, \dots, \bm{g^{c_\ell}})$ to $\mathcal{H}$, where
  \[\bm{g^{c_i}} (\{v, x\}) := \begin{cases}
                     \oldzero & \text{ if } \{v, x \} \in F,\\
                     \oldone & \text{ if } \{v, x\} \notin F \land \Bigl( f_F(x) = \square \lor \rk_x (v) < \rk_x (f_F(x)) \Bigr),\\
                     \oldminusone & \text{ otherwise.}
                   \end{cases}
  \]

  Note that indeed $|\mathcal{H}| = 2^{O(k \log k)}$, as there are only $2^{O(k \log k)} $ possible subsets $F$ (as all but $k$ edges from $\cut (H)$ end in $X_t$, and the vertices in $X_t$ have at most $k$ neighbors in $Y_c$ for each of the $\ell\le k$ heavy children) and only $2^{O(k \log k)}$ possible functions $f_F$.
  Moreover, the enumeration of $F$ and $f_F$ can also be done in $2^{O(k\log k)}$ time.

  Let now $\tilde{M}$ be a matching complying with $\bm{h}$.
  Let $\tilde{F}:= \tilde{M} \cap \cut (H)$ and $\bm{g} := t(\tilde{F}, f_{\tilde{F}})$.
  Let~$f_{\tilde{F}} (x)$ be an edge $\{v, x\}\in \delta (v) \cap \cut (H)$ such that $\rk_x (v) \ge \rk_x (\tilde{M} (x))$, and $\rk_x (v)$ minimal.
  We claim that $M := \Bigl( \tilde{M} \setminus (\bigcup_{c\in H} (E(G_c) \cup \cut(c))\Bigr) \cup M_{\bm{g}}$ is also a matching complying with~$\bm{h}$, where $M_{\bm{g}} = \bigcup_{c\in H} \tau_c [\bm{g^c}]$.

  First, we show that $\tau_c [ \bm{g^c} ] \neq \square $ as $\tilde{M} \cap (E(G_c) \cup \cut (c) ) $ complies with $\bm{g^c}$.
  If not, then---since $M$ and~$\tilde{M}$ coincide on $\cut (c)$---there is some $v\in Y_c$ such that $\bm{g^c} (\{v, x\}) = \oldone$, but $\rk_v (x) < \rk_v (\tilde{M}(x))$.
  However, by the definition of $\bm{g^c}$, we get that also $\rk_x (v) < \rk_x (\tilde{M}(x))$, and thus, $\{v, x\}$ is a blocking pair for $\tilde{M}$, a contradiction to the assumption that $\tilde{M} $ complies with $\bm{h}$.

  Second, note that all agents $x\in X_t$ are matched the same in $\tilde{M}$ and $M$, and for each edge~$e\in \cut (H)$, we have that $e\in \tilde{M}$ if and only if $e\in M$.
  Thus, $M$ fulfills the first two properties of a matching complying with $\bm{h}$ as $\tilde{M}$ is complying with $\bm{h}$.

  It remains to show that $M$ also fulfills the third property, i.e., that there is no blocking pair inside $Y_t$.
  Since $\tilde{M}$ did not contain a blocking pair inside $Y_t$, and $M$ arises from $\tilde{M}$ by replacing the matching in $Y_c$ by a matching complying with $\bm{g^c}$ for each $c\in H$, any such blocking edge must be contained in $\cut (H)$.
  So assume that there is such a blocking edge~$\{v, x\}\in \cut (c)$.
  Thus, $\rk_x (v) < \rk_x (M(x))$, which implies that $\bm{g^c} (v) = \oldone$.
  However, since $\tau_c[\bm{g}^c]$ complies with $\bm{g}^c$, we have that $\rk_v (x) \ge \rk_v (M(v))$, contradicting the assumption that $\{v, x\}$ is a blocking pair.
 \end{proof}

\subsubsection{The Light Children}\label{applightchildren}

Recall that for each $x\in X_t$, we guess a class $\mathcal{C}$ of light children to which $x$ is matched (i.e., there exists some $v\in V_c$ for some $c\in \mathcal{C}$ such that $x$ is matched to $v$) or that $x$ is not matched a light child.
In this section, assuming we guessed $x$ to be matched to $\mathcal{C}$, we want to reduce the number of possible partners of $x$ to a constant number.
The classes for which we successfully do so will be called \emph{good} classes, while all other classes will be called bad.

We start with the formal description of the guesses to which class the vertices in $X_t$ are matched via a node-to-children matching $\mathcal{N}$.
\subparagraph*{Auxiliary Node Graph}
For a node $t$ we define an auxiliary node graph $H_t$ to be a bipartite graph with one side of the bipartition $X_t$ and the other side containing a vertex $x^{\mathcal{C}}$ for each class~$\mathcal{C}$ and each $x \in X_t$ with $x \in \N(\mathcal{C})$.
There is an edge $\left\{ x, x^{\mathcal{C}} \right\}$ whenever~$x \in \N(\mathcal{C})$.
Now, let $\mathcal{N}$ be a matching in $H_t$ and let us slightly abuse notation and write $\mathcal{N}(x) = \mathcal{C}$ if~$\left\{x, x^{\mathcal{C}} \right\} \in \mathcal{N}$.
Furthermore, if a vertex $x\in X_t$ is unmatched in $\mathcal{N}$, thenwe write $\mathcal{N} (x) = \bot$.

As described in the proof idea (\Cref{sec:idea}), we enumerate all $k^{O(k)}$ matchings $\mathcal{N}$ in $H_t$.
We will therefore from now on fix such a matching $\mathcal{N}$, and try to find a matching obeying $\mathcal{N}$ (and of course complying with $\bm{h}$).
Formally, we are looking for an embedding of $\mathcal{N}$ in the graph $G$ complying with $\bm{h}$.
We say that a matching $M$ is an \emph{$\bm{h}$-embedding} of $\mathcal{N}$ if
\begin{itemize}
  \item $M$ complies with $\bm{h}$ and
  \item there is an edge $\left\{x, x^{\mathcal{C}} \right\} \in \mathcal{N}$ if and only if $\left\{x, y^c \right\} \in M$ for some $y^c\in Y_c$ for some~$c \in \mathcal{C}$.
\end{itemize}

  Let us denote by $E(\mathcal{C})$ the set of edges $\bigcup_{c \in \mathcal{C}} \left( E(G_c) \cup \cut(c) \right)$, and by $V(\mathcal{C})$ the vertex set~$\bigcup_{c \in \mathcal{C}} Y_c$.
  Similarly, we denote by $\cut (\mathcal{C})$ the set of edges $\bigcup_{c\in \mathcal{C}} \cut (c)$.
  Let $x\in X_t$ and $\mathcal{C}$ be a class with $x\in \N (\mathcal{C})$.
  We define $\best (\mathcal{C}, x) \subseteq V(\mathcal{C})$ to be the set of vertices~$v$ minimizing $\rk_x (v) $ among all $v\in V(\mathcal{C})$.

Given such a matching $\mathcal{N}$, we know for each $x\in X_t$ to which class~$M(x)$ belongs, but not which vertex~$M(x)$ is exactly.
That is, $\mathcal{N}(x)$ is known, whereas $M(x)$ is yet to be computed.
We will now show that in many cases, we can then choose at most four vertices, such that there exists a solution matching $x$ to one of these four vertices if there exists such a matching for $\mathcal{N}$ complying with $\bm{h}$.
If $\mathcal{N} (x)$ contains at most four children for all $x\in X$, then we can enumerate all possible choices of $M(x) $ for all $x\in X$ in $O(4^k)$ time.
Thus, our overall strategy is to show that $\mathcal{N} (x)$ can be reduced to contain at most four children in some cases (see \Cref{sec:singleton,sec:non-singleton}).
We then enumerate $M(x) $ on all such~$x\in X_t$, and using a reduction to 2-SAT, we then show that we can decide whether we can extend this partial matching in a stable way in polynomial time in \Cref{sec:2-SAT}.

We now first deal with classes of children containing only one neighbor in $X_t$, which we call \emph{singleton classes}.

\paragraph{Singleton Classes}
\label{sec:singleton}

We will show how to deal with a class $\mathcal{C}$ of light children which have only one neighbor in $X_t$.
More precisely, we will show that there exists a child~$c\in \mathcal{C}$ such that from any matching $M$ complying with a vector $\bm{h}$ and containing an edge from $\cut (\mathcal{C})$, we can perform a ``local'' exchange and get a matching $M'$ containing an edge from $\cut (c)$ and also complying with $\bm{h}$.
The exchange being ``local'' here means that we exchange only edges from $E(\mathcal{C})$.
This locality allows us to apply similar exchange arguments (\Cref{lem:once-nonzero,lem:once-zero,lem:same-child,lem:different-child}), which allows us to prune all but a finite number of members of a class (so-called ``good'' classes),
or at least to prune members of some members the corresponding classes (so-called ``bad classses'').

If the class has only one neighbor $x$ (although two different vertices from the class can have an edge to this neighbor), then we just match $x$ to the vertex it prefers most among those that do not lead to a blocking pair inside $Y_\mathcal{C} \cup \{x\}$.
Note that this is not necessarily the vertex best-ranked by $x$ among the neighbors of $x$ in $\mathcal{C}$.
For example, if $\rk_x (v_1) < \rk_x (w_1) < \rk_x(w_2) < \rk_x (v_2)$ with $v_i, w_i\in Y_{c_i}$ for two children $c_i$ of $x$ from the same class, and $\sig (c_i) = \{\bm{(\oldminusone, \oldminusone)}, \bm{(\oldone, \oldminusone)}, \bm{(\oldminusone, \oldone)}, \bm{(\oldone, \oldone)}, \bm{(\oldminusone, \oldzero)}\}$ (where the first coordinate refers to $v_i$, and the second to $w_i$), then we would match $x$ to $w_2$:
Note that $v_i$ cannot be matched to $x$ as $\bm{(\oldzero, z)}\notin \sig (c_i)$ for all $z\in \{\oldminusone, \oldzero, \oldone\}$.
Also $w_1$ cannot be matched to $x$:
In this case $x$ prefers $v_1$ to its partner~$w_1$ in the matching.
Since $\bm{(\oldone, \oldzero)}\notin \sig (c_i)$ it follows that the matching must contain a blocking pair inside $Y_{c_1}$ or $v_1$ prefers $x$ to its partner in the matching.
Since $x$ prefers $v_1$ to $w_1$, it follows that there is a blocking pair in $Y_c \cup \{x\}$.

\begin{lemma}\label{lHsimple}
  Let $\mathcal{C}$ be a class with $\N(\mathcal{C}) = \{x\}$ for some $x\in X_t$.
  Assume $\mathcal{N}(x) = \mathcal{C}$ and let $\bm{h}$ be given.
  We can compute in polynomial time a child $c \in \mathcal{C}$ such that the following holds:
  If there is an $\bm{h}$-embedding~$M$ of~$\mathcal{N}$, then there is a set of edges $F \subseteq E(\mathcal{C})$ such that the matching $M \Delta F$ is an $\bm{h}$-embedding of $\mathcal{N}$ and $(M \Delta F)(x) \in Y_c$.
\end{lemma}

\begin{proof}
  For each edge $e= \{v, x\}\in\cut (\mathcal{C})$ and each child $d\in \mathcal{C}$, we define a vector $\bm{h^{d, e}} \in {\{\oldminusone, \oldzero, \oldone\}}^{ \cut (d)}$ by setting $\bm{h^{d, e}} (\{w, x\})$ for each edge $\{w, x\} \in \cut (d)$ with $w\in Y_d$ according to the following case distinction.
  \[\bm{h^{d, e}} (\{w, x\}):= \begin{cases}
 \oldone & \text{ if } \rk_x (w) < \rk_x (v),\\
 \oldzero & \text{ if } v= w,\\
 \oldminusone & \text{ otherwise.}
\end{cases}\]

  Let $c$ be the child with $e=\{v, x\}\in \cut (c)$, where $\{v, x\}$ is the edge minimizing $\rk_x (v) $ among all $\{v, x\}\in E(\mathcal{C})$ such that for all $d\in \mathcal{C}$ we have $\tau_{d} [\bm{h^{d, e}}]\neq \square$.
  (Clearly, $c$ can be computed in polynomial time.)

  Assume that $M$ is an $\bm{h}$-embedding of $\mathcal{N}$ and denote $\{x, M(x)\}$ by $e'$.
  By our assumptions, we have~$\tau_{d} [\bm{h^{d, e'}}] \neq \square$ for all $d\in \mathcal{C}$, and thus $\rk_x (v) \le \rk_x (M(x))$.
  Let $c'$ be the child such that $M(x) \in Y_{c'}$.
  Let
  \[
  F:=
  \Bigl( M \setminus \bigl(E(G_c) \cup E(G_{c'}) \cup \{x, M(x)\}\bigr)\Bigr) \cup \tau_{c} [\bm{h^{c, e}}] \cup \tau_{c'} [\bm{h^{c', e}}] \,.
  \]

  It remains to show that $M \Delta F$ complies with $\bm{h}$.
  The second and third condition obviously follow from $M$ complying with $F$, as $M$ and $M \Delta F$ only differ on $V(\mathcal{C}) \cup \{x\}$, and $\rk_x (M(x))\ge \rk_x ( (M\Delta F)(x) )$.
  Since $M$ and $M \Delta F$ coincide on $Y_t \setminus \bigl( V(\mathcal{C}) \cup \{x\}\bigr)$, any blocking pair must contain a vertex from $V(\mathcal{C}) \cup \{x\}$.

  Since $\rk_x ((M\Delta F) (x)) \le \rk_x (M(x))$ and $M$ complies with $\bm{h}$, any blocking pair must be contained in $V(\mathcal{C}) \cup \{x\}$.
  Since $(M\Delta F) \cap E(G_d) $ complies with $\bm{h^{d, e}}$ for each $d\in \mathcal{C}$, there is no blocking pair contained inside $Y_d$ for all $d\in \mathcal{C}$.
  Thus, any blocking pair contains the vertex $x$, and an edge $\{x, y^d\}$ for a vertex $y^d\in Y_d $ for some $d\in \mathcal{C}$.
  But from the definition of~$\bm{h^{d, e}}$ it follows that $\{x, y^d\}$ is not a blocking pair.
\end{proof}

\paragraph{Non-singleton Classes}
\label{sec:non-singleton}
We are now left with classes $\mathcal{C}$ of light children for which we know $|N(\mathcal{C})| = 2$ (of course, for the fixed node $t$).
In the rest of \Cref{sec:non-singleton}, let $N(\mathcal{C}) = \{x_1, x_2\}$ and let for every child $c \in \mathcal{C}$ vertex $y_i^c \in Y_c$ be the endpoint of the edge in $\cut(c)$ whose other endpoint is $x_i \in X_t$ (thus, we have $\cut(c) = \left\{ \{ x_1, y_1^c \}, \{x_2, y_2^c \} \right\}$).
Our goal to reduce the number of children in each class to at most four.
We will call the classes for which we achieve this goal \emph{good}, while the other classes are \emph{bad}.
We first need some notation.

The following definition of signature with respect to a neighbor of a child will help us when classifying good and bad children.
The idea is to capture when it is possible to require that one of the edges between a node and a child is not preferred by its endpoint contained in the child.
That is, for a child $c$ of $t$ with $\N (c) =  \{x_1, x_2\}$, it stores how $y_{3-i}^c$ can be matched, given that $x_{i}$ prefers $y_{i}^c$ over its partner in a matching.
\begin{definition}
  Let $t$ be a node, let $c$ be a child with $\N(c) = \{x_1, x_2\}$.
  Then $\sig_{x_1}(c)$ denotes be the set of elements $z\in \{\oldminusone, \oldzero, \oldone\}$ such that $\bm{(\oldone, z)}\in \sig (c)$;
  $\sig_{x_2} (c)$ is defined analogously.
  Thus, we have $\sig_{x_i} (c): = \left\{ z \in \{\oldminusone, \oldzero, \oldone\}: \bm{h^{x_i,z}} \in \sig(c) \right\}$, where $\bm{h^{x_i,z}}(\{x_{3-i}, y^c_{3-i}\}) := z$ and $\bm{h^{x_i,z}}(\{ x_i, y_i^c\}) := \oldone$.
\end{definition}

Note that if e.g.\ $\sig_{x_2} (c) = \{ \oldzero \}$, then the edge $\{x_2, y_2^c \}$ will not be blocking for any matching extending the matching stored in $\tau_c[\bm{h}]$ with $\bm{h}(\{x_2, y_2^c \}) = \oldone$ and $\bm{h}(\{ x_1, y^c_1 \}) = \oldzero$.
By the definition of $\sig_{x_2} (c)$ and $\sig (c)$, this matching exists.

Note that the cases $\sig_{x_1} (\mathcal{C}) = \{\oldone\}$ and $\sig_{x_1} (\mathcal{C}) = \{\oldzero, \oldone\}$ are not possible, as any matching complying with $\bm{h} (e) = \oldone$ also complies with $\bm{h} (e) = {\oldminusone}$ (we stress here that the conditions we impose on a matching by setting $\bm{h} (e) = {\oldone}$ are stronger than those imposed by $\bm{h} (e) = \oldminusone$).
Thus, $\sig_{x_1} ( \mathcal{C})$ for a class $\mathcal{C}$ and the vertex $x_1\in N(\mathcal{C})$ can be one of the six following sets:
\begin{enumerate}
  \item $\sig_{x_1} (\mathcal{C})= \emptyset$.
  \item $\sig_{x_1} (\mathcal{C})= \{\oldminusone\}$.
  \item $\sig_{x_1} (\mathcal{C})= \{\oldminusone, \oldzero\}$.
  \item $\sig_{x_1} (\mathcal{C})= \{\oldminusone, \oldone\}$.
  \item $\sig_{x_1} (\mathcal{C})= \{\oldminusone, \oldzero, \oldone\}$.
  \item $\sig_{x_1} (\mathcal{C})= \{\oldzero\}$.
\end{enumerate}

Subsequently, we proceed as follows:
First, we provide a technical lemma giving a sufficient condition to perform certain exchanges of matching edges (\Cref{lcomplie}).
Then, we show this lemma to show that only a restricted set of children can be matched to a vertex from $X_t$ for every class (\Cref{lem:once-nonzero,lem:once-zero,lem:same-child,lem:different-child}).

\subparagraph*{Exchange Lemma}
We now provide a technical lemma which gives a sufficient condition to perform certain exchanges of matching edges.
This then allows us to simplify some further reasoning.
The basic idea of the exchange is that for each vertex $x\in X_t$ matched to a child~$c \in \mathcal{C}$, we can match $x$ to any child $d\in \mathcal{C}$ such that $x$ prefers its neighbor in~$Y_d$ over $Y_c$ as long as this does not introduce a blocking pair inside $V(\mathcal{C}) \cup \{x\}$.
Before doing so, we formally define the notion of a (feasible) exchange.

\begin{definition}
  Let $t\in V(T)$ be a non-leaf node.
  Let $\mathcal{C}$ be a class of light children, and let $M^*$ be a matching complying with some vector $\bm{h}\in \{\oldminusone, \oldzero, \oldone\}^{\cut (t)}$. Assume that $F= M^* \cap \cut (\mathcal{C})$.
  Let $F' \subseteq \cut (\mathcal{C})$ such that~$F$ contains an edge incident to an $x\in \N (\mathcal{C})$ if and only if $F'$ contains an edge incident to $x$.

  Then, for each $x\in \N(\mathcal{C})$, let $M' (x) := \begin{cases}
                         v & \text{ if } \{v, x\} \in F'\\
                         M^* (x) & \text{ otherwise}
                       \end{cases}$.

  For each $e = \{v, x\}\in \cut (c) \cap \delta (x)$ for some $c\in \mathcal{C}$ and $x\in \N (\mathcal{C})$, let \[\bm{h^c} (e) := \begin{cases}
    \oldone & \text{ if } \rk_x (v) < \rk_x (M' (x)),\\
    \oldzero & \text{ if } e\in F',\\
    \oldminusone & \text{ otherwise.}
  \end{cases}\]

  \emph{Matching $M$ arises from $M^*$ by exchanging the edges $F$ with $F'$}: Delete all edges with one endpoint in $Y_c$ for some $c\in \mathcal{C}$ and add for each $c\in \mathcal{C}$, the matching stored in $\tau_c [\bm{h^c}]$.
  See \Cref{fig:example-exchange} for an example.

  Exchanging $F$ and $F'$ is called \emph{feasible} if $\tau_c [\bm{h^c}] \neq \square$ for all $c\in \mathcal{C}$.
\end{definition}

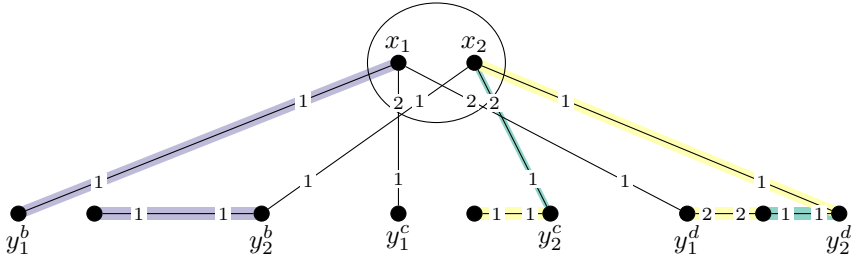
\begin{figure}
  \begin{center}
    \begin{tikzpicture}[
  every node/.style={inner sep=1pt},
  proc/.style={shape=ellipse, draw}
]
      \node[vertex, label=90:$x_1$] (x1) at (0, 0) {};
      \node[vertex, label=90:$x_2$] (x2) at (1, 0) {};
      \node[fit=(x1)(x2), proc, inner ysep=3ex] (t) {};

      \node[vertex, label=270:$y_1^c$] (y1c) at (0,-2) {};
      \node[vertex, label = 270:$y_2^c$] (y2c) at ($(y1c) + (2, 0)$) {};
      \node[vertex] (zc) at ($(y1c) + (1, 0)$) {};
      \node[vertex, label=270:$y_1^d$] (y1d) at ($(y1c) + (3.8,0)$) {};
      \node[vertex, label = 270:$y_2^d$] (y2d) at ($(y1d) + (2, 0)$) {};
      \node[vertex] (zd) at ($(y1d) + (1, 0)$) {};
      \node[vertex, label=270:$y_1^b$] (y1b) at ($(y1c) + (-5,0)$) {};
      \node[vertex, label = 270:$y_2^b$] (y2b) at ($(y1b) + (3.2, 0)$) {};
      \node[vertex] (zb) at ($(y1b) + (1,0)$) {};

      \draw (y1b) edge node[pos=0.2, fill=white, inner sep=1pt] {\scriptsize $1$} node[pos=0.76, fill=white, inner sep=1pt] {\scriptsize $1$} (x1);
      \draw (y1c) edge node[pos=0.2, fill=white, inner sep=1pt] {\scriptsize $1$} node[pos=0.76, fill=white, inner sep=1pt] {\scriptsize $2$} (x1);
      \draw (y1d) edge node[pos=0.2, fill=white, inner sep=1pt] {\scriptsize $1$} node[pos=0.76, fill=white, inner sep=1pt] {\scriptsize $2$} (x1);

      \draw (y2b) edge node[pos=0.2, fill=white, inner sep=1pt] {\scriptsize $1$} node[pos=0.76, fill=white, inner sep=1pt] {\scriptsize $1$} (x2);
      \draw (y2c) edge node[pos=0.2, fill=white, inner sep=1pt] {\scriptsize $1$} node[pos=0.76, fill=white, inner sep=1pt] {\scriptsize $2$} (x2);
      \draw (y2d) edge node[pos=0.2, fill=white, inner sep=1pt] {\scriptsize $1$} node[pos=0.76, fill=white, inner sep=1pt] {\scriptsize $1$} (x2);

      \draw (y2b) edge node[pos=0.2, fill=white, inner sep=1pt] {\scriptsize $1$} node[pos=0.76, fill=white, inner sep=1pt] {\scriptsize $1$} (zb);
      \draw (y2c) edge node[pos=0.2, fill=white, inner sep=1pt] {\scriptsize $1$} node[pos=0.76, fill=white, inner sep=1pt] {\scriptsize $1$} (zc);
      \draw (y2d) edge node[pos=0.2, fill=white, inner sep=1pt] {\scriptsize $1$} node[pos=0.76, fill=white, inner sep=1pt] {\scriptsize $1$} (zd);

      \draw (y1d) edge node[pos=0.2, fill=white, inner sep=1pt] {\scriptsize $2$} node[pos=0.76, fill=white, inner sep=1pt] {\scriptsize $2$} (zd);

      \begin{scope}[on background layer]
          \newcommand{\colorBetweenTwoNodes}[3]{
            \fill[#1] ($(#2) + (0, .08)$) to ($(#2) - (0, .08)$) to ($(#3) - (0,.08)$) to ($(#3) + (0,.08)$) -- cycle;
        }
          \colorBetweenTwoNodes{mylila}{y1b}{x1}
          \colorBetweenTwoNodes{mylila}{y2b}{zb}
          \colorBetweenTwoNodes{mygreen}{y2c}{x2}
          \colorBetweenTwoNodes{mygreen}{y2d}{zd}

          \colorBetweenTwoNodes{myyellow}{y2d}{x2}
          \colorBetweenTwoNodes{myyellow}{y1d}{zd}
          \colorBetweenTwoNodes{myyellow}{y2c}{zc}

      \end{scope}
    \end{tikzpicture}
  \end{center}
  \caption{An example for a feasible exchange.
  Matching $M^*$ contains the purple and the green edges, and matching $M$ arising from $M^*$ by exchanging the edge $\{x_2, y_2^c\}$ with $\{x_2, y_2^d\}$ contains the purple and the yellow edges (i.e., $M$ arises from $M^*$ by replacing the green edges by the yellow edges).}
  \label{fig:example-exchange}
\end{figure}

\begin{lemma}[Exchange Lemma]\label{lcomplie}
  Let $M^*$ be an $\bm{h}$-embedding of $\mathcal{N}$, and let $M$ be the matching arising from $M^*$ by exchanging the edges $F$ with $F'$ for some $F, F'\subseteq \cut (\mathcal{C})$ for a class $\mathcal{C}$ with $(\bigcup_{e\in F} e)\cap \N (\mathcal{C}) = (\bigcup_{e\in F'} e) \cap \N (\mathcal{C})$.

  If the exchange is feasible and $\rk_x (M (x) ) \le \rk_x (M^*(x))$ for all $x\in \N (\mathcal{C})$, then the matching $M$ is an $\bm{h}$-embedding of $\mathcal{N}$.
\end{lemma}

\begin{proof}
   The matchings $M$ and $M^*$ only differ on $S:=\N (\mathcal{C}) \cup \bigcup_{c\in \mathcal{C}} Y_c$.
   As $\N (\mathcal{C})$ is matched not worse in $M$ than in $M^*$, every blocking pair must contain a vertex from $V(\mathcal{C})$.
   As the exchange is feasible, there is no blocking pair in $Y_c$ for all $c\in \mathcal{C}$.
   By the definition of $\bm{h^c}$, no vertex from $\N (\mathcal{C})$ forms a blocking pair together with a vertex from $V(\mathcal{C})$.
   Thus, there are no blocking pairs in $Y_t$.

   The other two conditions for $M$ complying with $\bm{h}$ directly follow from $M^*$ complying with $\bm{h}$.

   Clearly, $M$ fulfills that there is an edge $\{x, x^\mathcal{D}\}\in \mathcal{N}$ if and only if $\{x, y^d\}\in M$ for some $y^d\in Y_d$ for some child~$d\in \mathcal{D}$ for all classes~$\mathcal{D}$, and thus $M$ is an $\bm{h}$-embedding of $\mathcal{N}$.
\end{proof}

Using \Cref{lcomplie}, we will now show for most classes of light children that we can restrict them to contain at most four children, i.e., to be good.
To ease the application of \Cref{lcomplie}, we define the \emph{vector $\bm{h^{M, c}}$ associated with a child $c$ and a matching $M$} by
\[
\bm{h^{M, c}} (\{x_i, y_i^c \}) :=
  \begin{cases}
    \oldone& \text{ if } \rk_{x_i}(y_i^c) < \rk_{x_i} (M(x_i)),\\
    \oldzero & \text{ if } \{x_i, y_i^c \}\in M,\\
    \oldminusone & \text{ otherwise}.
  \end{cases}
\]
Given the matching $M$ on $X_t$, this matching can be extended to a child $c$ if and only if $\tau_c [ \bm{h^{M, c}}] \neq \square$.

While we will not show for many classes having two neighbors in $X_t$ but being matched to only one of them that they are good, we can still prove some structure on these classes which enables us to reduce them to a \textsc{2-SAT} instance in \Cref{sec:2-SAT}.

\subparagraph*{Once Matched Children}

For the forthcoming \Cref{lem:once-nonzero,lem:once-zero}, we assume that for the considered class $\mathcal{C}$ it holds that $N(\mathcal{C}) = \{x_1, x_2\}$ and $\mathcal{N} (x_1) = \mathcal{C} = \mathcal{N} (x_2)$.
Before we show how to deal with such classes, we first need some notation.
For any vector~$\bm{h}\in \{\oldminusone, \oldzero, \oldone\}^{\cut (c)}$ for some $c\in \mathcal{C}$, the first coordinate will always denote an edge incident to~$x_1$, and the second coordinate an edge incident to $x_2$.
For a child $c \in \mathcal{C}$, let $\rk_{x_1} (c) := \rk_{x_1}(y_1^c)$ (recall that $y_1^c\in Y_c$ is the neighboring vertex of~$x_1$ in $Y_c$, i.e., the vertex from $Y_c$ such that $\{x_1, y_1^c\}\in E(G)$);
we define $\rk_{x_2} (c)$ analogously.
We say that a matching $M$ is \emph{valid} in the node $t$ for the matching~$\mathcal{N}$ if $M(x) \in \mathcal{N}(x)$ for all~$x \in X_t$.

We now show how to handle a class $\mathcal{C}$ with $N (\mathcal{C}) = \{x_1 , x_2\}$ and $\mathcal{N} (x_1) = \mathcal{C}$.
First, we deal with the case that $\sig_{x_1} (\mathcal{C}) \neq \{\oldzero\}$.

\begin{lemma}\label{lem:once-nonzero}
  Consider a class $\mathcal{C}$ with $N(\mathcal{C} = \{x_1, x_2\}$ and $\mathcal{N} (x_1) = \mathcal{C}  \neq \mathcal{N} (x_2)$.
  Then one can compute in polynomial time a set ${\mathcal{C} }'= \{c_1, \dots, c_\ell\} \subseteq \mathcal{C}$ such that $\rk_{x_1} (c_i) < \rk_{x_2} (c_{i+1})$ and $\rk_{x_2} (x_i ) < \rk_{x_2} (x_{i+1})$ for all $i\in [\ell -1]$ and the following holds.
  If there exists an $\bm{h}$-embedding~$M$ of $N$, then there exists one in which $x_1$ is matched to some $c\in {\mathcal{C}}' \cup \{ c^*\}$.
\end{lemma}

\begin{proof}
  We start with $\mathcal{C'} = \mathcal{C}$ and successively delete elements from $\mathcal{C'}$ until $\mathcal{C'}$ fulfills the first condition.
  If $\mathcal{C'}$ does not fulfill the first condition, then there exists some $c \neq d\in \mathcal{C'}$ with $\rk_{x_1} ( y_1^c) \le \rk_{x_1} (y_1^d)$ but $\rk_{x_2} (y_2^c) \ge \rk_{x_2} (y_2^d)$.
  We claim that we can delete~$d$ from~$\mathcal{C'}$.
  To see this, assume that $M$ is an $\bm{h}$-embedding of $\mathcal{N}$ containing edge $\{ x_1,y_1^d \}$.
  We show that the matching~$M'$ arising from $M$ by exchanging $\{ x_1, y_1^d \}$ with $\{ x_1, y_1^c \}$ also is an $\bm{h}$-embedding of~$\mathcal{N}$.
  To this end, let $F = (M \cap (E(G_c) \cup E(G_d) \cup \cut(d))) \cup \tau_c[\bm{(\oldzero, \oldminusone)}] \cup \tau_d[\bm{(\oldminusone, \oldminusone)}]$.
  Clearly, the children in $\mathcal{C}$ other than $c$ and $d$ remain unaffected by this change.
  We have that $\bm{h^{M, d}} (\{x_1, y_1^d\}) = \oldzero$.
  Thus, we have $\bm{(\oldzero, h^{M, d} (\{x_1, y_1^d\}))} \in \sig (\mathcal{C})$.
  As $\rk_{x_2} (y_2^d) \le \rk_{x_2} (y_2^c)$, it follows that $\bm{h^{M', c}} \in \sig (\mathcal{C})$.
  As $\rk_{x_1} ( y_1^d) \ge \rk_{x_1} (y_1^c)$, we have $\bm{h^{M', d}} (\{x_1, x_1^d) = \oldminusone$.
  If $\bm{h^{M', d}} (\{x_2, y_2^d\}) = \oldminusone$, then $\bm{h^{M', d}} = \bm{(\oldminusone, \oldminusone)} \in \sig (\mathcal{C})$ as $\bm{h^{M, c}}\in \sig (\mathcal{C})$.
  Otherwise we have $\bm{h^{M', d}} (\{x_2, y_2^d\}) = \oldone$.
  Note that $\bm{h^{M, d}} (\{x_2, y_2^d\}) = \bm{h^{M', d}} (\{x_2, y_2^d\})$.
  Thus, we have $\oldzero \in \sig_{x_1} ( \mathcal{C})$.
  As $\sig_{x_1}  (\mathcal C) \neq \{\oldzero\}$, it follows that $\oldminusone \in \sig_{x_1} (\mathcal{C})$, and thus, $\bm{h^{M', d}} \in \sig (\mathcal{C})$.
  It follows that the exchange is feasible.
  Since $x_1$ does not prefer $M(x_1)$ to $M' (x_1)$, \Cref{lcomplie} implies that $M'$ is also an $\bm{h}$-embedding of $\mathcal{N}$.

  The pair $\{c, d\}$ can clearly be found in polynomial time, and thus the total running time is also polynomial.
\end{proof}

We now take care of the remaining case of classes~$\mathcal{C}$ with two neighbors which are matched to only one, say $x_1$, of them, namely $\sig_{x_1} (\mathcal{C}) = \{\oldzero\}$.

\begin{lemma}\label{lem:once-zero}
  Consider a class $\mathcal{C}$ with $N(\mathcal{C} = \{x_1, x_2\}$ and $\mathcal{N} (x_1) = \mathcal{C}  \neq \mathcal{N} (x_2)$.
  Then one can compute in polynomial time a child $c^*\in \mathcal{C}$ and a set ${\mathcal{C} }'= \{c_1, \dots, c_\ell\} \subseteq \mathcal{C}$ such that $\rk_{x_1} (c_i) < \rk_{x_2} (c_{i+1})$ and $\rk_{x_2} (x_i ) < \rk_{x_2} (x_{i+1})$ for all $i\in [\ell -1]$ and the following holds.
  If there exists an $\bm{h}$-embedding $M$ of $N$, then there exists one in which $x_1$ is matched to some $c\in {\mathcal{C}}' \cup \{ c^*\}$.
\end{lemma}

\begin{proof}
  Let $c^*\in \mathcal{C}$ such that $y_2^{c^*} \in \best (\mathcal{C}, x_2)$.
  
  We proceed similar to the proof of \Cref{lem:once-nonzero}.
  We start with $\mathcal{C'} = \mathcal{C}\setminus \{c^*\}$ and successively delete elements from $\mathcal{C'}$ until $\mathcal{C'}$ fulfills the first condition.
  If $\mathcal{C'}$ does not fulfill the first condition, then there exists some $c \neq d\in \mathcal{C'}$ with $\rk_{x_1} ( y_1^c) \le \rk_{x_1} (y_1^d)$ but $\rk_{x_2} (y_2^c) \ge \rk_{x_2} (y_2^d)$.
  We claim that we can delete~$d$ from~$\mathcal{C'}$.
  To see this, assume that $M$ is an $\bm{h}$-embedding of~$\mathcal{N}$ containing edge $\{ x_1,y_1^d \}$.
  We show that the matching~$M'$ arising from $M$ by exchanging~$\{ x_1, y_1^d \}$ with $\{ x_1, y_1^c \}$ also is an $\bm{h}$-embedding of~$\mathcal{N}$.
  To this end, let $F = (M \cap (E(G_c) \cup E(G_d) \cup \cut(d))) \cup \tau_c[\bm{(\oldzero, \oldminusone)}] \cup \tau_d[\bm{(\oldminusone, \oldminusone)}]$.
  Clearly, the children in $\mathcal{C}$ other than $c$ and $d$ remain unaffected by this change.
  We have that $\bm{h^{M, d}} (\{x_1, y_1^d\}) = \oldzero$.
  Thus, we have $\bm{(\oldzero, h^{M, d} (\{x_1, y_1^d\}))} \in \sig (\mathcal{C})$.
  As $\rk_{x_2} (y_2^d) \le \rk_{x_2} (y_2^c)$, it follows that $\bm{h^{M', c}} \in \sig (\mathcal{C})$.
  As $\rk_{x_1} ( y_1^d) \ge \rk_{x_1} (y_1^c)$, we have $\bm{h^{M', d}} (\{x_1, x_1^d) = \oldminusone$.
  If $\bm{h^{M', d}} (\{x_2, y_2^d\}) = \oldminusone$, then $\bm{h^{M', d}} = \bm{(\oldminusone, \oldminusone)} \in \sig (\mathcal{C})$ as $\bm{h^{M, c}}\in \sig (\mathcal{C})$.
  Otherwise we have $\bm{h^{M', d}} (\{x_2, y_2^d\}) = \oldone$ and we will derive a contradiction, showing that this case cannot occur.
  Since $c^*\in \best (\mathcal{C}, x_2)$, it follows that $\oldone = \bm{h^{M', c^*}} (\{x_2, y_2^{c^*}\}) = \bm{h^{M, c^*}} (\{x_2, y_2^{c^*}\})$.
  However, since $d \neq c^*$, we have $\bm{h^{M, c^*} } (\{x_1, y_1^{c^*}\}) \neq \oldzero$.
  Since $\sig_{x_1} ( \mathcal{C}) = \{ \oldzero\}$, it follows that $\bm{h^{M, c^*}} \notin \sig (\mathcal{C})$, contradicting that $M$ is an $\bm{h}$-embedding.
  It follows that the exchange is feasible.
  Since $x_1$ does not prefer~$M(x_1)$ to $M' (x_1)$, \Cref{lcomplie} implies that $M'$ is also an $\bm{h}$-embedding of $\mathcal{N}$.

  The pair $\{c, d\}$ can clearly be found in polynomial time, and thus the total running time is also polynomial.
\end{proof}

For each class $\mathcal{C}$ with $N( \mathcal{C} ) = \{x_1, x_2\}$, $\mathcal{N} (x_1) = \mathcal{C} \neq \mathcal{N} (x_2)$, and $\sig_{x_1} (\mathcal{C}) = \{ \oldzero\}$, we will guess whether $x_1$ will be matched to $c^*$, $c_1$, or another children of $\mathcal{C}$.
In the former cases, $\mathcal{C}$ will be a good class, while in the later case, it will be a bad class.

\subparagraph*{Doubly Matched Children}
We now turn to classes~$\mathcal{C}$ with two neighbors~$N( \mathcal{C} ) = \{x_1, x_2\}$ such that both $x_1$ and $x_2$ are matched to $\mathcal{C}$, i.e., $\mathcal{N} (x_1) = \mathcal{C} = \mathcal{N} (x_2)$.
First, we deal with the case that $x_1$ and $x_2$ are matched to the same child of $\mathcal{C}$.
\begin{lemma}\label{lem:same-child}
  Let $\mathcal{C}$ be a class with $\mathcal N (x_1) = \mathcal{C} = \mathcal{N } (x_2)$.
  Then one can compute in polynomial time a set ${\mathcal C}' = \{c_1, \dots, c_\ell\} \subseteq \mathcal C$ with $\rk_{x_1} (c_i) < \rk_{x_1} (x_{i+1})$ and $\rk_{x_2} (c_i) > \rk_{x_2} (x_{i+1})$ for all~$i \in [\ell -1]$ such that the following holds.
  If there exists an $\bm{h}$-embedding matching~$M$ with $\{x_1, y_1^c\} \in M$ and $\{x_2, y_2^c\} \in M$ for some $c\in \mathcal{C}$ if and only if there exists an $\bm{h}$-embedding containing $\{x_1, y_1^{c'}\}$ and $\{x_2, y_2^{c'}\}$ for some $c' \in 
  {\mathcal{C}}'$.
\end{lemma}
\begin{proof}
  We start the shrinkage of $\mathcal{C}$ by setting $\mathcal{C'} = \mathcal{C}$ and we successively delete elements from~$\mathcal{C'}$ until $\mathcal{C'}$ fulfills the first condition of the lemma.
  As long as this is not true, there exist two children $c\neq d\in \mathcal{C'}$ with $\rk_{x_1} (y_1^c) \le \rk_{x_1} (y_1^d)$ and $\rk_{x_2} (y_2^c) \le \rk_{x_2} (y_2^d)$.
  Such children $c$ and~$d$ can clearly be found in polynomial time if they exist and if we find such a pair of children, we delete $d$ from $\mathcal{C}'$.

  Now, let $d \in \mathcal{C} \setminus \mathcal{C}'$ and let $c$ be the node because of which we have deleted $d$ in the above procedure.
  Let $M$ be an $\bm{h}$-embedding of $\mathcal{N}$ containing the edges $\{ x_1,y_1^d \}$ and $\{ x_2, y_2^d \}$.
  We claim that if we exchange edges $\{ x_1,y_1^d \}$ and $\{ x_2, y_2^d \}$ with $\{ x_1,y_1^c \}$ and $\{ x_2, y_2^c \}$, then the resulting matching $\hat{M}$ also is an $\bm{h}$-embedding of $\mathcal{N}$.
  Note that the exchange is feasible since both $\bm{h^{\hat{M}, c}} = \bm{(\oldzero, \oldzero)} \in \sig (\mathcal{C})$ and $\bm{h^{\hat{M}, d}} = \bm{(\oldminusone, \oldminusone)}\in \sig (\mathcal{C})$.
  Moreover, the two table entries used are nonempty.
  Now, observe that both $x_1$ and $x_2$ do not prefer their partner in $M \Delta F$ to their partner in $M$.
  Since there is no other change for the other children in $\mathcal{C}$, we are done.
\end{proof}

We now turn to the case that $x_1$ and $x_2$ are matched to different children inside class~$\mathcal{C}$.
The following lemma shows that $\mathcal{C}$ is good in this case.

\begin{lemma}\label{lem:different-child}
  Let $\mathcal{C}$ be a class with $\mathcal{N} (x_1) = \mathcal{C} = \mathcal{N} (x_2)$.
  Then one can compute in polynomial time a set of $\hat{\mathcal{C}}  \subseteq \mathcal{C}$ with $|\hat{\mathcal{C}}| \le 4$ such that the following holds.
  If there exists an $\bm{h}$-embedding~$M$ of~$\mathcal{N}$ matching such that $\{x_1, y_1^{c_1^*}\}\in M$ and $\{x_2, y_2^{c_2^*}\} \in M$ for $c_1^* \neq c_2^* \in \mathcal{C}$, then there exists an $\bm{h}$-embedding~$M'$ of $\mathcal{N}$ with $\{x_1, y_1^{c}\}\in M'$ and $\{x_2, y_2^{d}\}$ for some $c, d \in \hat{\mathcal{C}}$.
\end{lemma}

\begin{proof}
  Since $c_1^* \neq d_1^*$, we have $\bm{(\oldzero, \oldminusone)} \in \sig (\mathcal{C})$ and $\bm{(\oldminusone, \oldzero)} \in \sig (\mathcal{C})$.

  We now show that we can find in polynomial time four edges $\{x_1, v_1\}$, $\{x_2, v_2\}$, $\{ x_1, w_1\}$, and~$\{x_2, w_2\}$ with $v_i, w_i\in V(\mathcal{C})$ such that the following holds:
  The matching~$M_v$ arising from~$M$ by exchanging the edges $\{\{x_1, y_1^{c_1^*}\}, \{x_2, y_2^{c_2^*}\}\}$ with the two edges~$\{\{ x_1, v_1\}, \{x_2, v_2\}\}$ or the matching $M_w$ arising from~$M$ by exchanging the edges $\{\{x_1, y_1^{c_1^*}\}, \{x_2, y_2^{c_2^*}\}\}$ with $\{\{ x_1, w_1\}, \{x_2, w_2\}\}$ is an $\bm{h}$-embedding of $\mathcal{N}$.
  From this, the lemma clearly follows.

  Let $y_1^{c_1}\in \best (\mathcal{C}, x_1)$, and $y_2^{c_2}\in \best (\mathcal{C}, x_2)$.
  If possible, then choose them such that $c_1 \neq c_2 $.

  If $c_1 \neq c_2$ or $\tau_c [\bm{(\oldzero, \oldzero)}] \neq \square$ for all $c\in \mathcal{C}$, then let $\hat{M}$ be the matching arising from~$M$ by exchanging the edges $ \{\{x_1, y_1^{c_1^*}\}, \{x_2, y_2^{c_2^*}\}\}$ with $\{\{x_1, y_1^{c_1}\}, \{x_2, y_2^{c_2}\}\}$.
  This exchange is feasible, as we have $\bm{h^{M, c}} \in \{\oldminusone, \oldzero\}^2$, and if $c_1 \neq c_2$, we have $\bm{h^{M, c}} \in \{\bm{(\oldminusone, \oldminusone)}, \bm{(\oldminusone, \oldzero)}, \bm{(\oldzero, \oldminusone)}\}$ for all~$c\in\mathcal{C}$.
  Thus, $\hat{M}$ is an $\bm{h}$-embedding of $\mathcal{N}$ by \Cref{lcomplie}.

  Otherwise, we have $c_1 = c_2$, $\best(\mathcal{C}, x_1) = \best(\mathcal{C}, x_2) = \{y_1^{c_1}\}$, and $\tau_c [\bm{(\oldzero, \oldzero)}] = \square$ for all~$c\in \mathcal{C}$, and we distinguish several cases.

  \textbf{Case 1:}
  Suppose we have $\sig_{x_1} (\mathcal{C}) = \sig_{x_2} (\mathcal{C}) = \{\oldminusone, \oldone\}$.
  Then, we claim that no matching~$M^*$ complying with $\bm{h}$ can contain edge $\{x_1, y_1^{c_1}\}$ (and by symmetry also not~$\{x_2, y_2^{c_1}\}$):
  Assume for the sake of contradiction that $M^*(x_1) = y_1^{c_1}$.
  Now, $\tau_c [\bm{(\oldzero, \oldzero)}] = \square$ implies that~$M^*(x_2) \neq y_2^{c_1}$.
  Thus, $\bm{h^{M^*,c_1}} = \bm{(\oldzero, \oldone)}$, which implies $\oldzero \in \sig_{x_2}(\mathcal{C})$, contradicting the assumption $\sig_{x_2} (\mathcal{C}) = \{\oldminusone, \oldone\}$.

  As we know that $x_1$ and $x_2$ cannot be matched to $Y_{c_1}$, we mark child $c_1$ as unmatchable.
  We define ${y}_1^{\hat{c}_1}$ and~${y}_2^{\hat{c}_2}$ to be an unmarked vertex ranked best by $x_1$ and~$x_2$, respectively.
  If we cannot choose $\hat{c}_1 \neq \hat{c}_2$, then we mark $\hat{c}_1$, and recompute $\hat{c}_1$ and $\hat{c}_2$ among all unmarked vertices.
  If eventually all children inside $\mathcal{C}$ are marked, then no $\bm{h}$-embedding of $\mathcal{N}$ exists, since for any marked child~$c$, the edge $\{x, x^c\}$ cannot belong to any $\bm{h}$-embedding of $\mathcal{N}$ by the same arguments as for $c_1$.
  Otherwise, we have $\hat{c}_1 \neq \hat{c}_2$ at some point.
  We claim that the matching~$\hat{M}$ arising from $M$ by exchanging the edges $\{\{x_1, y_1^{c_1^*}\}, \{x_2, y_2^{c_1^*}\}\}$ with $\{\{x_1, y_1^{\hat{c}_1}\}, \{x_2, y_2^{\hat{c}_2}\}\}$ is an $\bm{h}$-embedding of $\mathcal{N}$.
  First note that $\tau_c [\bm{(\oldone, \oldone)}] \neq \square$ for all $c\in \mathcal{C}$ since $h^{M, c_1} = \bm{(\oldone, \oldone)}$.
  Thus, we have $\tau_c [ \bm{h^{\hat{M}, c}}] \neq \square$ for all children $c \in \mathcal{C} \setminus\{\hat{c}_1, \hat{c}_2\}$ since $\tau_c [\bm{(\oldone, \oldone)}] \neq \square$.
  For $c \in \{\hat{c}_1, \hat{c}_2\}$, this holds by the conditions $\bm{(\oldzero, \oldminusone)}\in \sig (\mathcal{C})$ and $\bm{(\oldminusone, \oldzero)}\in \sig (\mathcal{C})$.

  \bfseries Case 2: \mdseries We have $\sig_{x_2} (\mathcal{C}) \in \{\emptyset, \{\oldminusone\}\}$ (symmetrically, $\sig_{x_1} (\mathcal{C})  \in \{\emptyset, \{\oldminusone\}\}$).
  Any matching~$M^*$ complying with $\bm{h}$ must contain $\{x_2, y_2^{c_1}\}$:
  if $M^*$ does not contain $\{x_2, y_2^{c_1}\}$, then $\rk_{x_2} (y_2^{c_1}) < \rk_{x_1} (M^*(x_2))$ and thus $\bm{h^{M^*, c_1}} (\{y_2, y_2^{c_1}\}) = \oldone$.
  However, since $\best (\mathcal{C}, x_1) = \{c_1\}$, we have $\bm{h^{M^*, c_1}} (\{x_1, y_1^{c_1}\})\in \{\oldzero, \oldone\}$, and thus $\tau_{c_1} [\bm{h^{M^*, c_1}}] = \square$.

  So let $y_1^{d_1}$ be a vertex ranked second-best by $x_1$ among all vertices in $\N(x_1) \cap \bigl( \bigcup_{c\in \mathcal{C}} Y_c\bigr)$, i.e., $y_1^{d_1} $ minimizes $\rk_{x_1} (v)$ among all $v\in \{y_1^c : c\in \mathcal{C}\}\setminus \best (\mathcal{C}, x_1)$.
  We claim that the matching~$\hat{M}$ arising from $M$ by replacing $\{x_1, y_1^{c_1^*}\}$ and $\{x_1, y_1^{d_1}\}$ complies with $\bm{h}$.
  For all $c\in \mathcal{C}\setminus\{c_1, d_1\}$, we have $\bm{h^{\hat{M}, c}} = \bm{(\oldminusone, \oldminusone)}$ since $\best (\mathcal{C}, x_1) = \{y_1^{c_1}\} = \best (\mathcal{C}, x_2)$.
  We also have $\bm{h^{\hat{M}, d_1}} = \bm{(\oldzero, \oldminusone)}$ (because $\hat{M} (x_1) = y_1^{d_1}$ and $\hat{M} (x_2) \in \best (\mathcal{C}, x_2) = \{y_2^{c_1}\}$ but $y_2^{d_1} \notin \best (\mathcal{C}, x_2$) and~$\bm{h^{\hat{M}, c_1}} = \bm{(\oldone, \oldzero)}$ (because $y_1^{c_1} \in \best (\mathcal{C}, x_1)$ but $y_1^{d_1} \notin \best (\mathcal{C}, x_1)$ and $\hat{M} (x_2) = y_2^{c_1}$).
  However, also $\bm{h^{M, c_1}} = \bm{(\oldone, \oldzero)}$, and since $M$ is an $\bm{h}$-embedding of $\mathcal{N}$, the exchange is feasible, and thus by \Cref{lcomplie}, the matching~$\hat{M}$ is an $\bm{h}$-embedding of $\mathcal{N}$.

  \bfseries Case 3:
  \mdseries We have $\sig_{x_1} (\mathcal{C}) \in \{\{\oldminusone, \oldzero\}, \{\oldminusone, \oldzero, \oldone\}, \{\oldzero\}\}$ and $\sig_{x_2} (\mathcal{C})\in \{\{\oldminusone, \oldzero\},\allowbreak \{\oldminusone, \oldzero, \oldone\}, \allowbreak\{\oldzero\}\}$.
  Then assume without loss of generality that $M$ does not contain $\{x_1, y_1^{c_1}\}$ (this can be done as $M$ cannot contain both the edges $\{x_1, y_1^{c_1}\}$ and $\{x_2, y_2^{c_1}\}$ since $\tau_c [\bm{(\oldzero,\oldzero)}] = \square$).
  This implies that $\bm{h^{M, c_1}} (\{x_1, y_1^{c_1}\}) = \oldone$. Let $y_1^{d_1}$ be a vertex ranked second-best by $x_1$ among all vertices in $\N(x_1) \cap \bigl( \bigcup_{c\in \mathcal{C}} Y_c\bigr)$.
  We claim that the matching~$\hat{M}$ arising from $M$ by exchanging $\{\{x_1, y_1^{c_1^*}\}, \{x_2, y_2^{c_2^*}\}\}$ with $\{ \{x_1, y_1^{d_1}\}, \{x_2, y_2^{c_1}\}\}$ is an $\bm{h}$-embedding of~$\mathcal{N}$.

  We have $\rk_{x_2} (y_2^{c_1}) \le \rk_{x_2} (M(x_2))$ since $y_2^{c_1} \in \best (\mathcal{C}, x_2)$.
  Furthermore, we have $\rk_{x_1} (y_1^{d_1}) \le \rk_{x_1} (M(x_1))$ by the definition of $y_1^{d_1}$ and the assumption that $\{x_1, y_1^{c_1}\}\notin M$.
  We have $\bm{h^{\hat{M}, c_1}} = \bm{(\oldone, \oldzero)}$, $\bm{h^{\hat{M}, d_1}} = \bm{(\oldzero, \oldminusone)}$ and $\bm{h^{\hat{M}, c}} = \bm{(\oldminusone, \oldminusone)}$ for all $c\in \mathcal{C}\setminus \{c_1, d_1\}$.
  As $\oldzero\in \sig_{x_2} (\mathcal{C})$, we have $\tau_c [ \bm{(\oldone, \oldzero)}] \neq \square$ for all~$c\in \mathcal{C}$.
  As $\bm{h^{M, c_1^*}} = \bm{(\oldzero, \oldminusone)}$, we have $\tau_c [\bm{(\oldzero, \oldminusone})] \neq\square$ for all $c\in \mathcal{C}$, and thus, the exchange is feasible.

  \bfseries Case 4: \mdseries We have $\sig_{x_1} (\mathcal{C}) =\{\oldminusone, \oldone\}$ and $\sig_{x_2} (\mathcal{C}) \in \{\{\oldminusone, \oldzero\}, \{\oldminusone, \oldzero, \oldone\}, \{\oldzero\}\}$ (the case where $\sig_{x_1} (\mathcal{C})$ and $\sig_{x_2} (\mathcal{C})$ are swapped follows by symmetry).
  Matching $M$ cannot contain the edge $\{x_2, y_2^{c_1}\}$, as this implies $\{x_1, y_1^{c_1}\}\notin M$ and thus $\rk_{x_1} (y_1^{c_1}) < \rk_{x_1} (M(x_1))$, and therefore $\bm{h^{M, c_1}} = \bm{(\oldone, \oldzero)}$, but $\tau_{c_1} [\bm{(\oldone, \oldzero)}] = \square$ as $\oldzero\notin\sig_{x_1} (\mathcal{C}) $.

  Let $y_2^{d_2}$ be a vertex ranked second-best by $x_2$ among all vertices in $\N(x_2) \cap \bigl( \bigcup_{c\in \mathcal{C}} Y_c\bigr)$.
  We claim that the matching $\hat{M}$ arising from $M$ by exchanging edges $\{\{x_1, y_1^{c_1^*}\}, \{x_2, y_2^{c_2^*}\}\}$ with $\{\{x_1, y_1^{c_1}\}, \{x_2, y_2^{d_2}\}\}$ is an $\bm{h}$-embedding of $\mathcal{N}$.
  We have $\rk_{x_1} (y_1^{c_1}) \le \rk_{x_1} (M(x_1))$ and $\rk_{x_2} (y_2^{d_2}) \le \rk_{x_2} (M(x_2))$.
  For all~$c\in \mathcal{C}\setminus\{c_1, d_2\}$, we have $\bm{h^{\hat{M}, c}} = \bm{(\oldminusone, \oldminusone)}$.
  We have $\bm{h^{\hat{M}, d_2}} = \bm{(\oldminusone, \oldzero)}$ and $\bm{h^{\hat{M}, c_1}} = \bm{(\oldzero, \oldone)}$.
  As $\oldzero\in \sig_{x_2} (\mathcal{C})$, we have $\tau_{c_1} [\bm{(\oldzero, \oldone)}]\neq \square$.
  Furthermore, we assumed that $\bm{(\oldminusone, \oldzero)} \in \sig ( \mathcal{C})$, implying that $\tau_{d_2}[\bm{(\oldminusone, \oldzero)}]\neq \square$.
  Thus, the exchange is feasible.
  
  In all four cases, we can clearly find a set $\hat{\mathcal{C}}$ fulfilling the lemma in polynomial time.
\end{proof}

  For every class $\mathcal{C}$ with $N (\mathcal{C}) = \{x_1, x_2\}$, $\mathcal{N} (x_1) = \mathcal{C} \neq \mathcal{N} (x_2)$, we guess whether $x_1$ is matched to child $c^*$ or $c_1$ from \Cref{lem:once-zero}, or to another child $c\in \mathcal{C}$.
  In the former two cases, we call $\mathcal{C}$ good.
  Furthermore, for every class $\mathcal{C}$ with $\mathcal{N} (x_1) = \mathcal{C} = \mathcal{N} (x_2)$, we guess whether $x_1$ and $x_2$ are matched to the same child $c\in \mathcal{C}$ or not.
  If so, then we call this class good.
  All other classes $\mathcal{C}'$ with $\mathcal{N} (x) = \mathcal{C}$ for some $x\in X_t$ will be called bad.

 We have shown how to deal with the good classes;
 the following corollary ensures that it suffices to consider $2^{O(k)}$ many matchings from $X_t$ to the good children:

 \begin{corollary}\label{cgc}
   Let $\mathcal{N}$ be a matching in the auxiliary node graph~$H_t$ and let $\mathcal{N}_{\text{good}} \cup \mathcal{N}_{\text{bad}}$ be a partition of $\mathcal{N}$ into edges with their endpoint not in $X_t$ in good and bad children of the node~$t$.
   One can compute in $2^{O(k)}n^{O(1)}$ time a set $\mathcal{M}$ of $2^{O(k)}$ matchings (into the good classes) such that there exists an $\bm{h}$-embedding of $\mathcal{N}$ if and only if there exists an $\bm{h}$-embedding~$M$ of~$\mathcal{N}$ such that $\tilde{M} \cap E(\mathcal{C}_\text{good}) = M\cap  E(\mathcal{C}_\text{good})$ for some $\tilde{M} \in \mathcal{M}$, where $\mathcal{C}_{\text{good}}$ is the set of good children and $E(\mathcal{C}_\text{good}) := \bigcup_{\mathcal{C}\in \mathcal{C}_\text{good}} E(\mathcal{C})$.
 \end{corollary}

 \begin{proof}
   By \Cref{lem:once-zero,lem:different-child}, we can compute in $2^{O(k)} n^{O(1)}$ time a set $\mathcal{M}'$ of $2^{O(k)}$ subsets of $\cut (\mathcal{C}_\text{good})$ such that there exists an $\bm{h}$-embedding $M$ of~$\mathcal{N}$ if and only there exists an $\bm{h}$-embedding $M'$ of~$\mathcal{N}$ with $M' \cap \cut (\mathcal{C}_{\text{good}}) \in \mathcal{M}'$, where $\cut (\mathcal{C}_\text{good}) : = \bigcup_{\mathcal{C}\in \mathcal{C}_\text{good}} \cut (\mathcal{C})$.

   Note that if we know for a good class $\mathcal{C}$ with $\mathcal{N} (x_1) = \mathcal{C}$ which edges from $\cut (\mathcal{C})$ are contained in a stable matching $M$, then we can already determine a set of edges $F$ such that $(M\setminus E(\mathcal{C}) )\cup F$ is a stable matching.
   This is due to the fact that we know $\bm{h^{M, c}} (\{x_1, y_1^c\})$ for all $c\in \mathcal{C}$.
   If we also have $\mathcal{N} (x_2) = \mathcal{C}$ for some $x_2\neq x_1$, then we know $\bm{h^{M, c}}$ for each $c\in \mathcal{C}$.
   Otherwise, we take the matching stored in $\tau_c [ \bm{( \bm{h^{M, c}} (\{x_1, y_1^c\}), \oldone)}]$ if $\tau_c [ \bm{( \bm{h^{M, c}} (\{x_1, y_1^c\}), \oldone)}]\neq \square$, and $\tau_c [ \bm{( \bm{h^{M, c}} (\{x_1, y_1^c\}), \oldminusone)}]$ otherwise.
 \end{proof}

 Due to \Cref{cgc}, we will from now on assume that we fixed a matching $\mathcal{N}$ between~$X_t$ and the classes, and an embedding $M^*\in \mathcal{M}$ of $\mathcal{N}$ into the good classes.

Since ewe already guessed how the vertices from $X_t$ are matched to the heavy and the good children, it remains to compute for each vertex $x\in X_t$ for which have guessed that $x$ is matched to a bad class $\mathcal{C}$ to which child $c\in \mathcal{C}$ vertex $x$ is matched to.
We now reduce this problem to a 2-SAT instance (where 2-SAT is the problem of deciding whether a boolean formula in 2-conjunctive normal form is satisfiable).
\textsc{2-SAT} is known to be solvable in linear time~\cite{ASPVALL1979}.

\subsubsection{Reduction to 2-SAT}
\label{sec:2-SAT}

  In this subsection, we assume that we fixed a matching $\mathcal{N}$ from~$X_t$ to the classes, i.e., that each vertex $x \in X_t$ may only be matched to a vertex from a class~$\mathcal{C}$ if $\{x, x^ \mathcal{C}\}\in \mathcal{N}$.
  Recall that in \Cref{lheavy}, we computed a set $\mathcal{H}$ of $2^{O (k \log k)}$ tuples $(\bm{g}^{c_1}, \dots, \bm{g}^{c_\ell})$ with $\bm{g}^{c_i} \in \{\oldminusone, \oldzero, \oldone\}^{\cut (c_i)}$ where $c_1, \dots, c_\ell$ are the heavy children of $t$ such that if there exists an $\bm{h}$-embedding, then there also exists one coinciding with $\tau_{c_i} [ \bm{g}^{c_i}]$ for all $i \in [\ell]$.
  Similarly, in \Cref{cgc}, we computed a set~$\mathcal{M}$ of $2^{O(k)}$ matchings inside the good children such that if there exists an $\bm{h}$-embedding, then there also exists one coinciding with some $M \in \mathcal{M}$.
  For each class $\mathcal{C} = \{c_1, \dots, c_\ell\}$ with $N (\mathcal{C}) = \{x_1, x_2\}$, $\mathcal{N } (x_1) =  \mathcal{C} \neq \mathcal{N} (x_2)$ with $\oldone \notin \sig_{x_1} (\mathcal{C})$ and $\rk_{x_j} (y_j^{c_i}) < \rk_{x_j} (y_j^{c_{i+1}})$ for $i\in [\ell -1]$ and $j\in \{1,2\}$, we guess whether it is matched to $c_1$, $c^*$, or some other child.
  In the first two cases, we consider $\mathcal{C}$ as a good class, since we know which child of $\mathcal{C}$ shall be matched to $x_1$.
  Since there are at most $2^{O(k \log k)}$ matchings inside~$E(G[X_t])$, we can also enumerate all these matchings.
  Thus, we can guess which edges an $\bm{h}$-embedding contains inside the heavy children, the good children, and inside $E(G[X_t])$.
  Given such a guess $\tilde{M}$, it remains to decide whether we can extend $\tilde{M}$ on the bad and unmatched children to an $\bm{h}$-embedding.
  We will reduce the problem of deciding whether $\tilde{M}$ can be extended to a matching $M$ complying with $\bm{h}$ and being valid for $\mathcal{N}$ to a 2-SAT instance.

  First, we give a formal definition of a partial embedding and the problem we are going to solve.

  \begin{definition}
    A \emph{partial embedding} $\tilde{M}$ of $\mathcal{N}$ is a matching inside $E(\mathcal{C}_{\operatorname{good}})\cup E(\mathcal{C}_{\operatorname{heavy}}) \cup E(G[X_t])$ such that there is an edge $\{x, x^c\}$ from $x\in X_t$ to a vertex $x^c\in Y_c$ for some~$c$ contained in a class $\mathcal{C}\in \mathcal{C}_{\operatorname{good}}\cup \mathcal{C}_{\operatorname{heavy}}$ if and only if $\mathcal{N} (x) = \mathcal{C}$, where $\mathcal{C}_{\operatorname{good}}$ is the set of good children and $\mathcal{C}_{\operatorname{heavy}}$ is the set of heavy children of $t$.
  \end{definition}

  \defProblemTask{\textsc{Partial Embedding Extension}}
  {
  A graph $G$, preference lists $P$,
  a nice tree-cut decomposition $(T, \mathcal{X})$ of $G$, and a node $t \in V(T)$.
  Moreover, a matching $\mathcal{N}$ from $X_t$ to the classes,
  a partial embedding $\tilde{M}$ of $\mathcal{N}$,
  a vector $\bm{h}\in \{\oldminusone, \oldzero, \oldone\}^{\cut (t)}$, and
  the DP tables $\tau_c$ for all children $c$ of $t$.
  }
  {
  Decide whether there exists a matching $M \subseteq E(G_t)\setminus \bigl(E(\mathcal{C}_{\operatorname{good}}) \cup E(\mathcal{C}_{\operatorname{heavy}}) \cup E(G[X_t])\bigr)$ such that $\tilde{M} \cup M$ is an $\bm{h}$-embedding of $\mathcal{N}$.
  }

  \begin{lemma}\label{lrest}
    \textsc{Partial Embedding Extension} is solvable in polynomial time.
  \end{lemma}

  \begin{proof}
    We will reduce \textsc{Partial Embedding Extension} to 2-SAT. The correctness of the reduction will be proven in \Cref{l2sathin,l2satrueck}.
    As 2-SAT can be solved in linear time~\cite{ASPVALL1979}, the lemma follows.
  \end{proof}

  So let $\mathcal{I}$ be a \textsc{Partial Embedding Extension} instance.
  We assume that there is no blocking pair inside $E(\mathcal{C}_{\operatorname{good}}) \cup E( \mathcal{C}_{\operatorname{heavy}})$, since such a blocking pair is also blocking in every extension of $\tilde{M}$.
  We construct a 2-CNF formula~$\varphi$ similarly to the \textsc{2-SAT} formulation of \textsc{Stable Roommates} of Feder~\cite{Feder1992}.
  An example of our construction is given in \Cref{f2SAT}.
\subparagraph{The Construction}
  Due to \Cref{lem:once-nonzero,lem:once-zero,lem:same-child}, we may assume that for each bad class~$\mathcal{C}$ with $\N (\mathcal{C}) = \{x_1, x_2\} $, vertex $x_i$ prefers $y_i^c$ to $y_i^d$ or $x_i$ prefers $y_i^d$ to $y_i^c$ for each $c\neq d\in \mathcal{C}$, and that if $x_1$ prefers $y_1^c$ to~$y_1^d$, then also $x_2$ prefers $y_2^c $ to~$y_2^d$ if $\mathcal{N} (x_1) = \mathcal{C}\neq \mathcal{N} (x_2)$, and $x_2$ prefers~$y_2^d$ to~$y_2^c$ if $\mathcal{N} (x_1) = \mathcal{C} = \mathcal{N} (x_2)$.

  We assume that for each $x\in X_t$ with $\mathcal{N} (x) = \mathcal{C}$ for some class $\mathcal{C}$ and~$x \notin e$ for every edge~$e \in \tilde{M}$ we have that
  \begin{itemize}
    \item $\mathcal{C}$ is a bad class and
    \item $x$ ranks the vertices from $\mathcal{C}$ at position $1, 2, \dots, |\mathcal{C}|$.
  \end{itemize}
  To see this, observe the following.
  By \Cref{lem:once-nonzero,lem:once-zero,lem:same-child}, $x$ either prefers $y^c$ to~$y^d$ or prefers $y^d$ to $y^c$ for $c, d\in \mathcal{C}$.
  Furthermore, $x$ ranks all incident edges at a position from $\{1, 2, \dots, |\mathcal{C}|\}$:
  Let us write $\mathcal{C} = \{ c_1, \ldots, c_\ell \}$ with $\rk_x(c_1) < \cdots < \rk_x(c_\ell)$.
  Assume there exists a vertex $v \notin V(\mathcal{C})$ which is a neighbour of $x$ such that $\rk_x(c_i) < \rk_x(v) < \rk_x(c_{i+1})$ for some $i \in [\ell - 1]$.
  Recall that the edge $\{ x,v \}$ cannot be in $M$ (nor it is in $\tilde{M}$).
  Consequently, we may ``promote''~$v$ in its $x$-ranking to $\rk_x(v) = \rk_x(c_i)$, since this yields an equivalent instance (under the assumption that our matching is an $\bm{h}$-embedding of $\mathcal{N}$).

  For each $x\in X_t$ let us denote by $\max\rk_x$ the maximum $x$-rank of a neighbour of $x$, that is, $\max\rk_x := \max_{v\in \N (x)} \rk_x (v)$.
  We introduce variables $x^{j}$ for every $j$ with $1\le j\le \max\rk_x$, which will be true if and only if $x$ ranks its partner worse than $j$.
  Thus, we construct a subformula $\varphi_x$ as follows
  \[
    \varphi_x := \bigwedge_{j = 1}^{\max\rk_x -1} \left( x^{j} \lor \neg x^{j+1} \right) \,.
  \]
  Observe that in order to satisfy $\varphi_x$ we have that if $x^{j}$ is set to true, then $x^{j'}$ is set to true for all $j'\le j$ (and, by symmetry, if $x^j$ is set to false, then so is $x^{j'}$ for all $j' \ge j$).
  Note that given a satisfying truth assignment, we are going to match $x \in X_t$ to $y^c$ satisfying
  \begin{itemize}
    \item $c \in \mathcal{N}(x)$ and
    \item $\rk_x(y^c) = p$, where $p = \min\{ j \in \mathbb{N} : x_j = \false \}$.
    If no variable $x_j$ with $x_j=\false$ exists, then $x$ remains unmatched.
  \end{itemize}

  Now, we create a sentence $\varphi_{\operatorname{fixed}}$ (i.e., a conjunction of clauses) based on constraints from $\mathcal{N}$, $\bm{h}$, $\tilde{M}$ as follows.
  We start with $\varphi_{\operatorname{fixed}} = \emptyset$.
  \begin{itemize}
    \item If there is an edge $e \in \tilde{M}$ with $x \in e$, then we add the clauses $x^{p- 1}$ and $\neg x^{p}$ to $\varphi_{\operatorname{fixed}}$, where $p := \rk_x(\tilde{M}(x))$ (see clauses $x_1^1$ and $\neg x^2_1$ in \Cref{f2SAT} for an example).
    \item If $\mathcal{N}(x) = \bot$ (recall that this means that $x$ is unmatched by $\mathcal{N}$), then we add clauses of the form $x^{j}$ to $\varphi_{\operatorname{fixed}}$ for all $1\le j\le \max_{y\in \N_{G} (x)\cap Y_t} \rk_x (y)$, ensuring that $x$ is not matched in $M$.
    \item If $\mathcal{N}(x) \neq \bot$, then we add clause~$\neg x^p$ to $\varphi_{\operatorname{fixed}}$ for $p := \max_{y\in \N_{G} (x)\cap Y_t} \rk_x (y)$ to ensure that $x$ is matched in $M$ (see clauses $\neg x_2^2$, $\neg x^2_3$ and $\neg x^2_4$ in \Cref{f2SAT} for an example).
    \item For each edge $\{x, v\} \in \cut (t)$ with $\bm{h} (\{ x, v\}) = \oldone$, we add the clause $\neg x^{\rk_x (v)}$ to $\varphi_{\operatorname{fixed}}$.
    \item For each edge $\{x, x_*\}\in E(G[X_t])$, we add the clause $\neg x^{\rk_{x} (x_*)}\lor \neg x_*^{\rk_{x_*} (x)}$ to $\varphi_{\operatorname{fixed}}$ (see clause $\neg x_2^2 \lor \neg x^2_3$ in \Cref{f2SAT} for an example).
  \end{itemize}
  Observe that $\varphi_{\operatorname{fixed}}$ ensures that for every $x\in X_t$ with $\mathcal{N} ( x)$ not being a bad class, the variables $x^1, \dots, x^{\max\rk_x}$ correspond to $x$ being matched to $\tilde{M} (x)$.
  Furthermore, $\varphi_{\operatorname{fixed}}$ ensures that there is no blocking pair inside $E(G[X_t])$ or between a vertex $x\in X_t$ and a vertex from a good or heavy child, and that for every edge $\{x, v\}\in \cut (t)$ with $x\in X_t$ and $\bm{h} (\{x, v\})$ that $x$ does not prefer $v$ to its partner assigned by the \textsc{2-SAT} formula.

  \begin{figure}[]
    \begin{center}
    \begin{tikzpicture}
      \node[vertex, label=90:$x_1$] (x1) at (0, 0) {};
      \node[vertex, label=90:$x_2$] (x2) at (2, 0) {};
      \node[vertex, label=90:$x_3$] (x3) at (4, 0) {};
      \node[vertex, label=90:$x_4$] (x4) at (6, 0) {};
      
      \tikzstyle{child}=[draw, fill]

      \node[child, label=270:$a_1$] (a1) at (-3, -2) {};
      \node[child, label=270:$a_2$] (b1) at (-1, -2) {};
      \node[child, label=270:$b_1$] (c1) at (1, -2) {};
      \node[child, label=270:$b_2$] (c2) at (3, -2) {};
      \node[child, label=270:$c_1$] (d1) at (5, -2) {};
      \node[child, label=270:$c_2$] (d2) at (7, -2) {};

      \node[child, label= 90:$d_1$] (rd) at (1, 1) {};
      \node[child, label= 90:$e_1$] (e) at (3, 1) {};
      \node[child, label= 90:$f_1$] (f) at (5, 1) {};

      \draw[dotted] (rd) circle (0.6);
      \draw (x1) edge node[pos=0.3, fill=white, inner sep=1pt] {\scriptsize $1$} (rd);
      \draw (x2) edge node[pos=0.3, fill=white, inner sep=1pt] {\scriptsize $2$} (rd);

      \draw[dotted] (e) circle (0.6);
      \draw (x3) edge node[pos=0.3, fill=white, inner sep=1pt] {\scriptsize $1$} (e);
      \draw (x2) edge node[pos=0.3, fill=white, inner sep=1pt] {\scriptsize $2$} (e);

      \draw[dotted] (f) circle (0.6);
      \draw (x3) edge node[pos=0.3, fill=white, inner sep=1pt] {\scriptsize $1$} (f);
      \draw (x4) edge node[pos=0.3, fill=white, inner sep=1pt] {\scriptsize $2$} (f);

      \draw[dotted] (-2., -2) ellipse (1.7 and 1);
      \draw[dotted] (2., -2) ellipse (1.7 and 1);
      \draw[dotted] (6., -2) ellipse (1.7 and 1);

      \draw[dashed] (x1) -- (-2, -2);
      \draw[dashed] (x2) -- (2, -2);
      \draw[dashed] (x3) -- (2, -2);
      \draw[dashed] (x4) -- (6, -2);

      \draw (x1) edge node[pos=0.2, fill=white, inner sep=1pt] {\scriptsize $1$} (a1);
      \draw (x2) edge node[pos=0.2, fill=white, inner sep=1pt] {\scriptsize $2$} (a1);
      \draw (x1) edge node[pos=0.2, fill=white, inner sep=1pt] {\scriptsize $2$} (b1);
      \draw (x2) edge node[pos=0.2, fill=white, inner sep=1pt] {\scriptsize $2$} (b1);
      \draw (x2) edge node[pos=0.2, fill=white, inner sep=1pt] {\scriptsize $2$} (c1);
      \draw (x3) edge node[pos=0.2, fill=white, inner sep=1pt] {\scriptsize $1$} (c1);
      \draw (x2) edge node[pos=0.2, fill=white, inner sep=1pt] {\scriptsize $1$} (c2);
      \draw (x3) edge node[pos=0.2, fill=white, inner sep=1pt] {\scriptsize $2$} (c2);
      \draw (x3) edge node[pos=0.2, fill=white, inner sep=1pt] {\scriptsize $1$} (d1);
      \draw (x4) edge node[pos=0.2, fill=white, inner sep=1pt] {\scriptsize $1$} (d1);
      \draw (x3) edge node[pos=0.2, fill=white, inner sep=1pt] {\scriptsize $2$} (d2);
      \draw (x4) edge node[pos=0.2, fill=white, inner sep=1pt] {\scriptsize $2$} (d2);
      \draw (x2) edge node[pos=0.2, fill=white, inner sep=1pt] {\scriptsize $2$} node[pos=0.76, fill=white, inner sep=1pt] {\scriptsize $2$} (x3);
    \end{tikzpicture}
    \end{center}
    \begin{center}
      \begin{tabular}{l | l | l}
      reason for clause & clause(s) & subformula\\
      \hline
      $\varphi_{x_i}$ & $x_i^{j}\lor \neg {x_i^{j+1}}$ & $\varphi_{x_i}$\\
      $\{a_2, x_1\}\in \tilde{M}$ & $x^1_1$, $\neg {x_1^{2}}$ & $\varphi_{\operatorname{fixed}}$\\
      $\mathcal{N} (x_2) \neq \bot$ & $\neg x_2^2$ & $\varphi_{\operatorname{fixed}}$ \\
      $\mathcal{N} (x_3) \neq \bot$ & $\neg x^2_3$ & $\varphi_{\operatorname{fixed}}$ \\
      $\mathcal{N} (x_4) \neq \bot$ & $\neg x^2_4$ & $\varphi_{\operatorname{fixed}}$ \\
      $\oldzero\notin \sig_{x_2} (a_1)$ and $\{a_2, x_1\}\in \tilde{M}$ & $\neg x^2_{2}$ &$\varphi_{\operatorname{good}}$\\
      $b_1$ & ${x_2^{1}}\lor {x_3^{1}}$ & $\varphi_{\operatorname{bad}}$\\
      $c_1$ & $\neg x_3^{1} \lor {x_4^{1}}$ & $\varphi_{\operatorname{bad}}$\\
      $c_2 $ & $\neg x_3^{2} \lor {x_4^{2}}$ & $\varphi_{\operatorname{bad}}$\\
      $\{x_2, x_3\}$ & $ \neg x^2_{2} \lor \neg x_3^{2}$ & $\varphi_{\operatorname{fixed}}$\\
      $d_1$ & $\neg x_2^2$ & $\varphi_{\operatorname{unmatched}}$\\
      $e_1$ & $\neg x_2^2\lor \neg x^1_3$ & $\varphi_{\operatorname{unmatched}}$\\
      $f_1$ & $\neg x^1_3$, $\neg x^2_4$ & $\varphi_{\operatorname{unmatched}}$
      \end{tabular}

    \end{center}

    \caption{An example for the reduction to 2-SAT.
    The picture depicts the vertices from~$X_t=\{x_1, x_2, x_3, x_4\}$ as well as the children of $t$ (depicted as squares).
    Children in the same class are surrounded by a dotted ellipse.
    The matching $\mathcal{N}$ is given by the dashed edges.
    Children~$a_1$ and $a_2$ belong to a good class, and we have $\tau^{-1}_{a_i} (\square)= \{\bm{(\oldone, \oldone)}, \bm{(\oldzero, \oldone)}\}$.
    We have $\{a_2, x_1\}\in \tilde{M}$.
    Children $b_1$ and $b_2$ belong to a bad class, and we have $\tau^{-1}_{b_i} (\square) = \{\bm{(\oldzero, \oldminusone)}, \bm{(\oldzero, \oldone)}\}$.
    Children $c_1$ and~$c_2$ belong to a bad class, and we have $\tau^{-1}_{c_i} (\square) = \{\bm{(\oldzero, \oldone)}, \bm{(\oldzero, \oldminusone)}, \bm{(\oldone, \oldone)}\}$.
    Children $d_1$, $e_1$, and~$f_1$ belong to unmatched classes, and we have $\tau^{-1}_{d_1} ( \square) = \{ \bm{(\oldminusone, \oldone)}, \bm{(\oldone, \oldone)}\}$ as well as $\tau^{-1}_{e_1} ( \square) = \{\bm{(\oldone, \oldone)}\}$ and $\tau^{-1}_{f_1} ( \square) = \{ \bm{(\oldminusone, \oldone)}, \bm{(\oldone, \oldminusone)}, \bm{(\oldone, \oldone)}\}$.
    The variables of the 2-SAT instance are $x^1_{1}, x_1^{2}, x_2^{1}, x^2_{2}, x_3^{1}, x_3^{2}, x_4^{1}$, and~$x_4^{2}$.
    The clauses constructed because of an edge respectively vertex are given in the table.}
    \label{f2SAT}
  \end{figure}

  Next, we consider the heavy and the good children belonging to a class $\mathcal{C}$.
  The matching inside $E(\mathcal{C})$ is already fixed by the partial embedding $\tilde{M}$.
  For each child $c\in \mathcal{C}$ and each~$x\in N(\mathcal{C})$ with $\rk_{x^c} (x) < \rk_{x^c} (\tilde{M} (x^c))$, we add the clause $\varphi_c = \neg x^r$ with $r = \rk_x (x_c)$.
  See clause~$\neg x_2^2$ in \Cref{f2SAT} for an example.
  We call the conjunction of the formulas over all heavy and all good children $\varphi_\text{good}$.

  Next, we consider the unmatched children and for each such child $c$ with $N(Y_c) = \{x_1\}$ or $N (Y_c) = \{x_1, x_2\}$, we create its formula $\varphi_c$.
  We distinguish few cases based on the first case applicable to $\sig(c)$:
  \begin{enumerate}
    \item
    If $\bm{(\oldone,\oldone)} \in \sig(c)$ or $\bm{(\oldone)} \in \sig (c)$, then $c$ does not yield any further conditions on the rank of the partner of its neighbors.
    Thus, we set $\varphi_c$ to an empty formula (e.g., to true).
    \item
    If $\bm{(\oldone,\oldminusone)} \in \sig(c)$ and $\bm{(\oldminusone,\oldone)} \notin \sig(c)$, then we set $\varphi_c = \neg {x}_2^{s}$, where $s = \rk_{x_2}(y_2^c)$ (see clause $\neg x_2^2$ for $d_1$ in \Cref{f2SAT} for an example).
    Note that this imposes the constraint that $x_2$ must be matched to an agent it finds at least as good as $y_2^c$.
    The symmetric case $\bm{(\oldminusone, \oldone)}\in \sig (c)$ and $\bm{(\oldone, \oldminusone)} \notin \sig (c)$ is handled analogously.
    \item
    If both $\bm{(\oldone,\oldminusone)},\bm{(\oldminusone,\oldone)} \in \sig(c)$, then we set $\varphi_c = \neg {x}_1^r \lor \neg {x}_2^{s}$, where $r = \rk_{x_1}(y_1^c)$ and $s = \rk_{x_2}(y_2^c)$ (see clause $\neg x_2^2 \lor \neg x^1_3$ in \Cref{f2SAT} for an example).
    Similarly to the previous case this demands $x_1$ or $x_2$ to be matched to a partner they find at least as good as their possible partners in $Y_c$.
    \item
    If $\bm{(\oldminusone,\oldminusone)} \in \sig(c)$, then we set $\varphi_c = \neg {x}_1^r \land \neg {x}_2^{s}$, where $r = \rk_{x_1} (y_1^c)$ and $s = \rk_{x_2} (y_2^c)$ (see clauses $\neg x^1_3$ and $\neg x^2_4$ for $f_1$ in \Cref{f2SAT} for an example).
    Note that this setting requires both $x_1$ and $x_2$ to be matched to a partner they find at least as good as their possible partners in $Y_c$.
    \item
    If $\bm{(\oldminusone)} \in \sig (c)$, then we set $\varphi_c = \neg x^r$, where $r = \rk_{x_1} (y^c) $.
    Note that this setting requires~$x$ to be matched to a partner it finds at least as good as $y^c$.
  \end{enumerate}
  Let $\varphi_{\operatorname{unmatched}}$ be the conjunction of the formulas $\varphi_c$ for all unmatched children.

  So far we have constructed formulas that ensure that the sought matching obeys the obvious conditions given by the choices for unmatched, good, and heavy children and by the vector $\bm{h}$.
  Thus, we are left with the bad children.
  Recall that if $\mathcal{C}$ is a bad class with~$\mathcal{N}(x) = \mathcal{C}$, then we have altered the ranking function of $x$ such that we have $\rk_x(v) \in \{1, \ldots, |\mathcal{C}| \}$ for all~$v \in N(x)$.
  \begin{enumerate}
    \item
    Let $\mathcal{C}$ be a bad class with $\mathcal{N}(x_1) = \mathcal{C}$ and $\mathcal{N}(x_2) = \mathcal{C}$ for some $x_1 \neq x_2$.
    Recall that in this case we have to match both $x_1$ and $x_2$ to the same child in $\mathcal{C}$ and that if $\rk_{x_1} (y_1^c) = p$, then $\rk_{x_2}(y_2^c) = |\mathcal{C}| - p + 1$ (by \cref{lem:same-child} and since $x\in \{x_1, x_2\}$ ranks the vertices from~$\mathcal{C}$ at positions~$1,2,\dots, |\mathcal{C}|$).
    Thus we set $\varphi_{\mathcal{C}} = \bigwedge_{p = 1}^{|\mathcal{C}| - 1} \left( {x_1}^{p} \lor {x_2}^{|\mathcal{C}| - p } \right)$  (see clause $x^1_2\lor x^1_3$ in \Cref{f2SAT} for an example).
    \item
    Let $\mathcal{C}$ be a bad class with $\N (\mathcal{C}) = \{x_1, x_2\}$, $\mathcal{N}(x_1) = \mathcal{C}$, and $\mathcal{N}(x_2) \neq \mathcal{C}$ for some $x_1\neq x_2$.
    For every $c\in \mathcal{C}$ with $p = \rk_{x_1} (y_1^c)$ and $q = \rk_{x_2} (y_2^c)$, we add the following clauses, depending on $\sig_{x_1} ( \mathcal C)$:
    \begin{itemize}
      \item If $\sig_{x_1} (\mathcal{C}) \in \{\emptyset, \{\oldminusone\}, \{ \oldzero\}, \{\oldminusone, \oldzero\}\}$, then we add clauses~$x_1^1$ and $\neg x_2^{\rk_{x_2} (y_2^{c_1})}$ (as we guessed that $x_1$ is not matched to $c_1$).
      \item If $\sig_{x_1} (\mathcal{C}) = \{\oldminusone, \oldone\}$, then we add clause~$x_1^p \lor \neg x_2^q$.
      \item If $\sig_{x_1} ( \mathcal{C}) = \{\oldminusone, \oldzero, \oldone\}$, then we add no clause.    
    \end{itemize}

  \end{enumerate}

  We denote by $\varphi_{\operatorname{bad}}$ the conjunction of the formula $\varphi_{\mathcal{C}}$ for every bad class~$\mathcal{C}$.

  Let $\varphi$ be the conjunction of all of the above constructed formulas, i.e., $\varphi = \bigl( \bigwedge_{x\in X_t} \varphi_x \bigr) \wedge \varphi_{\operatorname{fixed}} \wedge \varphi_{\operatorname{good}} \wedge \varphi_{\operatorname{unmatched}} \land \varphi_{\operatorname{bad}}$.

\begin{lemma}\label{l2sathin}
  If there exists an $\bm{h}$-embedding of $\mathcal{N}$ of the form $M \cup \tilde{M}$, then there exists a satisfying truth assignment for $\varphi$.
\end{lemma}

\begin{proof}
  We set $x^{j}$ to true if and only if $\rk_x((M\cup \tilde{M})(x)) > j$ and we claim that this is a satisfying truth assignment.
  It is not hard to see that, since $M\cup\tilde{M}$ complies with $\bm{h}$, such an assignment satisfies formulas $\varphi_{\operatorname{good}},\varphi_{\operatorname{unmatched}}, \varphi_{\operatorname{fixed}}$, and $\varphi_x$ for all $x \in X_t$.

  Now, consider a bad class $\mathcal{C}$ with $\mathcal{N} (x_1) = \mathcal{C} = \mathcal{N} (x_2)$ for some distinct $x_1, x_2 \in X_t$.
  Recall that we have that $M(x_1) = y_1^c$ and $M(x_2) = y_2^c$ for some $c \in \mathcal{C}$;
  consequently, we have that $x_1^{p -1}$ and $x_2^{|C| - p}$ are set to true for $p = \rk_{x_1}(y_1^c)$.
  But now, if~$x_1^{p-1}$ is true, then $x_1^{p'}$ is true for all $p' \le p-1$.
  By the same argument, $x_2^{p'}$ is true for all $p' \le |\mathcal{C}| - p$, and the formula $\varphi_{\mathcal{C}}$ is satisfied by the constructed assignment.

  Finally, let $\mathcal{C} =\{c_1, \dots, c_\ell\}$ be a bad class with $\N (\mathcal{C}) = \{x_1, x_2\}$, $\mathcal{N} (x_1) = \mathcal{C}\neq \mathcal{N} (x_2)$, and $\rk_{x_j} (y_j^i) < \rk_{x_j} (y_j^{i+1})$ for $i \in [\ell -1]$ and $j\in \{1, 2\}$.
  Let $M(x_1) = y_1^c$ and $q_c := \rk_{x_1} (y_1^c)$ for some~$c \in \mathcal{C}$.
  Fix some $d\in \mathcal{C}$ and let $p:= \rk_{x_1} (y_1^d)$ and $q:= \rk_{x_2}(y_2^d)$.
  If $\sig (\mathcal{C} = \{\emptyset, \{\oldminusone\}, \{ \oldzero\}, \{\oldminusone, \oldzero\}\}$, then we guessed that $x_1$ is not matched to $y_1^{c_1}$.
  Consequently, $x_1^1$ is set to true.
  Since $M \cup \tilde{M}$ complies with $\bm{h}$, it follows that $x_2$ prefers $(M\cup \tilde{M}) (x_2)$ to~$y_2^{c_1}$, and consequently, $x_2^{\rk_{x_2}(y_2^{c_1})} $ is set to false.
  Thus, clause~$\neg x_2^{\rk_{x_2} (y_2^{c_1})}$ is satisfied.
  If $\sig_{x_1} (\mathcal{C}) = \{ \oldminusone, \oldone\}$, then $x_1^{p_c-1}$ is set to true, $x_1^{p_c}$ is set to false, and $x_2^p$ is set to true for~$p = \rk_{x_1} (x_1^c)$.
  If $p < p_c$, then $x_1^{p_c-1}$ is true, since $x_1^{p'}$ is true for all~$p' \le p_c -1$.
  If $p\ge p_c$, then $q \ge \rk_{x_2} (y_1^c)$.
  Since $M\cup \tilde{M}$ is an $\bm{h}$-embedding, it follows that $x_2$ does not prefer $x_2^{q_c}$ to $( M \cup \tilde{M}) (x_2)$.
  Thus, $x_2^{p_c}$ is false.
  Since this implies that $x_2^{p'}$ is false for all~$p' \ge p_c$, we have that $x_2^q$ is false.
  Thus, clause~$x_2^p \lor \neg x_2^q$ is satisfied.

  We conclude that our truth assignment satisfies $\varphi$ and the lemma follows.
\end{proof}

We now show the reverse direction, namely that every satisfying truth assignment implies a solution to the \textsc{Partial Embedding Extension} instance.

\begin{lemma}\label{l2satrueck}
  If there exists a satisfying truth assignment for $\varphi$, then there exists an $\bm{h}$-embedding of $\mathcal{N}$ of the form $M \cup \tilde{M}$.
\end{lemma}

\begin{proof}
  Assume that there is a satisfying truth assignment $f\colon\{x^{j}: x\in X_t\}\rightarrow \{ \true, \false\}$.
  Since $f$ satisfies $\varphi_x$, we know that for each $x\in X_t$ there exists some $j_x$ such that $x^{j}$ is set to false if and only if $j\ge j_x$.
  We first compute the matching $M$ and then argue that $M \cup \tilde{M}$ is an $\bm{h}$-embedding of $\mathcal{N}$.

  We start with a claim for classes~$\mathcal{C} $ of bad children with $\mathcal{N} (x_1) = \mathcal{C} = \mathcal{N} (x_2)$ for some $x_1, x_2 \in X_t$.
  \begin{claim}
    Let $\mathcal{C}$ be a class of bad children with $\mathcal{N}(x_1) = \mathcal{N}(x_2) = \mathcal{C}$ and let $f$ be an assignment satisfying $\varphi$.
    Then there exists an assignment $f' \colon \{x^{j}: x\in X_t\}\rightarrow \{ \true, \false\}$ satisfying $\varphi$ for which there exists an integer $q$ such that both $x_1^{q - 1}$ and $x_2^{|C| - q}$ are set to true in $f'$, and $f'$ coincides with $f $ on $\{x^j : x\in X_t \setminus \{x_1, x_2\}\}$.
  \end{claim}
  \begin{claimproof}
    Assume that in $f$ no such $q$ exists for $\mathcal{C}$ and recall that $\varphi_{\mathcal{C}} = \bigwedge_{p = 2}^{|\mathcal{C}| - 1} \left( {x}_1^p \lor {x}_2^{|\mathcal{C}| - p} \right)$.
    Let $q' := \min \left(\{j\in \mathbb{N} : x_1^j = \false\}\cup\{\max\rk(x_1)\}\right)$.
    Let $f'$ be the same assignment as $f$ with the only exception that we set $x_2^{p}$ to false for $p \ge |\mathcal{C}| - q' + 1$.
    Clearly, $f'$ satisfies $\varphi$; indeed we have only selected a more permissive assignment for $x_2$ (any clause outside $\mathcal{C}$ does not contain the literal~$x_2^{p}$ for any $p$).
  \end{claimproof}

  It follows that we can select an initial matching $M'$ between vertices in $X_t$ and their possible partners in bad children as follows.
  For a bad class $\mathcal{C}$ and an $x\in X_t$ with~$\mathcal{N} (x) = \mathcal{C}$, we add the edge $\{x, y^c\}$ to $M'$ for the $c\in \mathcal{C}$ such that $\rk_x(y^c) = p$, where $p : = \min \{ j \in \mathbb{N} : x_j = \false \}$.
  Now, we extend $M'$ to a matching $M$ on $\clos(t) $ by adding for each bad or unmatched child $c$ the matching stored in $\tau_c [\bm{h^{M', c}}]$.
  We claim that all such entries exist and that the resulting matching $M \cup\tilde{M}$ is an $\bm{h}$-embedding of $\mathcal{N}$.
  First, we show that $\tau_c [ \bm{h^{M', c}}] \neq \square$.
  Let $c \in \mathcal{C}$ with $N (\mathcal{C}) = \{ x_1, x_2\}$.
  If $c$ is a bad children with $\mathcal{N} ( x_1) = \mathcal{C} = \mathcal{N} (x_2)$, then $\bm{h^{M', c}} \in \{\bm{(\oldminusone, \oldone)}, \bm{(\oldone, \oldminusone)}, \bm{(\oldzero, \oldzero)}\}\subseteq \sig (c)$.
  
  If $c$ is a bad children with $\mathcal{N} (x_1 ) = \mathcal{C} \neq \mathcal{N} (x_2)$, then make a case distinction on the different values $\sig_{x_1} (\mathcal{C})$ can take.
  Let $C = \{c_1, \dots, c_\ell\}$ with $\rk_{x_1} (x_1^i) < \rk_{x_1} (x_1^{i+1})$ and $\rk_{x_2} (x_2^i) < \rk_{x_2} (x_2^{i+1})$.
  Let $c\in \mathcal{C}$ with $p := \rk_{x_1} (x_1^c)$ and $q:= \rk_{x_2} (x_2^c)$.
  If $\sig_{x_1} (\mathcal{C}) \in \{ \emptyset, \{\oldminusone\}, \{ \oldzero\},\{ \oldminusone, \oldzero\} \}$, then $M' (x_1) \neq y_1^{c_1}$ (otherwise the class $\mathcal{C}$ is good).
  Thus, clause $x_1^1$ is satisfied.
  Furthermore, we have $\bm{h^{M', c_1}} (\{x_1, y_1^c\})= \oldone$.
  Thus, clause $x_2^{\rk_{x_2} (y_2^{c_1})}$ is satisfied.
  If $\sig_{x_1} (\mathcal{C}) = \{\oldminusone, \oldone\}$, then
  $\{\oldminusone, \oldone\}^2 \subseteq \sig (c)$, and thus, for all $c \in \mathcal{C}$ such that $\{y_1^c, x_1\} \notin M'$, we have that $\tau_c [ \bm{h^{M', c}}] \neq \square$.
  If $\{y_1^c, x_1 \} \in M'$, then clause $x_1^{\rk_{x_1} (y_1^c)} \lor \neq x_2^{q}$ with $q = \rk_{x_2} (y_2^{c})$ implies that $x_2$ does not prefer $y_2^c$ to $M' (x_2)$.
  Consequently, we have $\bm{h^{M', c}} = \bm{(\oldzero, \oldminusone)} \in \sig (c)$, and it follows that $\tau_c [ \bm{h^{M', c}}] \neq \square$.
  Thus, we have $\tau_c [ \bm{h^{M', c}}] \neq \emptyset$.

  Observe that $\varphi_{\operatorname{fixed}}$ assures that for each $\{ x, v\}\in \cut (t)$ with $\bm{h} (\{ x, v\}) = \oldone$ we have that $\rk_x (M(x)) \le \rk_x (v)$.
  Furthermore, it assures that no edge $\{x, x'\}\in E(G[X_t])$ is blocking.
  By $\varphi_{\operatorname{good}}$, no edge $\{ x, y^c \}$ for a heavy or a good child $c$ of $t$ is blocking.
  By $\varphi_{\operatorname{unmatched}}$ no edge $\{ x,y^c \}$ for an unmatched child $c$ of $t$ is blocking.
  The lemma follows.
\end{proof}

\begin{theorem}\label{cind}
  Let $t\in V(T)$ and $\bm{h}\in \{\oldminusone, \oldzero, \oldone\}^{\cut (t)}$. Let $t_1, \dots, t_j$ be the children of $t$.

  If we know $\tau[t_i, h_i]$ for all $i\in [j]$ and $h_i\in \{\oldminusone, \oldzero, \oldone\}^{\cut (t_i) }$, then we can compute $\tau[t, h]$ in $2^{O(k\log k)}n^{O(1)}$ time.
\end{theorem}

\begin{proof}
  By \Cref{lheavy}, we have to consider $2^{O(k\log k)}$ matchings from $X_t$ to the heavy children of $t$.

  As there are $O(k^2)$ classes of light children, there are $2^{O(k\log k)}$ matchings between $X_t$ and the classes.

  By \Cref{cgc}, we get $2^{O(k)}$ matchings inside the light children.

  Since there are $2^{O(k \log k)}$ matchings inside $G[X_t]$, this results in $2^{O(k \log k)}$ partial embeddings for each of these matchings.

  Each of these partial embeddings defines a \textsc{Partial Embedding Extension} instance, which can be solved in polynomial time by \Cref{lrest}.
  Clearly, a matching complying with~$\bm{h}$ exists if and only if one of the \textsc{Partial Embedding Extension} instances is a YES-instance.

  Therefore, the total running time is $2^{O(k\log k)} n^{O(1)}$.
 \end{proof}

Having proven the correctness of the induction step, we now state the main theorem of this section.

\begin{theorem}\label{tfpttcw}
 Both \textsc{Perfect-SRTI} and \textsc{SRTI-Existence} can be solved in $2^{O(k\log k)}n^{O(1)}$ time, where $k:=\tcw (G)$.
\end{theorem}

\begin{proof}
  Let $(T, \mathcal{X})$ be a nice tree-cut decomposition of $G$ of width $k$.
  We will first explain the algorithm for \textsc{SRTI-Existence}, and in the end we highlight how this algorithm can be adapted to \textsc{Perfect-SRTI}.
 We compute the values $\tau_t [\bm{h}]$ by bottom-up induction over the tree $T$.

 For a leaf $t\in V(T)$ and a vector $\bm{h}$, we enumerate all matchings $M_t$ on $G[X_t\cup \N (Y_t)]$. We check whether $M_t$ complies with $\bm{h}$. If we find such a matching, then we store one of these matchings in $\tau_t [\bm{h}]$, and else set $\tau_t[\bm{h}] =\square$. As $|X_t\cup \N (Y_t)|\le 2k$, and $G$ is simple, the number of matchings is bounded by $2^{O(k\log k)}$.

 The induction step, that is, computing the table entries for the inner nodes of the tree-cut decomposition is the most-involved part and sketched below.

 For the root $r\in V(T)$, we have $Y_r = V(G)$ and $\cut (r) = \emptyset$. Thus, a matching on $Y_r = V(G)$ complying with $\bm{h^r}\in\{\oldminusone, \oldzero, \oldone\}^\emptyset$ is just a stable matching (note that $\bm{h^r}$ is unique). Hence, $G$ contains a stable matching if and only if $\tau_r [\bm{h^r}] \neq \square$.

 The induction step is executed for each $t\in V(T)$ and each $\bm{h}\in\{\oldminusone,\oldzero, \oldone\}^{\cut (t)}$, and therefore at most $n 3^k$ times. As each execution takes $2^{O(k\log k)}n^{O(1)}$ time (\Cref{cind}), the total running time of the algorithm is bounded by $2^{O(k\log k)}n^{O(1)}$.

 To solve \textsc{Perfect-SRTI}, we store in any dynamic programming table $\tau_t$ only matchings such that every vertex inside $Y_t$ is matched.
\end{proof}

We now show that the above algorithm directly yields a $\frac{1}{2}$-FPT-approximation.

 \begin{corollary}\label{capx}
   \textsc{Max-SRTI} allows for a factor-$\frac{1}{2}$-approximation computable in $2^{O(k\log k)}n^{O(1)}$ time, where $k:= \tcw (G)$.
 \end{corollary}

 \begin{proof}
   Any stable matching is a maximal matching, and thus any stable matching is a $\frac{1}{2}$-approximation of a maximum stable matching.

   Therefore, any algorithm finding an arbitrary stable matching, if one exists, is a $\frac{1}{2}$-approximation algorithm, and such an algorithm with the claimed running time exists by \Cref{tfpttcw}.
 \end{proof}

  We now extend the algorithm of \Cref{tfpttcw} to also work for \textsc{Max-SMTI}.

 \subsection{Max-SMTI and Tree-cut Width}
 \label{sec:Max-SMTI}

 When trying to generalize the algorithm from \Cref{tfpttcw} to solve \textsc{Max-SRTI}, a straightforward idea is to adopt the DP tables $\tau_c [ \bm{h}]$ to store a matching complying with $\bm{h}$ of maximum cardinality instead of an arbitrary one.
 However, the induction step is unable to compute maximum-cardinality matchings complying with $\bm{h}$.
 For example, even if we know that a given vertex $x\in X_t$ shall be matched to a class of singleton children, it is not clear how to reduce the possibilities $x$ can be matched to to a finite number, as it might be beneficial to match $x$ to a vertex it likes less because this results in a larger matching.
 A natural approach to cope with this issue is to include the sizes of the matchings in the definition of the classes;
 however, then the number of classes is not bounded by a function of the tree-cut width any more.
 In \Cref{sec:max-SMTI-tcw}, we will show that on bipartite acceptability graphs (that is, for \textsc{Max-SMTI}), we can anyway bound the number of classes in a function of the tree-cut width.
 In \Cref{sec:max-SMTI-alg}, we then show how to adapt the algorithm from \Cref{tfpttcw} to solve \textsc{Max-SMTI}.
 
 \subsubsection{Max-SMTI and Preference Changes}
\label{sec:max-SMTI-tcw}

We start with a different characterization of weakly stable matchings due to Manlove et al.~\cite{ManloveIIMM02}, which reduces weak stability to stability in instances without ties.
In order to do so, we say that a \textsc{Max-SMTI} instance~$\mathcal{I}$ arises from a \textsc{Max-SMTI} instance $\mathcal{I}'$ by \emph{breaking ties} if $\mathcal{I}$ and $\mathcal{I}'$ have the same set of vertices, and the preferences in $\mathcal{I}'$ of every vertex~$v$ arise from the preferences of $v$ in $\mathcal{I}$ by replacing every tie containing vertices $\{v_1, v_2, \dots, v_\ell\}$ by some strict order of $v_1, \dots, v_\ell$.

\begin{observation}[\cite{ManloveIIMM02}]
\label{lbreakingtiesManlove}
  Let $\mathcal{I}$ be a \textsc{Max-SMTI} instance.
  Then, a matching is stable if and only if it is stable for some instance~$\mathcal{I}'$ arising from $\mathcal{I}$ by breaking the ties.
\end{observation}
For instances without ties, Manlove et al.~\cite{ManloveIIMM99} showed that changing the preferences of a single vertex changes the size of a stable matching by at most one.

\begin{lemma}[{\cite[Lemma 2.6]{ManloveIIMM99}}]
\label{lnoties}
  Let $\mathcal{I}$ be an \textsc{SMI} instance. Changing the preference lists of one vertex changes the size of a stable matching by at most one.
\end{lemma}

Combining \Cref{lbreakingtiesManlove,lnoties}, we generalize \Cref{lnoties} to the case of ties.

\begin{lemma}\label{lwithties}
  Let $\mathcal{I}$ be a \textsc{Max-SMTI} instance.
  Changing the preference lists of one vertex~$v$ changes the size of a maximum stable matching by at most one.
\end{lemma}

\begin{proof}
  Let $\mathcal{J}$ be the modified instance, and let $P'_v$ be preferences obtained by breaking the ties in the preferences of $v$ in instance $\mathcal{J}$.
  Let $M$ be a maximum stable matching in~$\mathcal{I}$.
  By \Cref{lbreakingtiesManlove}, there is some way of breaking the ties such that $M$ is stable in the arising instance $\mathcal{I}'$.
  We create an \textsc{SMI} instance $\mathcal{J}'$ from $\mathcal{I}'$ by replacing the preferences of~$v$ by $P'_v$.
  By \Cref{lnoties}, there exists a stable matching~$M'$ in $\mathcal{I'}$ with $||M|-|M'| | \le 1$.
  By \Cref{lbreakingtiesManlove}, this matching is stable for the instance $\mathcal{J}$, showing that there exists a stable matching in $\mathcal{J}$ whose size is by at most one smaller than the size of a maximum stable matching in $\mathcal{I}$.

  As $\mathcal{I}$ can be obtained from $\mathcal{J}$ by changing the preference list of $v$, the lemma follows.
\end{proof}

We immediately obtain the following corollary:

\begin{corollary}\label{cmodification}
  Let $\mathcal{I}$ be a \textsc{Max-SMTI} instance. Changing the preference lists of $k$ vertices changes the size of a maximum stable matching by at most $k$.\hfill\qedsymbol
\end{corollary}

Note that \Cref{lwithties,cmodification} do not generalize to \textsc{Max-SRTI} as the instance~$\mathcal{J'}$ is not guaranteed to admit a stable matching.

\Cref{cmodification} allows us to ``prune'' matchings complying with some vector $\bm{h} \in \{\oldminusone, \oldzero, \oldone\}^{\cut (t)}$ for some node $t$ if they are much smaller than a maximum-cardinality matching complying with some $\bm{h'} \in \{\oldminusone, \oldzero, \oldone\}^{\cut (t)}$.

\begin{lemma}\label{lprune2}
  Let $\mathcal{I}$ be a \textsc{Max-SMTI} instance with acceptability graph $G$.
  Let $\emptyset \neq X\subset V(G)$ be a set of vertices with $k := |N(X)|$ neighbors.
  Assume that there exist two  matchings $M$ and~$M'$ on $X\cup N(X)$ such that there is no blocking pair for either of the two matchings with both endpoints in $X$.

  If $|M| <|M'| + 3k$, then no maximum stable matching contains $M$.
\end{lemma}

\begin{proof}
  Let $S:= N(X)$ be the set of vertices in the neighborhood of $X$.

  Assume that there is a maximum stable matching $M^*$ containing $M$. We now modify the preference lists of $S':= S\cup \{M'(s): s\in S\}\cup \{M^* (s): s\in S\}$ as follows:
  for each~$s\in S'$ with~$M' (s) \neq s$, we delete vertices on the preferences of $s$ except for $M'(s)$.
  If $s=M'(s)$ (i.e., $s$ is unmatched in $M'$), then we delete all vertices from the preferences of~$s$.
  This results in an instance $\mathcal{J}$.

  \begin{claim}
    $N^*:= (M^*\setminus M) \cup M'$ is a stable matching for instance $\mathcal{J}$.
  \end{claim}
  \begin{claimproof}
    Assume that there exists a blocking pair $\{v, w\}$.
    Clearly, none of the two vertices is contained in $S'$, as each vertex of $S'$ is matched to its top choice or has no neighbor.

    For every vertex $u \in V(G) \setminus (X \cup S')$, we have $M^* (u) = N^* (u)$, and also the preference lists of~$u$ in $\mathcal{I}$ and $\mathcal{J}$ coincide (up to the deletion of edges not contained in $N^*$).
    As $M^*$ was stable, $N^*$ also does not contain a blocking pair inside $G - (X \cup S')$.
 
    For each $u\in X \setminus S'$, we have $N^* (u) = M' (x)$, and also the preference list of $u$ coincides in $\mathcal{I}$ and $\mathcal{J}$ (again up to the deletion of edges not contained in $N^*$).
    As $M'$ is stable, $N^*$ also does not contain a blocking pair inside~$X$.
    Thus, thus $N^*$ is stable.
  \end{claimproof}

  Note that $|N^*| = |M^*| -|M| + |M'| > |M^* | + 3k$. Therefore, we have $\OPT (\mathcal{I}) < \OPT (\mathcal{J}) + 3k$.
  However, the instance $\mathcal{J}$ arose from $\mathcal{I}$ by the changing of at most $3k$ preference lists (as clearly $|S'| \le 3|S| = 3k$).
  Thus, by \Cref{cmodification}, we have $|\OPT (\mathcal{I}) - \OPT (\mathcal{J}) |\le 3k$, a contradiction to $\OPT (\mathcal{I}) < \OPT (\mathcal{J}) + 3k$.
\end{proof}

\subsubsection{Max-SMTI Parameterized by Tree-Cut Width is in FPT}
\label{sec:max-SMTI-alg}

We now use \Cref{lprune2} to extend the FPT-algorithm for \textsc{SRTI-Existence} to an FPT-algorithm for \textsc{Max-SMTI}.
This shows that \textsc{Max-SMTI} parameterized by tree-cut width is computationally easier than \textsc{Max-SRTI} unless $FPT = W[1]$.

The adaption of the algorithm from \Cref{tfpttcw} works as follows.
The dynamic programming table $\tau$ shall now contain a maximum-cardinality locally stable matching, and contain~$\square$ if no such matching exists.
The initialization works as for \textsc{SRTI-Existence}, and the solution can be read from $\tau_r [\mathbf{h^r}]$, where $\mathbf{h^r}$ is the unique vector in $\{\oldminusone, \oldzero, \oldone\}^{\cut (r)}$.
We can treat the heavy children as for \textsc{SRTI-Existence} as their number is upper-bounded in a function of the tree-cut width.
For light children, we have to take into account that there may exist two children $c$ and $d$ with $\sig (c) = \sig (d)$ but the sizes of the corresponding matchings differ.
As we are looking for a maximum-cardinality matching, we have to incorporate the different sizes of the matchings.
Due to \Cref{lprune2}, however, we can assume that the difference between the sizes of two matchings for the same child is at most six (if the difference is larger, then we simply store $\square$ instead of the smaller matching in $\tau$).
Thus, we change the definition of classes to incorporate the sizes of the matchings.
Two children~$c, d$ are in the same class if $\sig (c) = \sig (d)$ for all $\bm{h}$, and additionally for any two vectors $\bm{h}, \bm{h}'$ we have $|\tau_c [\bm{h} ]| - |\tau_c [\bm{h'}]| = |\tau_d[\bm{h}]| - \tau_d[\bm{h'}]|$.
Since $|\tau_c [ \bm{h}]| - |\tau_c [\bm{h'}]| \le 6$ whenever $\tau_c [\bm{h}] \neq \square$ and $\tau_c [\bm{h'}] \neq \square$, the number of classes is still~$O(k^2)$.

Lastly, we need to solve a maximization variant of \textsc{Partial Embedding Extension}.
We now define this variant and afterwards show how to solve it in FPT-time.

\defProblemTask{\textsc{Maximum Partial Embedding Extension}}
{
A graph $G$, profiles $P$,
a nice tree-cut decomposition $(\mathcal{X}, \mathcal{T})$ of $G$, a node $t \in \mathcal{T}$,
a matching $\mathcal{N}$ from $X_t$ to the classes,
a partial embedding $\tilde{M}$ of $\mathcal{N}$,
a vector $\bm{h}\in \{\oldminusone, \oldzero, \oldone\}^{\cut (t)}$, and
DP tables $\tau_c$ for all children $c$ of $t$.
}
{
Find a matching $M \subseteq E(G_t)$ of maximum cardinality such that $\tilde{M} \cup M$ is an $\bm{h}$-embedding of $\mathcal{N}$, or decide that no such matching exists.
}

\begin{theorem}\label{thm:max-smti-induction}
  On bipartite graphs, \textsc{Maximum Partial Embedding Extension} can be solved in $2^{O(k^2)} n^{O(1)}$ time for the parameter $k := |X_t|$.
\end{theorem}

\begin{proof}
  Note that for each class $\mathcal{C}$, we have that $|N(V(\mathcal{C}))| \le 2$.
  \Cref{lprune2} then implies that there are at most six possible sizes of a maximum matching in $V(\mathcal{C} ) \cup N(V(\mathcal{C}))$.
  Since there are $O(k^2)$ different classes, we can guess the size of the matching inside $V(\mathcal C) \cup  N(V(\mathcal{C}))$ for each class $\mathcal{C}$ in $2^{O (k^2)}$ time.
  We denote the guess for a class $\mathcal{C}$ by~$\ell_{\mathcal{C}}$.

  Note that \Cref{lem:once-zero,lem:once-nonzero,lem:same-child} also hold for this refined definition of classes and maximum $\bm{h}$-embeddings.
  For a class $\mathcal{C}$, we can apply \Cref{lem:different-child} for maximum $\bm{h}$-embeddings~$M$ if we additionally guess~$\bm{h^{M, c_1}}$ and $\bm{h^{M, c_2}}$ for the two children $c_1, c_2\in \mathcal{C}$ with $\{x_1, y_1^{c_1}\}\in M$ and $\{x_2, y_2^{c_2}\} \in M$.
  Thus, we can construct a 2-SAT-formula which is satisfiable if and only if there exists partial embedding extension.
  To find a maximum partial embedding extension, we now add clauses ensuring that the matching inside $V(\mathcal{C} ) \cup N (V (\mathcal{C}))$ has size at least $\ell_{\mathcal{C}}$.

  Consider a class $\mathcal{C}$ with $ N(\mathcal{C}) =\{x_1, x_2\}$ which is either unmatched or a bad class $\mathcal{C}$ with~$N (\mathcal{C}) = \{x_1, x_2\}$, $\mathcal{N} ( x_1) = \mathcal{C}$, and $\mathcal{N} (x_2) \neq \mathcal{C}$
  Assuming that $x_1$ is matched to a vertex it ranks at position~$r_1$, we can compute a rank $s(r_1)$ such that if $x_1$ is matched at rank~$r_1$, then there exists a matching of size at least $\ell_{\mathcal{C}}$ inside $V(\mathcal{C}) \cup N(\mathcal{C})$ such that there is no blocking pair inside $V(\mathcal{C})\cup N(\mathcal{C})$ if and only if $x_2$ is matched at least as good as~$s(r_1)$.
  Clearly, we have that $s(r_1)$ is monotonically decreasing in $r_1$.
  Thus, we can model the constraint that the matching on $\mathcal{C}$ shall have size at least $\ell_{\mathcal{C}}$ by clauses $\neg x_1^{r-1} \lor \neg x_2^{s(r)}$ for each~$r\in \{1, 2, \dots, \max \rk_{x_1} \}$.

  For each good class~$\mathcal{C}$ with $N (\mathcal{C}) = \{x_1, x_2\}$, we have already guessed for $x_1$ or $x_2$ (without loss of generality~$x_1$) to which vertex it is matched.
  We therefore can compute a number $r$ such that the size bound can be satisfied if and only if $x_2$ is matched to rank~$r$ or better;
  this can be represented by clause~$\neg x_2^{r}$.

  For each bad class $\mathcal{C}$ with $N (\mathcal{C}) = \{x_1, x_2\}$, $\mathcal{N} ( x_1) = \mathcal{C} = \mathcal{N} (x_2)$, we can test for which children~$c\in \mathcal{C}$ we can find a matching in $V(\mathcal{C}) \cup N(\mathcal{C})$ which matches both $x_1$ and $x_2$ to~$Y_c$ and does not admit a blocking pair inside $V(\mathcal{C}) \cup N (\mathcal{C})$ of size at least $\ell_{\mathcal{C}}$, and delete all other children from~$\mathcal{C}$ in advance.
  
  Therefore, the resulting \textsc{2-SAT} formula evaluates to true if and only if there exists a partial embedding extension~$M$ with $|M \cap \bigl(V(\mathcal{C} ) \cup N (\mathcal{C})\bigr)| \ge \ell_{ \mathcal{C}}$.
  By guessing all possible values for $\ell_{\mathcal{C}}$ (recall that we only have to consider six different values for $\ell_{\mathcal{C}}$ for every class~$\mathcal{C}$), the theorem follows.
\end{proof}

Having shown that we can solve the induction step, it is easy to show that \textsc{Max-SMTI} parameterized by tree-cut width is FPT.

\begin{corollary}
  \textsc{Max-SMTI} can be solved in $2^{O(k^2)} n^{O(1)}$ time, where $k:= \tcw (G)$.
\end{corollary}

\begin{proof}
  Fix a nice tree-cut decomposition $(T, \mathcal{X})$ of $G$ of width $k$.
  Similar to \Cref{tfpttcw}, we compute a values $\tau_c [ \bm{h}]$ for every $c\in V(T)$ and $\bm{h} \in \{\oldminusone, \oldzero, \oldone\}^{\cut (c)}$, which shall contain a maximum-cardinality matching complying with~$\bm{h}$.
  For every leaf $t\in V(T)$, we compute~$\tau_t$ by brute-forcing over all matchings in $\clos (Y_t)$ in $2^{O(k\log k)}$ time.
  \Cref{thm:max-smti-induction} (together with \Cref{lem:once-nonzero,lem:once-zero,lem:same-child,lem:different-child} stated for maximum-cardinality $\bm{h}$-embeddings) shows that given $\tau_c$ for every child of a node~$t\in V(T)$, we can compute $\tau_t$ in $2^{O(k^2)} n^{O(1)} $ time.
  The matching stored in $\tau_r$ for the root~$r$ of $T$ then gives us a maximum-cardinality stable matching.
\end{proof}

\section{An FPT-algorithm with Respect to Feedback Edge Number}

In this section, we give an FPT-algorithm for \textsc{Max-SRTI} parameterized by the feedback edge number.
This will be achieved by reducing an \textsc{Max-SRTI}-instance to $3^{\fes (G)}$ \textsc{Max-SRTI}-instances of treewidth at most two.
    \footnote{Note that the
preliminary version of this paper~\cite{BredereckHKN19} contained a
typo, claiming erroneously a running time of $2^{\fes (G)} n^{O(1)}$
instead of $3^{\fes (G)} n^{O(1)}$.}
These instances can be solved in polynomial time, as \textsc{Max-SRTI} parameterized by treewidth is contained in XP~\cite{Adil2018}.

The idea of the algorithm is as follows.
We fix a feedback edge set $F$ and guess for each edge $e \in F$ its influence on the stable matching, that is, whether $e$ is contained in a maximum stable matching, and if not, which of the two endpoints of $e$ prevents $e$ from being blocking.
We encode these guesses by small gadgets which we add to $G - F$.
As $G - F$ is a tree and these gadgets increase the treewidth only by a constant, we can then use the XP-algorithm for \textsc{Max-SRTI} parameterized by treewidth due to Adil et al.~\cite{Adil2018}, thus solving the resulting instance.

  \begin{theorem}\label{tfptfes}
    \textsc{Max-SRTI} can be solved in $3^{\fes (G)} n^{O(1)}$ time, where $n:= |V(G)|$.
  \end{theorem}
  
\begin{proof}
  Let $F$ be a minimum feedback edge set.

  We enumerate all matchings $F'\subseteq F$.
  Let $X$ be the vertices matched by $F'$, i.e., $X=\{v\in V(G): \exists \{v, w\} \in F'\}$.
  Denote by $G_{F'}:= G-X$ the subgraph of $G$ induced by the vertices unmatched in $F'$.

  For each edge $e\in F\setminus F'$, we guess the endpoint of $e$ which prefers the matching over the other endpoint of $e$.
  More formally, we enumerate all functions $f : F\setminus F' \rightarrow V(G)$ with $f (e) \in e$ for all $e\in F \setminus F'$.
  The vertex $f(e)$ for an edge $e \in F\setminus F'$ then corresponds to the endpoint of $e$ which prefers the matching over the other endpoint of $e$.

  We try to extend $F'$ to a maximum stable matching.

  If there exists a blocking pair $\{v, w\}$ consisting of two matched vertices, then we directly discard $F'$, as it cannot be part of any solution: We do not change the matching of $v$ or $w$, and thus they will stay a blocking pair.

  We design a \textsc{Max-SRTI} instance with acceptability graph~$H$ with $\tw (H) \le 2$ which contains a stable matching of size $t$ if and only if $G$ contains a size-$(t-\ell +|F'|)$ stable matching containing $F'$ for some $\ell $ to be defined later.

  Let $G' := G_{F'} - F$.
  For a vertex $v\in V(G')$, let $X_v:=\{w\in \N_{G} (v)\cap X: \rk_w (v) < \rk_w (M(w)\}$ be the set of neighbors of~$v$ matched in $F'$ who prefer $v$ to their matched partner, and let $Y_v:= \{ w\in \N_{G} (v) : \{v, w\} \in F \setminus F'\land f(\{v, w\} ) = v)\}$ be the set of vertices which are neighbored to $v$ through an edge $e $ of $F\setminus F'$ with $f(e) = v$.
  Let $\alpha (v):= \min_{w\in X_v\cup Y_v} \rk_v (w)$ (where $\min \emptyset = \infty$).
  For each $v\in V(G')$, we delete all edges $\{v, w\}$ with $\rk_v (w) >\alpha (v)$, as no stable matching can contain them.

  Let $S:=\{v\in V(G'): \alpha (v) <\infty\}$ be the set of vertices which are part of a blocking pair if they remain unmatched. Let $\ell := |S|$. For each $v\in S$, we add a 3-cycle $\{v, v', v''\}$ and set $\rk_{v'} (v ) = 1$, $\rk_{v'} (v'') =2$, $\rk_{v''} (v ) = 2$, $\rk_{v''} (v' ) = 1$, $\rk_v (v' ) = \alpha (v) +2$, and $\rk_v (v'') = \alpha (v) + 1$. This adds $2\ell$ vertices, and ensures that $v$ is matched in every stable matching. We call the resulting graph $H_{F'}$.

  \begin{claim}
    If $G$ admits a stable matching $M$ with $F'\subseteq M$ of size $t$ and $\rk_v (M(v) ) \le \rk_v (w) $ for all $e = \{v, w\}\in F\setminus F'$ with $f(e) = v$, then $H_{F'}$ contains a stable matching of size $t+\ell -|F'|$.
  \end{claim}

  \begin{claimproof}
  Let $M^* := M\cap E(H_{F'})$.
  We claim that $M':=M^*\cup\bigcup_{v\in S} \{\{v', v''\}\}$ is a stable matching of size $t+\ell - |F'|$ in $H_{F'}$.

  First note that $|M'| = |M^*| + \ell = |M| - |F'| + \ell$.

  Assume that $M'$ contains a blocking pair $\{v, w\}$.
  As $M$ contains no blocking pair, each~$v\in V(G')$ is matched at rank at most $\alpha (v)$ (possibly unmatched if $\alpha (v) =\infty$), and thus at the same rank as in $M$.
  Thus, no blocking pair contains $v'$ or $v''$ for some $v\in V(G')$.
  As all other edges in $E(H_{F'}) $ are contained in $E(G)$, this implies that $H_{F'}$ contains no blocking pair.
  \end{claimproof}

  \begin{claim}
    If $H_{F'}$ contains a stable matching of size $t$, then $G$ contains a stable matching of size $t-\ell +|F'|$ with $F'\subseteq M$.
  \end{claim}

  \begin{claimproof}
  Let $M'$ be a stable matching of size $t$.
  Any vertex in $v\in V(G')$ with $\alpha (v) <\infty$ is matched to another vertex in $V(G')$, as otherwise two of the three vertices $v$, $v'$, and $v''$ form a blocking pair.
  We claim that $M:=F'\cup (M'\cap E(G'))$ is a stable matching of size $t-\ell + |F'|$.

  To see this, note that $M'\cap E(G')$ contains $t-\ell$ edges, non of which is contained in $F$ nor in $F'$.
  Thus, $|M| = t-\ell +|F'|$.

  Assume that there is a blocking pair $\{v, w\}$ for $M$.
  Then, we have $\{v, w\}\notin E(H_{F'})$, implying that $\{v, w\} \in F\setminus F'$ or $v\in X$ or $w\in X$.
  Note that $v\in X$ and $w\in X$ is not possible, as we discarded all sets $F'$ which imply such blocking pairs.
  If $\{v, w\} \in F \setminus F'$, then we assume that $f(\{v, w\} ) = v$ (the case $f(\{v, w\}) = w$ is symmetric).
  As we added vertices $v'$ and $v''$ to $H_{F'}$, vertex $v$ has to be matched in $M$ at rank at most $\alpha_v \le \rk_v (w)$, implying that $v$ does not prefer $w$ to $M(v)$ and, therefore, $\{v, w\} $ is not a blocking pair.
  Hence, we have either $v\in X$ or $w\in X$.
  By symmetry, we can assume that $w\in X$.
  Then we have $\alpha (v) < \infty$ and therefore added vertices $v'$ and $v''$.
  Thus, $v$ is matched in $M'$, and due to the deletion of edges, it is matched at rank at most $\alpha (v)$, implying that $v$ does not prefer $w$ to $M(v)$.
  Thus, $\{v, w\}$ is not a blocking pair.
  \end{claimproof}

  Graph $H_{F'}$ arises from the forest $G'$ by adding for each $v \in G' $ with $\alpha (v) < \infty$ vertices~$v'$ and~$v''$ together with edges $\{v, v'\}$, $\{v, v''\}$, and $\{v', v''\}$.
  Any forest admits a tree decomposition of width one, where for each vertex $v$ there exists a bag $B_v$ containing only $v$.
  Starting with such a tree-decomposition for $G'$, duplicating $B_v$ and adding a bag containing $\{v, v', v''\}$ in between for every $v\in V(G')$ with $\alpha (v) < \infty$ yields a tree decomposition of width two.
  Thus, $H_{F'}$ has treewidth at most two.
  Consequently, finding maximum stable matching in~$H_{F'}$ or deciding that $H_{F'}$ does not admit a stable matching can be done in polynomial time (as \textsc{Max-SRTI} is in XP with respect to treewidth \cite{Adil2018}), proving the theorem.
\end{proof}

\section{Conclusion}
Taking the viewpoint of parameterized graph algorithmics, we investigated
the line between fixed-parameter tractability and W[1]-hardness
of \textsc{Stable Roommates with Ties and Incomplete Lists}.
Studying parameterizations mostly relating to the `tree-likeness'
of the underlying acceptability graph, we arrived at a fairly complete
picture (refer to \cref{fig:SRTIResultsOverview})
of the corresponding parameterized complexity landscape.
There is a number of future research directionss stimulated by our work.
First, we did not touch on questions about polynomial kernelizability of the
fixed-parameter tractable cases. Indeed, for the parameter feedback edge
number we believe that a polynomial kernel should be possible.
Another issue is how tight the running time for
our fixed-parameter algorithm for
the parameter tree-cut width~$k$ is; more specifically, can we show that
our exponential factor $k^{O(k)}$ is basically optimal or can it be improved
to say~$2^{O(k)}$?
Clearly, there is still a lot of room to study
\textsc{Stable Roommates with Ties and Incomplete Lists} through the
lens of further graph parameters.
On a general note, we emphasize that so far our investigations are
on the purely theoretic and classification side; practical algorithmic
considerations are left open for future research.

\bibliography{bib}

\end{document}

%% file: fancyProblem.tex
\usepackage{framed}
\usepackage{tabularx}

\newlength{\RoundedBoxWidth}
\newsavebox{\GrayRoundedBox}
\newenvironment{GrayBox}[1]%
   {\setlength{\RoundedBoxWidth}{.93\columnwidth}
    \def\boxheading{#1}
    \begin{lrbox}{\GrayRoundedBox}
       \begin{minipage}{\RoundedBoxWidth}}%
   {   \end{minipage}
    \end{lrbox}
    \begin{center}
    \begin{tikzpicture}%
       \node(Text)[draw=black!20,fill=white,rounded corners,inner sep=2ex,text width=\RoundedBoxWidth]
             {\usebox{\GrayRoundedBox}};
        \coordinate(x) at (current bounding box.north west);
        \node [draw=white,rectangle,inner sep=3pt,anchor=north west,fill=white]
        at ($(x)+(6pt,.75em)$) {\boxheading};
    \end{tikzpicture}
    \end{center}}

\newenvironment{defproblemx}[1]{\noindent\ignorespaces%
                                \FrameSep=6pt%
                                \parindent=0pt%
                \begin{GrayBox}{#1}%
                \begin{tabular*}{\columnwidth}{!{\extracolsep{\fill}}@{\hspace{.1em}} >{\itshape} p{1.5cm} p{0.86\columnwidth} @{}}%
            }{
                \end{tabular*}%
                \end{GrayBox}%
                \ignorespacesafterend
            }

\newcommand{\defProblemTask}[3]{%
  \begin{defproblemx}{#1}
    Input: & #2 \\
    Task: & #3
  \end{defproblemx}
}